\documentclass[11pt]{amsart}

\usepackage{amssymb,latexsym,amsmath,extarrows, mathrsfs, amsthm,color,amsrefs}
\usepackage{graphicx}
\usepackage{MnSymbol}

\usepackage[margin=1in, centering]{geometry}
\usepackage[bookmarksnumbered, colorlinks, plainpages]{hyperref}
\usepackage{hyperref} 
\hypersetup{
    colorlinks=true,       
    linkcolor=blue,          
    citecolor=magenta,        
    filecolor=magenta,      
    urlcolor=cyan           
}

\usepackage{mathtools}
\mathtoolsset{showonlyrefs,showmanualtags}

\newtheorem{theorem}{Theorem}[section]

\newtheorem{lemma}[theorem]{Lemma}
\newtheorem{proposition}[theorem]{Proposition}
\newtheorem{remark}[theorem]{Remark}

\newtheorem{definition}{Definition}[section]
\newtheorem{corollary}[theorem]{Corollary}

\newtheorem{fact}[theorem]{Fact}

\newcommand{\wt}{\widetilde}

\newcommand{\Z}{\mathbb Z}
\newcommand{\R}{\mathbb R}
\newcommand{\C}{\mathbb C}
\newcommand{\T}{\mathbb T}
\newcommand{\N}{\mathbb N}

\def\mes{\mathrm{mes}}
\definecolor{deepgreen}{cmyk}{1,0,1,0.5}

\def\calB{\mathcal{B}}
\def\calC{\mathcal{C}}
\def\calH{\mathcal{H}}
\def\tor{\mathbb{T}}

\def\Mat{\mathrm{Mat}}

\def\dist{\mathrm{dist}}
\def\les{\;\mathrm{\lesssim}\;}
\def\calS{\mathcal{S}}
\def\disk{\mathbb{D}}
\def\calE{\mathcal{E}}

\title[Localization for quasi-periodic block matrices]{Non-perturbative localization for quasi-periodic Jacobi block matrices}

\author[R.\ Han]{Rui Han}
\address{Department of Mathematics \\ Louisiana State University  \\  Baton Rouge, LA 70803, USA}
\email{rhan@lsu.edu}

\author[W.\ Schlag]{Wilhelm Schlag}
\address{Department of Mathematics \\ Yale University \\ New Haven, CT 06511, USA}
\email{wilhelm.schlag@yale.edu}

\thanks{
R.\ Han is partially supported by NSF DMS-2143369. 
W.\ Schlag is partially supported by NSF grant DMS-2054841.
}

\begin{document}

\begin{abstract}
We prove non-perturbative Anderson localization for quasi-periodic Jacobi block matrix operators assuming non-vanishing of all Lyapunov exponents. The base dynamics on tori $\tor^b$ is assumed to be a Diophantine rotation. Results on arithmetic localization are obtained for $b=1$, and applications to the skew shift, stacked graphene, XY spin chains, and coupled Harper models are discussed.  
\end{abstract}

\maketitle
\tableofcontents

\section{Introduction}

Let 
\begin{equation}\label{eq:BVsys}
   (H_\theta \Phi)_n   = B_{n+1}(\theta)\Phi_{n+1}+ B_n^{(*)}(\theta) \Phi_{n-1}+ V_n (\theta)\Phi_n
\end{equation}
where $F_n(\theta):=F(\theta+n\omega)$ for any $d\times d$-matrix valued function. We set $B^{(*)}(\theta)=(B(\theta))^*$ for $\theta\in \T$, and require it to be the analytic extension of $(B(\theta))^*$ off of the real torus. Here $\theta,\omega\in\tor^b$ and we assume that $\omega$ is Diophantine, i.e.,
\begin{align}
    \omega\in \mathrm{DC}:=&\bigcup_{a>0, A>b} \mathrm{DC}_{a,A}, \text{ where}\\
    \mathrm{DC}_{a,A}=&\left\{\omega\in \T^b:\, \|k\cdot \omega\|_{\T^b}\geq \frac{a}{|k|^A}\, \text{ for all } k\in \Z^d\setminus \{0\}\right\}.
\end{align}  
We further assume that $B,V\in C^{\omega}(\tor^b_{\eta}, \Mat(d,\C))$ are analytic, where $$\T^b_{\eta}:=\{\theta+i\varepsilon:\, \theta\in \T^b,\, \varepsilon\in \R^d,  \text{ and } |\varepsilon|\leq \eta\}$$with some positive $\eta>0$. 
We assume throughout the paper that $V$ is Hermitian, and that~$B$ is invertible ($\det B(\theta)\neq 0$ for any $\theta\in \T^b_{\eta}$). 
The difference equation $H_{\theta} \Phi=E\Phi$ is equivalent to the cocycle 
\begin{align}\label{def:M_E}
  \calC &\::  (\theta,\Psi)\in \tor^b\times \C^{2d}\mapsto (\theta+\omega, M_E(\theta)\Psi), \notag  \\
  \quad M_{E}(\theta) &= \left [ \begin{matrix}
   ( E-V(\theta))B(\theta)^{-1}  & -B^{(*)}(\theta) \\
  B(\theta)^{-1} & 0 
\end{matrix}\right]  
\end{align}
in the sense that for $n\ge1$, \[ \calC^n (\theta,\Psi)=(\theta+n\omega, M_{n,E}(\theta)\Psi), \quad M_{n,E}(\theta)=\prod_{j=n-1}^0 M_E(\theta+j\omega), \quad 
\Psi_n:=\binom{B_{n}\Phi_n}{\Phi_{n-1}}\] satisfies $\Psi_n = M_{n,E} (\theta) \Psi_0$.  Since $M_E(\theta)$ is (complex) symplectic, see \eqref{def:symp}, for $\theta\in \T^b$, the Lyapunov exponents $\{L_j(\omega,M_E)\}_{j=1}^{2d}$, see definition in \eqref{def:LE}, satisfy $L_{2d-j}=-L_j$ for $1\le j\le d$. 
In analogy with \cite{BG} we establish localization for $H_0$ under a nonvanishing condition on the Lyapunov exponents.

\begin{theorem}
    \label{thm:main}
    Assume that $L_d(\omega,M_E)\ge \gamma>0$ for all $E,\omega$. Then $H_0$ exhibits Anderson localization for almost every~$\omega$.
\end{theorem}

For a perturbative version of this result see Bourgain-Jitomirskaya~\cite{BJ} and Klein~\cite{Kl}.
Theorem~\ref{thm:main} has wide applications since operators in quantum mechanics often appear in block form, in particular when dealing with stacked materials or when not restricted to nearest neighbor hopping.   
Section~\ref{sec:application} includes some examples on stacked and twisted bilayer graphene models.

Our next theorem,  which can be seen as a generalization of \cite{HS2}*{Theorem 1.6}, concerns arithmetic Anderson localization for long-range scalar-valued Jacobi matrices with potential defined on the one-dimensional torus $\T$. 
For the one-dimensional torus, there is another quantity, Avila's quantized acceleration of the cocycle, which we denote by $\kappa^d(\omega,M_E)$, that plays a crucial role in determining the spectral behavior of the corresponding operator.
The acceleration was first introduced by Avila in his global theory paper \cite{Global} for $\mathrm{SL}(2,\R)$ cocycles, and extended to higher dimensional cocycles in \cite{AJS}.
A characterization of the acceleration, for scalar valued Schr\"odinger operator $d=1$ and for Diophantine~$\omega$, through the zeros of Dirichlet determinants was proved in~\cite{HS2}. 

In this paper, we give a characterization of the acceleration $\kappa^d(\omega,M_E)$ associated to Jacobi block matrices with $d\geq 2$ in terms of the zero count of the finite volume determinants with periodic boundary condition. See Theorem \ref{thm:Riesz_un}. 
We also give applications of such characterization in arithmetic Anderson localization for various models.
All of our arithmetic Anderson localization results concern the one-dimensional torus $\T$.
For $\theta,\omega\in \T$, let 
\begin{align}\label{def:tH}
    (\tilde{H}_{\theta}\phi)_n=\sum_{1\leq |k|\leq d} v_k\phi_{n-k}+g(\theta+nd^{-1}\omega) \phi_n,
\end{align}
where throughout the paper we assume that $g\in C_{\eta}^{\omega}(\T,\R)$ is real-valued and non-constant and
\begin{align}\label{eq:basic_assume_v}
\overline{v_k}=v_{-k} \text{ for any } 1\leq k\leq d, \text{ and }
v_d\neq 0.
\end{align}
Note that we use $d^{-1}\omega$ as the frequency in $\tilde{H}_{\theta}$ to better fit into the framework of~\eqref{eq:BVsys}.
Let $A_E$ be the corresponding $1$-step transfer matrix (see \eqref{def:A_E}) and $\kappa^d(d^{-1}\omega,A_E)$ be its acceleration, which is an integer as long as $L_d(d^{-1}\omega,A_E)>0$ (see \cite{AJS}*{Lemma 6.4}).

\begin{theorem}\label{thm:acceleration=1}
Fix any $\omega\in \mathrm{DC}$.
Suppose $g(\theta)=g(-\theta)$.
Then for any 
\begin{align}\label{def:Theta_d}
\theta\in \Theta_d:=\bigcup_{a'>0, t>1} \Big\{\theta: \|2\theta-n\omega\|_{\T}\geq \frac{a'}{(1+|n|)^{t}}, \text{\ for any\ \  } n\in d^{-1}\Z\Big\},
\end{align}
$\tilde{H}_{\theta}$ is Anderson localized on $\sigma(H_{\theta})\cap \{E: L_d(d^{-1}\omega,A_E)>0, \text{\  and\  } \kappa^d(d^{-1}\omega,A_E)=1\}$.
\end{theorem}

As a direct corollary, we obtain the following result that strengthens \cite{HS3}*{Theorem 1.7} for $\omega\in \mathrm{DC}$.
In fact, let 
\begin{align}
    (\tilde{H}_{\theta,v,\nu}^{\cos,g}\phi)_n=\sum_{1\leq |k|\leq d}v_k \phi_{n-k}+(2 \cos(2\pi (\theta+nd^{-1}\omega))+\nu \cdot g(\theta+nd^{-1}\omega))\phi_n
\end{align}
be a perturbation of the following operator considered in \cite{HS3}:
\begin{align}\label{eq:operator_HS3}
    (\tilde{H}^{\cos}_{\theta,v}\phi)_n=\sum_{1\leq |k|\leq d}v_k\phi_{n-k}+2 \cos(2\pi (\theta+nd^{-1}\omega))\phi_n.
\end{align} 

\begin{corollary}
Fix $\omega\in \mathrm{DC}$. Suppose $I$ is a closed interval such that \[\inf_{E\in I} L_d(d^{-1}\omega, A^{\cos}_E)\geq \gamma>0,\] where $A^{\cos}_E$ is the one-step transfer matrix corresponding to the unperturbed operator in \eqref{eq:operator_HS3}.
Then there exists $\nu_0=\nu_0(\omega,\gamma,v,g)>0$ such that for any $|\nu|\leq \nu_0$, $\tilde{H}^{\cos,g}_{\theta,v,\nu}$ is Anderson localized in $I\cap \sigma(\tilde{H}^{\cos,g}_{\theta,v,\nu})$ for any $\theta\in \Theta_d$.
\end{corollary}
The proof combines the techniques developed in \cite{HS2,HS3}.
The special non-perturbed case $\nu=0$ already leads to a proof of a quantitative version of Avila's almost reducible conjecture for Schr\"odinger cocycles with trignometric potentials \cite{HS3}. See also Avila's proof \cite{Av2} for the general analytic case.

The next theorem also concerns arithmetic Anderson localization, for operators in \eqref{eq:BVsys} satisfying certain symmetries. 
\begin{theorem}\label{thm:acceleration=d_2}
Let $H_{\theta}$ be as in \eqref{eq:BVsys} on the $1$-d torus $\T$, and with $B(\theta)\equiv B$ being constant.
Suppose there exists an orthonormal matrix $J\in \mathrm{Mat}(d,\C)$, such that 
\begin{align}\label{assume:t2_1}
JV(\theta)J^{-1}=V^T(-\theta), \text{ and } JBJ^{-1}=B^T.
\end{align}
Let $f_{E,n}(\theta)$ be the finite volume Dirichlet determinant with periodic boundary conditions, see \eqref{def:fn}.
Assume further that $f_{E,n}(\theta)$ is $d^{-1}$-periodic, namely,
\begin{align}\label{assume:t2_2}
f_{E,n}(\theta+d^{-1})=f_{E,n}(\theta).
\end{align}
Then for $\omega\in \mathrm{DC}$ and 
\begin{align}\label{def:theta}
    \theta\in \Theta:=\bigcup_{a'>0, t>1}\left\{\theta: \|2\theta-n\omega\|_{\T}\geq \frac{a'}{(1+|n|)^t}, \text{ for any } n\in \Z\right\}, 
\end{align}
$H_{\theta}$ is Anderson localized on $\sigma(H_{\theta})\cap \{E:\, \kappa^d(\omega,M_E)\leq 2d-1, \text{ and } L_d(\omega,M_E)>0\}$.
\end{theorem}

\begin{remark}
    Without the assumption \eqref{assume:t2_2}, one can show that $H_{\theta}$ is Anderson localized on $\sigma(H_{\theta})\cap \{E:\, \kappa^d(\omega,M_E)=1, \text{ and } L_d(\omega,M_E)>0\}$. 
    Note for $E\in \sigma(H_{\theta})$, by \cite{AJS}*{Lemma 6.4}, $0\neq \kappa^d(\omega,M_E)\in \Z$. By the quantization and upper-semicontinuity of $\kappa^d$, and the arguments in~\cite{B}*{p.~80--83}, the theorem applies to a set of positive measure of energies or to no energy at all. 
\end{remark}

The first application of Theorem \ref{thm:acceleration=d_2} concerns a model arising from the study of the anisotropic XY spin chain. We refer the reader to  Sec.~\ref{sec:spin} for the connection between the following model and spin chains, and the previous results. 
\begin{theorem}\label{thm:XY}
    Let 
\begin{align}\label{def:HXY}
(\widetilde{H}^{XY}_{\omega,\theta,\rho,v}\phi)_n=B \phi_{n+1}+V(\theta+n\omega) \phi_n+B^* \phi_{n-1},
\end{align}
where $\theta,\omega\in \T$, $\rho\in \R$ and
\begin{align}
    B=\left(\begin{matrix}1 &\rho\\ -\rho &-1\end{matrix}\right), \text{ and\ \  } V(\theta)=\left(\begin{matrix} v(\theta) &0 \\ 0 &-v(\theta)\end{matrix}\right).
\end{align}
Here $v$ is an even analytic function satisfying $v(\theta+\frac{1}{2})=-v(\theta)$.
Then for any $\omega\in \mathrm{DC}$ and $\theta\in \Theta$ as in \eqref{def:theta}, $\widetilde{H}_{\omega,\theta,\rho,v}^{XY}$ is Anderson localized in $\sigma(\widetilde{H}^{XY}_{\omega,\theta,\rho,v})\cap\{E: \kappa^2(\omega,M_E)\leq 3, \text{ and } L_2(\omega,M_E)>0\}$.
In particular, with $v(\theta)=2\lambda\cos(2\pi\theta)$, $\widetilde{H}^{XY}_{\omega,\theta,\rho,2\lambda\cos}$ is Anderson localized in $\{E: L_2(\omega,M_E)>0\}$.
\end{theorem}

The next application of Theorem~\ref{thm:acceleration=d_2} concerns the skew-shift model with rational frequencies.

\begin{theorem}\label{thm:skew_shift}
    For any reduced rational $p/q$, $q\geq 3$ \footnote{If $p/q=1/2$, the operator reduces to the almost Mathieu operator.}, there exists $\lambda_0=\lambda_0(p/q)>0$ such that the skew-shift operator 
    \begin{align}\label{def:H_sk}
        (H^{sk}_{\lambda,x,y,p/q}\phi)_n=\phi_{n+1}+\phi_{n-1}+2\lambda \cos(2\pi(x+ny+n(n-1)p/q))\phi_n,
    \end{align}
    has zero Lyapunov exponent for every $y\in \mathrm{DC}$, see \eqref{def:LE_sk_y}, on the spectrum for $0<|\lambda|< \lambda_0$. The Lyapunov exponent here is averaged in $x\in\tor$. 
\end{theorem}
\begin{remark}
    In \cite{B1}*{p.~66}, Bourgain suggested that one might be able to combine the large deviation estimates with numerical computations to establish positive Lyapunov exponent for the operator in \eqref{def:H_sk} for small $\lambda$ for the Lyapunov exponent averaged in both $x,y\in \T$.
    That may be true due to the fact that $\sigma(H^{sk}_{\lambda,x,y,p/q})$ depends sensitively on $y$, and hence any fixed $E\neq 0$ is not in the spectrum of $H^{sk}_{\lambda,x,y,p/q}$ for some $y$.
\end{remark}
The proof of Theorem \ref{thm:skew_shift} is built on establishing Anderson localization for the dual model, obtained as a corollary of Theorem \ref{thm:acceleration=d_2}. We actually prove a stronger almost localization result, see Theorem \ref{thm:sk_dual_almost_AL}. Combining almost localization with the quantitative duality techniques from \cite{AJ}, one should be able to prove quantitative almost reducibility of $H^{sk}_{\lambda,x,y,p/q}$ and conclude purely absolutely continuous spectrum for this operator through the perturbative theory of Eliasson \cite{E}. We leave this and other applications of quantitative almost reducibility of this operator for future work.

In Section~\ref{sec:coupledHarper}, we establish some properties of the coupled Harper operators ($\omega$ Diophantine)
\begin{equation}
     \label{eq:Harpx2}
     \begin{split}
          \phi_{n+1}+\phi_{n-1}+\epsilon\psi_n + 2\lambda_1 \cos(2\pi(x+n\omega)) \phi_n &= E\phi_n \\
     \psi_{n+1}+\psi_{n-1}+\epsilon\phi_n + 2\lambda_2 \cos(2\pi(x+n\omega)) \psi_n &= E\psi_n
     \end{split}
\end{equation}
where $\epsilon\in\R$ is small,  and $\lambda_2\ge \lambda_1>0$. This falls under the scope of~\eqref{eq:BVsys}. Amongst other results, we show that for $\lambda_2\gg 1$ and $0<\epsilon,\lambda_1\ll 1$,  both Anderson localization and a.c.\ states occur on sets of positive measure of energies~$E$ in the spectrum of this system.  Moreover, it follows from the two-sided Oseledets theorem, see~\cite{V}*{Theorem 4.2}, that the states associated with the a.c.~spectrum are not of hybrid type. I.e., they are truly extended states and cannot exhibit exponential decay to one side. 

\medskip

The rest of the paper is organized as follows: some preliminaries are presented in Section~\ref{sec:Pre}, and some technical lemmas are in Section~\ref{sec:lemma}. The proofs of these lemmas are in Section~\ref{sec:Zaehler} for the numerator of the Green's function, and Section~\ref{sec:nenner} for the lower bound of the denominator, respectively. 
The main theorems are proved in Sections~\ref{sec:main} (Theorem \ref{thm:main}) and Section~\ref{sec:AL_arithmetic} (Theorems \ref{thm:acceleration=1}, \ref{thm:acceleration=d_2}). 
The applications are discussed in Section~\ref{sec:XY_sk} (XY-spin chain and skew-shift), Section~\ref{sec:application} (stacked graphene models), and Section~\ref{sec:coupledHarper} (coupled Harper's model).

\section{Lyapunov exponents, large deviations, and the Green's function}\label{sec:Pre}

Throughout, we adhere to the following notations. 
For a function $g$ on $\T^b$, we denote its $L^p(\T^b)$ norm by $\|g\|_{\T^b, p}$, and we write $\langle g\rangle:=\int_{\T^b } g(\theta)\, \mathrm{d}\theta$ for averages.
For $x\in \R^d$, let $\|x\|_{\tor^b}:=\mathrm{dist}(x,\Z^d)$ be the distance to the nearest integer vector. 
Let $\mathcal{C}_1:=\{z\in \C:\, |z|=1\}$ be the unit circle, and $\mathcal{A}_R:=\{z\in \C:\, 1/R\leq |z|\leq R\}$.
For a set $U\subset\R^b$, let $\mathrm{mes}(U)$ be its Lebesgue measure.
For $\ell\in \Z$ and $q\in \N$, let $[\ell]_q\in \{0,...,q-1\}$ be such that $[\ell]_q\equiv \ell\, (\mathrm{mod}\,  q)$.
Throughout the paper, we restrict to energy $E\in \R$.

\subsection{Transfer matrices}
As we mentioned in the introduction, $M_E$ as in \eqref{def:M_E} is the transfer matrix associated to the block-valued operator $H_{\theta}$ in \eqref{eq:BVsys}.
The long-range scalar-valued operator $\tilde{H}_{\theta}$ as in \eqref{def:tH} can, on the one hand, be viewed as a $d\times d$ block-valued operator, where the corresponding blocks are
\begin{align}\label{def:tH_BV}
        &V(\theta)=
        \left(\begin{matrix}
g(\theta+(d-1)d^{-1}\omega) &v_1 &\cdots &v_{d-2} &v_{d-1}\\
           \overline{v_1} &g(\theta+(d-2)d^{-1}\omega)&\ddots &\ddots &v_{d-2}\\
\vdots &\ddots &\ddots &\ddots &\vdots          \\
\overline{v_{d-2}} &\ddots &\ddots &g(\theta+d^{-1}\omega) &v_1\\
\overline{v_{d-1}}         &\overline{v}_{d-2} &\cdots &\overline{v_1} &g(\theta)
        \end{matrix}\right), \text{ and } \\
                &\qquad\qquad B=
                \left(\begin{matrix}
\overline{v_d} &\overline{v_{d-1}} &\cdots &\overline{v_2} &\overline{v_1}\\
             &\overline{v_d} &\overline{v_{d-1}} &\ddots &\overline{v_2}\\
& &\ddots &\ddots & \vdots         \\
& & &\overline{v_d} &\overline{v_{d-1}}\\
& & & &\overline{v_d}
        \end{matrix}\right)
    \end{align}

On the other hand, as a scalar-valued operator, the eigenvalue equation $\tilde{H}_{\theta}\phi=E\phi$ can also be rewritten as:
\begin{align}\label{eq:tranfer_hH_varepsilon_1step}
\left(\begin{matrix}
\phi_{n+d}\\
\vdots\\
\phi_{n+1}\\
\phi_n\\
\vdots\\
\phi_{n-d+1}
\end{matrix}\right)
=A_{E}(\theta+nd^{-1}\omega)
\left(\begin{matrix}
\phi_{n+d-1}\\
\vdots\\
\phi_{n}\\
\phi_{n-1}\\
\vdots\\
\phi_{n-d}\end{matrix}\right),
\end{align}
where 
\begin{equation}\label{def:A_E}
A_{E}(\theta):=
\frac{1}{v_{-d}}
\left(\begin{matrix}
-v_{1-d} &\cdots &-v_{-1} &E-g(\theta) &| &-v_1 &\cdots &-v_{d-1} &-v_d\\
v_{-d} & & & &|\\
&\ddots & & &| & \\
& &v_{-d}  & &| &\\
\hline
& & &v_{-d} &| &\\
& & & &| &v_{-d} \\
& & & &| & &\ddots\\
& & & &| & & &v_{-d}
\end{matrix}\right).
\end{equation}
It is easy to verify that for any $\theta\in \T^b_{\eta}$,
\begin{align}\label{eq:matrix_M=A}
    M_E(\theta)=\mathrm{diag}(B,I_d)\cdot \prod_{j=d-1}^0 A_E(\theta+jd^{-1}\omega)\cdot \mathrm{diag}(B^{-1},I_d),
\end{align}

A complex matrix $M\in \mathrm{Mat}(2d,\C)$ is symplectic if 
\begin{align}\label{def:symp}
    M^*\Omega M=\Omega,
\end{align}
where 
\begin{align}
    \Omega=\left(\begin{matrix} 0 & I_d\\ -I_d &0\end{matrix}\right).
\end{align}
One can easily verify that for $M_E$ as in \eqref{def:M_E}, and $E\in \R$, $\theta\in \T^b$,
\begin{align}
    (M_E(\theta))^* \Omega M_E(\theta)=\Omega.
\end{align}
However for $\theta\in \T^b\setminus \T$, $M_E(\theta)$ is in general not symplectic.

\subsection{Lyapunov exponents}
Let $(\omega, A)\in (\T^b, C^{\omega}(\T, \mathrm{Mat}(k,\C)))$. 
Let 
\begin{align}
A_n(\omega,\theta)=A(\theta+(n-1)\omega)\cdots A(\theta).
\end{align}
Let the finite-scale and infinite-scale Lyapunov exponents be defined as 
\begin{align}\label{def:LE}
L_{j,(n)}(\omega, A):=\frac{1}{n}\int_{\T^b} \log \sigma_j(A_n(\omega,\theta))\, \mathrm{d}\theta, \text{ for } 1\leq j\leq k,
\end{align}
where $\sigma_j(A)$ is the $j$-th singular value of $A$, and the $j$-th Lyapunov exponent
\begin{align}
L_j(\omega, A)=\lim_{n\to\infty}L_{j,(n)}(\omega, A).
\end{align}
It is easy to see that for $1\leq j\leq k$,
\begin{align}
L^j_{(n)}(\omega, A):=\sum_{\ell=1}^j L_{\ell,(n)}(\omega, A)=\frac{1}{n}\int_{\T}\log \|\textstyle{\bigwedge^j} A_n(\omega,\theta)\|\, \mathrm{d}\theta,
\end{align}
where $\textstyle{\bigwedge^j} A$ is the $j$-th exterior power of $A$. Similarly $L^j(\omega,A)=\sum_{\ell=1}^j L_{\ell}(\omega, A)$.

We also denote the phase complexified Lyapunov exponents $L_{j,(n)}(\omega,A(\cdot+i\varepsilon))=:L_{j,(n),\varepsilon}(\omega,A)$, $L_j(\omega,A(\cdot+i\varepsilon))=:L_{j,\varepsilon}(\omega,A)$, $L^j_{(n)}(\omega,A_E(\cdot+i\varepsilon))=:L^j_{(n),\varepsilon}(\omega,A_E)$, and $L^j(\omega,A_E(\cdot+i\varepsilon))=:L^j_{\varepsilon}(\omega,A_E)$, respectively. Here $\varepsilon\in \R^b$.
Since $M_E(\theta)$ is symplectic for $\theta\in \T^b$,   for each $1\leq j\leq d$  
\begin{align}
    L_{j,\varepsilon=0}(\omega,M_E)=-L_{2d+1-j,\varepsilon=0}(\omega,M_E).
\end{align}
However the above is in general not true if $\varepsilon\neq 0$.

\subsection{Avila's acceleration}
Within this subsection, we restrict to the 1-d torus $\T$.
Let $(\omega, A)\in (\T, C^{\omega}(\T, \mathrm{SL}(2, \R)))$.
The (top) Lyapunov exponent $L^1_{\varepsilon}(\omega,A)=L_{1,\varepsilon}(\omega, A)$ is a convex and even function in $\varepsilon$.
Avila defined the acceleration to be the right-derivative as follows:
\begin{align}
\kappa^1_{\varepsilon}(\omega, A):=\lim_{\varepsilon'\to 0^+} \frac{L^1_{\varepsilon+\varepsilon'}(\omega, A)-L^1_{\varepsilon}(\omega, A)}{2\pi \varepsilon'}.
\end{align}
As a cornerstone of his global theory \cite{Global}, he showed that for $A\in \mathrm{SL}(2,\R)$ and irrational $\alpha$, $\kappa_{\varepsilon}^1(\omega, A)\in \Z$ is always quantized. 

The concept of acceleration was further extended to  $(\omega,A)\in (\T,C^{\omega}(\T, \mathrm{Mat}(k,\C)))$ in \cite{AJS}, where for $1\leq j\leq k$,
\begin{align}
    \kappa^j_{\varepsilon}(\omega,A):=\lim_{\varepsilon'\to 0^+} \frac{L^j_{\varepsilon+\varepsilon'}(\omega, A)-L^j_{\varepsilon}(\omega, A)}{2\pi \varepsilon'}.
\end{align}

By \eqref{eq:matrix_M=A}, for $|\varepsilon|\leq \eta$ and each $1\leq j\leq 2d$,
\begin{align}\label{eq:M=A}
    L_{j,\varepsilon}(\omega,M_E)=d\cdot L_{j,\varepsilon}(d^{-1}\omega,A_E),
\end{align}
and
\begin{align}\label{eq:kappa_M=dA}
\kappa^j_{\varepsilon}(\omega,M_E)=d\cdot \kappa^j_{\varepsilon}(d^{-1}\omega,A_E).
\end{align}

Recall that $B,V$ are analytic functions on $\T_{\eta}$ for some $\eta>0$.  We may shrink $\eta$ when necessary such that 
\begin{align}\label{eq:L_linear}
L^d_{\varepsilon}(\omega,M_E)=L^d_{\varepsilon=0}(\omega,M_E)+2\pi \kappa^d_{\varepsilon=0}(\omega,M_E)|\varepsilon|
\end{align}
holds for any $|\varepsilon|\leq \eta$. 
For the rest of the paper, when $\varepsilon=0$, we shall omit $\varepsilon$ from various notations of Lyapunov exponents and accelerations.
On some occasions, we shall also omit $\omega$ and $M_E$ in $L^j(\omega,M_E)$, $L_j(\omega,M_E)$ and $\kappa^d(\omega,M_E)$.

\subsection{Large deviation estimates and the Avalanche Principle}

We will require the following standard tools. 
Now $b\ge1$ again.
Note that we do not distinguish the various $\delta$'s in the following Lemmas~ \ref{lem:upperbd}, \ref{lem:LDTsig} and \ref{lem:Ln-L}.
\begin{lemma}\label{lem:upperbd}
For $\omega\in \mathrm{DC}$,
there exists $\delta>0$ so that for each $1\leq j\leq d$ and all large $n$, one has 
\begin{align}
\frac{1}{n} \log \big \|\textstyle{\bigwedge^j} M_{n,E}(\theta+i\varepsilon)\big\|\leq {L}^{j}_{(n),\varepsilon}(\omega,M_E)+ n^{-\delta},
\end{align}
uniformly in $\theta\in \tor^b$ and $|\varepsilon|\leq \eta$, $\varepsilon\in\R^b$. 
\end{lemma}

The following large deviation estimates play a crucial role in our argument. 
These results were first established in Lemma~1.1 of~\cite{BG} by Bourgain and Goldstein, and further developed  by Goldstein and Schlag in~\cites{GS1,GS2}.

\begin{lemma}\label{lem:LDTsig}
For $\omega\in \mathrm{DC}$, there exists $\delta>0$ such that for any $|\varepsilon|\leq \eta$ and $n$ large enough, the following large deviation set 
\begin{align}\label{def:badnE}
\mathcal{B}_{n,E,\varepsilon}:=\left\{\theta\in\tor^b:\, \frac{1}{n}\log \|\textstyle{\bigwedge^d}  M_{n,E}(\theta+i\varepsilon)\|\leq {L}_{(n),\varepsilon}^d(\omega,M_E)-n^{-\delta}\right\}
\end{align}
satisfies $\mathrm{mes}(\mathcal{B}_{n,E,\varepsilon})\leq e^{-n^\delta}$. 
\end{lemma}

We will also use the Lipschitz continuity of $L^{d}_{(n),\varepsilon}(\omega,E)$ with respect
to $\varepsilon$.
\begin{lemma}\label{lem:Lip_eps}\cite{GSV}*{Corollary 2.12}
There exists $C=C(B,V,|E|)>0$, such that for each $1\leq j\leq 2d$,
    \begin{align}\label{eq:Lip_eps}
        |L^j_{(n),\varepsilon}(\omega,M_E)-L^j_{(n),\varepsilon'}(\omega,M_E)|\leq C\sum_{\ell=1}^b |\varepsilon_\ell-\varepsilon_\ell'|,
    \end{align}
   for all sufficiently small $|\varepsilon|$, and  
    uniformly in $n$. In particular,  the same bound hold with $L^j_{\varepsilon}$ instead of $L^j_{(n),\varepsilon}$.
\end{lemma}

By Lemma~\ref{lem:Lip_eps}, the Lyapunov exponents are Lipshitz continuous in $\varepsilon$. Since throughout the paper, we work under the condition that $L_d(\omega,M_E)\geq \gamma>0$, we may shrink $\eta$ (depending on $\gamma$) to guarantee 
\begin{align}\label{eq:L_d>0_varepsilon}
\inf_{|\varepsilon|\leq \eta} L_{d,\varepsilon}(\omega, M_E)\geq \frac{1}{2}L_d(\omega,M_E)>0, \text{\ \  and \ \ } \sup_{|\varepsilon|\leq \eta}L_{d,\varepsilon}(\omega,M_E)\leq -\frac{1}{2}L_d(\omega,M_E)<0.
\end{align}
 
The Avalanche Principle was first introduced for $\mathrm{SL}(2,\R)$ cocycles by Goldstein and Schlag, see~\cite{GS1}.
It was extended to larger matrices in \cites{S1,DK}.

\begin{theorem}\label{thm:APDK}\cite{DK}*{Proposition 2.42}
Let $m\ge 2$ be fixed. 
There exist $c_0, C_0>0$ such that for any $0<\varepsilon<1$, $0<\kappa<c_0\, \varepsilon^2$ and $g_0,g_1,...,g_{n-1}\in \mathrm{Mat}(m,\R)\setminus\{0\}$ satisfying 
\begin{align}
\frac{\sigma_1(g_j)}{\sigma_2(g_j)}>\frac{1}{\kappa}, \text{\ \  for all\ \  } 0\leq j\leq n-1\\
\frac{\|g_j g_{j-1}\|}{\|g_j\| \|g_{j-1}\|}>\varepsilon, \text{\ \  for all\ \  } 1\leq j\leq n-1
\end{align}
one has
\begin{align}
\Big| \log \|g_{n-1}\cdots g_0\|+\sum_{j=1}^{n-2}\log \|g_j\|-\sum_{j=1}^{n-1} \log \|g_j g_{j-1}\|\Big|\leq C_0\, n \frac{\kappa}{\varepsilon^2}.
\end{align}
\end{theorem}

The following rate of convergence of $L^d_{(n),\varepsilon}$ to $L^d_{\varepsilon}$ holds, see~\cite{GS1}*{Lemma 10.1}. 

\begin{lemma}\label{lem:Ln-L}
Let $\omega\in \mathrm{DC}$.
Suppose $L_d(\omega,M_E)\geq \gamma>0$, then there exists $\delta>0$ such that for any $|\varepsilon|\leq \eta$ and $n\geq n(\gamma)$, we have
\begin{align}
    L^d_{\varepsilon}(\omega,M_E)\leq L^d_{(n),\varepsilon}(\omega,M_E)\leq L^d_{\varepsilon}(\omega,M_E)+n^{-\delta}.
\end{align}
\end{lemma}
The proof uses the Avalanche principle, which requires the positivity of $L_{d,\varepsilon}(\omega,E)$, provided by \eqref{eq:L_d>0_varepsilon}.

\subsection{Green's function and Poisson formula}
As in \cite{HS3} we work with finite volume Hamiltonians under periodic boundary conditions. 
Thus, we define the $nd\times nd$ matrices  
\begin{align}\label{def:Pn}
P_{n}(\theta)=
\left(\begin{matrix}
V(\theta+(n-1)\omega) & B^{(*)}(\theta+(n-1)\omega) & &  &B(\theta)\\
B(\theta+(n-1)\omega) &V(\theta+(n-2)\omega) &\ddots \\
& \ddots &\ddots &\ddots \\
& &\ddots &\ddots &B^{(*)}(\theta+\omega)\\
B^{(*)}(\theta) & & &B(\theta+\omega) &V(\theta)
\end{matrix}\right),
\end{align}
Let 
\begin{align}\label{def:fn}
f_{E,n}(\theta):=\det(P_{n}(\theta)-E)
\end{align}
and
\begin{align}\label{def:Green}
G_{E,n}(\theta):=(P_{n}(\theta)-E)^{-1}
\end{align}
be the finite volume Green's function with the periodic boundary conditions. 
By Cramer's rule  
\begin{align}\label{eq:mufn}
G_{E,n}(\theta;x,y)=\frac{\mu_{n,x,y}(\theta)}{f_{E,n}(\theta)},
\end{align}
where $\mu_{n,x,y}(\theta)$ is the determinant of the submatrix of $(P_{n}(\theta)-E)$ defined by deleting the $x$-th row and $y$-th column.
Let $u$ be a solution to the eigenvalue equation ${H}_{\theta}u=Eu$.
For any $k\in \Z$
the following Poisson formula holds for all  $k\leq m\leq k+nd-1$:
\begin{align}
{u}_m=&\sum_{y_1=0}^{d-1} G_{E,n}(\theta+k\omega; m-k, y_1) \left(B^{(*)}(\theta)\cdot \left(\begin{matrix}{u}_{k+nd-1}-{u}_{k-1}\\ \vdots\\ {u}_{k+(n-1)d}-{u}_{k-d}\end{matrix}\right)\right)_{y_1}\\
&\qquad+\sum_{y_2=(n-1)d}^{nd-1}G_{E,n}(\theta+k\omega;m-k,y_2)\left(B(\theta) \left(\begin{matrix}u_{k+d-1}\\ \vdots\\ u_k\end{matrix}\right)-B(\theta+n\omega)\left(\begin{matrix}{u}_{k+(n+1)d-1}\\ \vdots \\ {u}_{k+nd}\end{matrix}\right)\right)_{y_2-(n-1)d},
\end{align}
in which $(M)_y$ refers to the element of vector $M$ in row $y$.
This implies 
\begin{align}\label{eq:Poisson_exp}
|{u}_m|\leq &C_d \|B\|_{\T^b, \infty}\cdot \max_{y\in \{0,...,d-1\}\cup \{(n-1)d,...,nd-1\}} |G_{E,n}(\theta+k\omega; m-k, y)|\cdot \\
&\qquad\qquad\qquad\qquad\qquad\qquad\qquad\cdot \max_{\ell\in \{-d,...,d-1\}\cup \{(n-1)d,...,(n+1)d-1\}} |{u}_{k+\ell}|
\end{align}

\subsection{Numerator and denominator of the Green's function}\label{sec:lemma}
 Lemma~\ref{lem:numerator} bounds the numerator of the Green's function, and Lemma~\ref{lem:deno} the denominator.

\begin{lemma}\label{lem:numerator}
Let $\omega\in \mathrm{DC}$.
Let $3d\leq y\leq (n-1)d-1$ and $0\leq x\leq d-1$ or $(n-1)d\leq x\leq nd-1$. Set $\ell:=\lfloor y/d\rfloor$. Then 
for any $\varepsilon>0$, and uniformly in $\theta\in \T^b$, 
\begin{align}
|\mu_{n,x,y}(\theta)|\leq C_{d,B} \cdot e^{n(\langle\log |\det B|\rangle+\varepsilon)} \cdot \left(e^{\ell {L}^{d-1}+(n-\ell){L}^d}+e^{\ell {L}^d+(n-\ell){L}^{d-1}}\right),
\end{align}
where $L^j=L^j(\omega,M_E)$, provided $n>N(\varepsilon)$ is large enough. 
Here $C_{d,B}$ is a constant depending only on~$d$ and $\|B^{-1}\|_{\T^b, \infty}$. 
\end{lemma}

The proof proceeds as in~\cite{HS3}, see Section~\ref{sec:Zaehler}. The upper bound with $\varepsilon$ suffices to obtain exponential decay of the Green's function. 

Regarding the denominator, we first have the following connection between $f_{E,n}$ and the transfer matrix $M_{n,E}$.
\begin{lemma}\label{lem:detP}
One has pointwise in $\theta\in \T^b_{\eta}$ that
\begin{align}
|f_{E,n}(\theta)|=  |\det (M_{n,E}(\theta)-I_{2d})| \cdot \prod_{j=0}^{n-1}|\det B(\theta+j\omega)| .
\end{align}
\end{lemma}

\begin{lemma}\label{lem:deno}
Let $\omega\in \mathrm{DC}$, and $\delta>0$ be as in Lemma \ref{lem:LDTsig}.
Assume $L_d(\omega,M_E)\geq \gamma>0$.
There exist $\delta_1\in (0,\delta)$, $N_0>1$ large and $0<\kappa_0\ll 1$ so that the {\it $\kappa_0$-admissible} sequence 
\begin{align}\label{def:admissible}
\mathcal{N}:=\{n\geq N_0: \|n\omega\|_{\tor^b}\leq \kappa_0\}
\end{align}
has the following property: 
for any $|\varepsilon|\leq \eta/2$, and all large $\kappa_0$-admissible $n$, the following large deviation set 
\begin{align}\label{def:B_fEn}
\mathcal{B}_{f,E,n,\varepsilon}:=
\big\{\theta\in \tor^b: \log |f_{E,n}(\theta+i\varepsilon)|<n(\langle \log |\det B(\cdot+i\varepsilon)|\rangle +{L}^d_{\varepsilon}(\omega,M_E)) - n^{1-\delta_1}\big\}
\end{align}
satisfies $\mathrm{mes}(\mathcal{B}_{f,E,n,\varepsilon})<e^{-n^{\delta_1}}$.
\end{lemma}
\begin{remark}\label{rem:admissible}
    For every large integer $n>0$ there exists an admissible $\tilde n>0$ with $|n-\tilde n|\le C_*$ for some constant $C_*$.
\end{remark}

The proof is analogous to the denominator bound in \cite{HS3}, using the original strategy of Proposition~3.3 in~\cite{GS2}. 
We postpone the proofs of Lemmas \ref{lem:deno} and the following corollary to Sec.~\ref{sec:nenner}.
\begin{lemma}\label{lem:fn_ave}
Under the same conditions as Lemma \ref{lem:deno}.
There exists $\delta_2\in (0,\delta_1)$ such that for any $|\varepsilon|\leq \eta/2$ and large $\kappa_0$-admissible $n$, 
\begin{align}
\frac{1}{n} \int_{\T^b}\log |f_{E,n}(\theta+i\varepsilon)|\, \mathrm{d}\theta\geq L^d_{\varepsilon}(\omega,M_E)+\langle \log |\det B(\cdot+i\varepsilon)|\rangle -n^{-\delta_2}.
\end{align}
\end{lemma}
This lemma (with $b=1$) will only be used in the proof of arithmetic Anderson localization in Sec.\ref{sec:AL_arithmetic}. 
In view of Lemmas \ref{lem:deno} and \ref{lem:fn_ave}, we will further shrink $\eta$ to $\eta/2$ such that those estimates hold for $|\varepsilon|\leq \eta$.

The following pointwise upper bound of $f_{E,n}$, which does not require admissible~$n$, complements the preceding lower bound.
\begin{lemma}\label{lem:fn_upper}
Let $\omega\in \mathrm{DC}$, and $\delta>0$ be as in Lemma \ref{lem:upperbd}. For $n$ large enough, we have uniformly in $\theta\in \T^b$ and $|\varepsilon|\leq \eta$ that
    \begin{align}
\frac{1}{n}\log |f_{E,n}(\theta+i\varepsilon)|\leq L^d_{\varepsilon}(\omega,M_E)+\langle \log |\det B(\cdot+i\varepsilon)|\rangle+n^{-\delta}.
    \end{align}
\end{lemma}

\section{Bounding the numerator: Lemma~\ref{lem:numerator}}
\label{sec:Zaehler}
We write the monodromy matrices in block form
\begin{align}
M_{n,E}(\theta) =\left(\begin{array}{c|c}
M_{n,E}^{UL}(\theta)  & M_{n,E}^{UR}(\theta) \\
\hline
M_{n,E}^{LL}(\theta)  & M_{n,E}^{LR}(\theta)
\end{array}\right),
\end{align}
where each $M_{n,E}^{\dagger}$ is a $d\times d$ block, $\dagger=UL, UR, LL, LR$.
We will make use of the following recursive relations: for $n=1$,
\begin{align}\label{eq:rec1}
\begin{cases}
M_{1,E}^{UL}(\theta) =-(V(\theta) -E)B^{-1}(\theta)\\
M_{1,E}^{UR}(\theta) =-B^{(*)}(\theta)\\
M_{1,E}^{LL}(\theta) =B^{-1}(\theta)\\
M_{1,E}^{LR}(\theta) =0
\end{cases}
\end{align}
and for each $n\geq 2$, one has 
\begin{align}\label{eq:rec2}
\begin{cases}
M_{n,E}^{UL}(\theta)=-M_{n-1,E}^{UL}(\theta+\omega) (V(\theta)-E)B^{-1}(\theta)+M_{n-1,E}^{UR}(\theta+\omega) B^{-1}(\theta)\\
M_{n,E}^{UR}(\theta)=-M_{n-1,E}^{UL}(\theta+\omega) B^{(*)}(\theta)\\
M_{n,E}^{LL}(\theta)=-M_{n-1,E}^{LL}(\theta+\omega) (V(\theta)-E)B^{-1}(\theta)+M_{n-1,E}^{LR}(\theta+\omega)B^{-1}(\theta)\\
M_{n,E}^{LR}(\theta)=-M_{n-1,E}^{LL}(\theta+\omega) B^{(*)}(\theta)
\end{cases}
\end{align}
We now turn to the proof of Lemma~\ref{lem:numerator}, which is a straightforward adaption of Section~4 of~\cite{HS3}. 
We restrict ourselves to the case $0\leq x\leq d-1$ and $3d\leq y\leq (n-1)d-1$. For $(n-1)d\leq x\leq nd-1$, one proceeds analogously, see~\cite{HS3}. 
With  $y=\ell d+r$, $\ell\in [3,n-2]$ and $r\in [0,d-1]$, we let 
\begin{align}
R_{x,y}:=\left(\begin{array}{c|c}
P_n(\theta)-E & {\bf e}_{dn,x}\\
\hline
{\bf e}_{dn,y}^{*} & 0
\end{array}\right),
\end{align}
where ${\bf e}_{m,j}^*=(\delta_j(m-1),...,\delta_j(1),\delta_j(0))$. 
By definition, 
\begin{align}\label{eq:muxy=Rxy}
|\mu_{n,x,y}(\theta)|= |\det R_{x,y}|
\end{align}
and, with $V(\theta+j\omega)-E=:C_j$, $B(\theta+j\omega)=:B_j$ and $B^{(*)}(\theta+j\omega)=:B_j^{(*)}$,
\begin{align}\label{eq:Rxy}
\qquad
R_{x,y}
=
&\left(\begin{array}{c|c|c|c|c|c|c|c|c|c}
C_{n-1} & B_{n-1}^{(*)} & & &  & & & &B_0 &\\
\hline
B_{n-1} &\ddots &\ddots & & & & & & & \\
\hline
&\ddots &\ddots &\ddots & & & & & &\\
\hline
& &\ddots &\ddots &B_{\ell+1}^{(*)} & & & & &\\
\hline
& & & B_{\ell+1} &C_{\ell} &B_\ell^{(*)} & & & &\\
\hline
& & & &B_\ell &C_{\ell-1} &\ddots & & &\\
\hline
& & & & &\ddots &\ddots &\ddots & &\\
\hline
& & & & & &\ddots &\ddots &B_1^{(*)} &\\
\hline
B_0^{(*)}& & & & & & &B_1 &C_0 &{\bf e}_{d,x}\\
\hline
& & & &{\bf e}_{d,r}^* & & & & &
\end{array}\right)
=:
\left(\begin{matrix}
\mathrm{Row}_1\\
\mathrm{Row}_2\\
\vdots\\
\text{Row}_{n+1}
\end{matrix}\right)
\end{align}
Performing the identical row operations as in~\cite{HS3} we obtain that $|\det R_{x,y}|=|\det R^{(1)}_{x,y}|$ where
\begin{align}\label{eq:R(1)}
R^{(1)}_{x,y} \!\!
=\!\!
\left(\!\!\begin{array}{c|c|c|c|c|c|c|c|c|c|c|c|c}
\!\!0 \!\!& 0 &0 &\cdots & &  & & &\cdots &\!\!0\!\! \!\!&-M_{n-1}^{UL}(1)B_1 &B_0-M_{n-1}^{UR}(1) & \!\!0 \\
\hline
\!\!B_{n-1}\!\! &0 &0 &\cdots & & & & &\cdots &\!\!0\!\!\!\! &-M_{n-2}^{UL}(1)B_1 &-M_{n-2}^{UR}(1) & \!\!0 \\
\hline 
\!\!0\!\!&\!\!B_{n-2} \!\!&\!\!C_{n-3} \!\!&\!\!B_{n-3}^{(*)} \!\!& & & & & & & & &\\
\hline
\vdots& &\ddots &\ddots &\ddots & & & & & & & &\\
\hline
& & & \ddots &\ddots &\ddots & & & & & & &\\
\hline
\vdots & & & &\!\!B_{\ell+1} \!\!&\!\!C_{\ell}\!\! &\!\!B_\ell^{(*)}\!\!& & & & & &\\
\hline
\!\!0\!\!&\cdots & &\cdots &\!\!0 \!\!&\!\!B_\ell \!\!&\!\!0 \!\!&\!\!0\!\! &\cdots &\!\!0\!\! \!\!&-M_{\ell-1}^{UL}(1)B_1 &-M_{\ell-1}^{UR}(1) & \!\!0\\
\hline
\vdots & & & & &\!\!0 \!\!&\!\!B_{\ell-1}\!\! &\!\!C_{\ell-2} \!\!&\!\!B_{\ell-2}^{(*)} \!\!& & &0 &\!\!0\\
\hline
& & & & & & &\ddots &\ddots &\ddots & &\vdots &\vdots\!\!\\
\hline
\vdots& & & & & & & &\ddots &\ddots &\ddots &0 &\vdots\!\!\\
\hline
\!\!0\!\!& & & & & & & & &\ddots &\ddots &B_1^{(*)} &\!\!0\!\!\\
\hline
\!\! B_0^{(*)}\!\! &\!\!0 \!\!&\cdots & & & & & &\cdots &\!\!0 \!\!\!\!&B_1 &C_0 &{\bf e}_{d,x}\!\!\\
\hline
\!\!0\!\! &\cdots & &\cdots &\!\! 0 \!\!&\!\!{\bf e}_{d,r}^*\!\! & \!\!0\!\! &\cdots & &\cdots &  0 & 0 & \!\!0 
\end{array}\!\!\right)
\end{align}
where   $M_{k,E}(\theta+j\omega)=:M_k(j)$. By inspection, rows $1$, $2$, $n-\ell+1$, $n$, and $n+1$ are
\begin{align}
\left(\begin{matrix}
\text{Row}_1^{(n-2)}\\
\text{Row}_2^{(n-3)}\\
\text{Row}_{n-\ell+1}^{(\ell-2)}\\
\text{Row}_{n}\\
\text{Row}_{n+1}
\end{matrix}
\right)
=\left(\begin{array}{ccccccccccc}
0 &0 &\cdots &0 &0 &0 &\cdots &0 &-M_{n-1}^{UL}(1)B_1 &B_0-M_{n-1}^{UR}(1) &0\\
B_{n-1} &0 &\cdots &0 &0 &0 &\cdots &0 &-M_{n-2}^{UL}(1)B_1 &-M_{n-2}^{UR}(1) &0\\
0 &0 &\cdots &0 &B_\ell &0 &\cdots &0 &-M_{\ell-1}^{UL}(1)B_1 &-M_{\ell-1}^{UR}(1) &0\\
B_0^{(*)} &0 &\cdots &0 &0 &0 &\cdots &0 &B_1 &C_0 &{\bf e}_{d,x}\\
0 &0 &\cdots &0 &{\bf e}_{d,r}^* &0 &\cdots &0 &0 &0 &0
\end{array}\right),
\end{align}
in which only columns $1$, $n-\ell$, $n-1$, $n$, $n+1$ are non-vanishing.
Define
\begin{align}
    S_1:=\left(\begin{array}{ccccc}
0 &0 &-M_{n-1}^{UL}(1)B_1 &B_0-M_{n-1}^{UR}(1) &0\\
B_{n-1} &0 &-M_{n-2}^{UL}(1)B_1 &-M_{n-2}^{UR}(1) &0\\
0 &B_\ell &-M_{\ell-1}^{UL}(1)B_1 &-M_{\ell-1}^{UR}(1) &0\\
B_0^{(*)} &0 &B_1 &C_0 &{\bf e}_{d,x}\\
0 &{\bf e}^*_{d,r} & 0 &0 &0
    \end{array}\right)
\end{align}
as a $(4d+1)\times (4d+1)$ submatrix of rows $1$, $2$, $n-\ell+1$, $n$ and $n+1$. It is unique with the property that any other  $(4d+1)\times (4d+1)$ submatrix of these rows has vanishing determinant. 
Let
\begin{align}
    S_2:=
    \left(\begin{array}{cccccccccccccccc}
       B_{n-2} &C_{n-3} &B_{n-3}^{(*)} & & & & & & & & & & & & &\\
          &B_{n-3}       &C_{n-4} &\ddots & & & & & & & & & & & &\\
          &        &\ddots  &\ddots & & & & & & & & & & & &\\
                   \\
          & & & & &\ddots &\ddots  & & & & & & & & &\\
          & & & & &\ddots &C_{\ell+2} &B_{\ell+2}^{(*)} & & & & & & & &\\
          & & & & &  &B_{\ell+2} &C_{\ell+1} &0 & & & & & & &\\
          & & & & &  & &B_{\ell+1} &B_{\ell}^{(*)} &0 & & & & & &\\
          & & & & &  & & &B_{\ell-1} &C_{\ell-2} &B_{\ell-2}^{(*)} & & & & &\\
          & & & & &  & & & &B_{\ell-2} &\ddots &\ddots & & & &\\
          & & & & &  & & & & &\ddots & & & & &\\
          \\
          & & & & &  & & & & & & &\ddots &\ddots &\ddots &\\
          & & & & &  & & & & & & & &B_4 &C_3 &B_3^{(*)}\\
          & & & & &  & & & & & & & & &B_3 &C_2   \\ 
          & & & & &  & & & & & & & & & &B_2         
    \end{array}\right),
\end{align}
which is the submatrix of $R^{(1)}_{x,y}$ obtained by deleting rows $1$, $2$, $n-\ell+1$, $n$, $n+1$ and columns $1$, $n-\ell$, $n-1$, $n$, $n+1$.
Hence 
\begin{align}\label{eq:num1}
|\det R^{(1)}_{x,y}|=
|\det S_1|\cdot |\det S_2|= |\det S_1| \cdot \prod_{\substack{j=2\\ j\ne\ell}}^{n-2}|\det B_j|,
\end{align}
We simplify 
\begin{align}\notag
|\det S_1|=
&\left|\det\left(\begin{array}{ccccc}
0 &0 &-M_{n-1}^{UL}(1)B_1 &B_0-M_{n-1}^{UR}(1) &0\\
B_{n-1} &0 &-M_{n-1}^{LL}(1)B_1 &-B_{n-1}M_{n-1}^{LR}(1) &0\\
0 &B_\ell &-M_{\ell-1}^{UL}(1)B_1 &-M_{\ell-1}^{UR}(1) &0\\
B_0^{(*)} &0 &B_1 &C_0 &{\bf e}_{d,x}\\
0 &{\bf e}^*_{d,r} & 0 &0 &0
    \end{array}\right)\right| \notag\\
=
&|\det B_1|\cdot \left|\det\left(\begin{array}{ccccc}
0 &0 &-M_{n-1}^{UL}(1) &B_0-M_{n-1}^{UR}(1) &0\\
B_{n-1} &0 &-B_{n-1}M_{n-1}^{LL}(1) &-B_{n-1}M_{n-1}^{LR}(1) &0\\
0 &B_\ell &-M_{\ell-1}^{UL}(1) &-M_{\ell-1}^{UR}(1) &0\\
B_0^{(*)} &0 &I_d &C_0 &{\bf e}_{d,x}\\
0 &{\bf e}^*_{d,r} & 0 &0 &0
    \end{array}\right)\right|\notag 
        \end{align}
        Pulling out the matrices in the first column, we may further simplify this in the form 
    \begin{align}
  |\det S_1|  = &|\det B_0^{(*)}| |\det B_1||\det B_{n-1}|\cdot \left|\det\left(\begin{array}{ccccc}
0 &0 &-M_{n-1}^{UL}(1) &B_0-M_{n-1}^{UR}(1) &0\\
 I_d &0 &- M_{n-1}^{LL}(1) &- M_{n-1}^{LR}(1) &0\\
0 &B_\ell &-M_{\ell-1}^{UL}(1) &-M_{\ell-1}^{UR}(1) &0\\
I_d &0 & (B_0^{(*)})^{-1} & (B_0^{(*)})^{-1}C_0 & (B_0^{(*)})^{-1}{\bf e}_{d,x}\\
0 &{\bf e}^*_{d,r} & 0 &0 &0
    \end{array}\right)\right|\notag \\
    = &|\det B_0^{(*)}| |\det B_1||\det B_{n-1}|\cdot \left|\det\left(\begin{array}{cccccc}
0& 0 &0 &-M_{n-1}^{UL}(1) &B_0-M_{n-1}^{UR}(1) &0\\
0& I_d &0 &- M_{n-1}^{LL}(1) &- M_{n-1}^{LR}(1) &0\\
0& 0 &B_\ell &-M_{\ell-1}^{UL}(1) &-M_{\ell-1}^{UR}(1) &0\\
I_d &0 &0 &0 &-B_0 &0\\
0& I_d &0 & (B_0^{(*)})^{-1} & (B_0^{(*)})^{-1}C_0 & (B_0^{(*)})^{-1}{\bf e}_{d,x}\\
0& 0 &{\bf e}^*_{d,r} & 0 &0 &0
    \end{array}\right)\right|\notag
\\
=&|\det B_0^{(*)}| |\det B_1||\det B_{n-1}|\cdot \left|\det\left(\begin{array}{c|c|c|c}
\left(\begin{matrix}0& 0\\
0 &I_d\end{matrix}\right)& &\left(\begin{matrix}0 & B_0\\ 0 & 0\end{matrix}\right)-M_{n-1}(1) &\\
\hline
 & B_{\ell} &-M_{\ell-1}^U(1) &  \\
 \hline
I_{2d} & &-M_1^{-1}(0) &\left(\begin{matrix}0\\ I_d\end{matrix}\right) (B_0^{(*)})^{-1}{\bf e}_{d,x}\\
\hline
 &{\bf e}^*_{d,r} & &
    \end{array}\right)\right|\notag\\
=&\frac{|\det B_0^{(*)}|}{|\det M_1(0)|} |\det B_1||\det B_{n-1}|\cdot \left|\det\left(\begin{array}{c|c|c|c}
\left(\begin{matrix}0& 0\\
0 &I_d\end{matrix}\right)& &\left(\begin{matrix}I_d & 0\\ 0 & 0\end{matrix}\right)-M_{n}(0) &\\
\hline
 & B_{\ell} &-(I_d,\, 0) M_{\ell}(0) &  \\
 \hline
I_{2d} & &-I_{2d} &\left(\begin{matrix}0\\ I_d\end{matrix}\right) (B_0^{(*)})^{-1}{\bf e}_{d,x}\\
\hline
 &{\bf e}^*_{d,r} & &
    \end{array}\right)\right|\notag\\
=&|\det B_0||\det B_1||\det B_{n-1}|\cdot \left|\det\left(\begin{array}{c|c|c|c}
\left(\begin{matrix}0& 0\\
0 &I_d\end{matrix}\right)& &\left(\begin{matrix}I_d & 0\\ 0 & 0\end{matrix}\right)-M_{n}(0) &\\
\hline
 & B_{\ell} &-(I_d,\, 0) M_{\ell}(0) &  \\
 \hline
I_{2d} & &-I_{2d} &\left(\begin{matrix}0\\ I_d\end{matrix}\right) (B_0^{(*)})^{-1}{\bf e}_{d,x}\\
\hline
 &{\bf e}^*_{d,r} & &
    \end{array}\right)\right|\notag
    \end{align}
As in \cite{HS3}, this can now be rewritten in the form 
\begin{equation}
    \label{eq:detS1B}
    |\det S_1|= |\det (B_0B_1B_\ell B_{n-1})|
\biggl|\det\biggl(\begin{array}{c|c}
    -M_n(0)+I_{2d} &-\left(\begin{matrix}0\\ I_d\end{matrix}\right) (B_0^{(*)})^{-1}{\bf e}_{d,x}\\
    \hline
    {\bf e}_{d,r}^*B_\ell^{-1}(I_d, 0)\cdot M_{\ell}(0) &0
\end{array}\biggr)\biggr| 
\end{equation}
At this point we proceed as for the upper bound on~$|\det S_3^{(3)}|$ in Section~4 of~\cite{HS3}. 
This leads to the following, note $B_0^{(*)}=(B_0)^*$ for $\theta\in \T^b$:
\begin{align}
   |\mu_{n,x,y}| &\le C_d \|B^{-1}\|_\infty^2 \prod_{j=0}^{n-1}|\det B_j| \cdot \sum_{m_0} \|\textstyle{\bigwedge^j}^{m_0} M_{\ell}(0)\|\cdot \|\textstyle{\bigwedge^j}^{m_0-1} M_{n-\ell}(\ell)\| 
\end{align}
Combining this with Lemma~\ref{lem:upperbd} proves the claimed result.

\section{Bounding the denominator: Lemmas~ \ref{lem:detP}, \ref{lem:deno}, \ref{lem:fn_ave} and \ref{lem:fn_upper}}
\label{sec:nenner}

\begin{proof}[Proof of Lemma \ref{lem:detP}]
    This is implicit in the calculations of the previous section, cf.~the upper left-hand corner of the block matrix of~\eqref{eq:detS1B}. For more details, see~\cite{HS3}*{Lemma 5.1}. 
\end{proof}

\begin{proof}[Proof of Lemma \ref{lem:fn_upper}]
    Let $\{v_{j}^{(n)}(\theta+i\varepsilon)\}_{j=1}^{2d}$ be the set of normalized singular vectors of $M_{n,E}(\theta+i\varepsilon)$ such that 
    \[M_{n,E}(\theta+i\varepsilon) v_j^{(n)}(\theta+i\varepsilon)=\sigma_j(M_{n,E}(\theta+i\varepsilon)) \cdot w_j^{(n)}(\theta+i\varepsilon).\]
    Then 
    \begin{align}\label{eq:Mn_singular_decom}
    M_{n,E}(\theta+i\varepsilon)=W_n(\theta+i\varepsilon)D_n(\theta+i\varepsilon)V_n^*(\theta+i\varepsilon),
    \end{align}
    where $D_n(\theta+i\varepsilon)=\mathrm{diag}(\sigma_j(M_{n,E}(\theta+i\varepsilon))_{j=1}^{2d}$ and $V_n,W_n$ are the matrices with columns $v_j^{(n)}$ and $w_j^{(n)}$ respectively.
    By Hadamard's inequality,
    \begin{align}
        |\det(M_{n,E}(\theta+i\varepsilon)-I_{2d})|=|\det (D_n(\theta+i\varepsilon)-W_n^*(\theta+i\varepsilon)V_n(\theta+i\varepsilon))|\leq \prod_{j=1}^{2d}\|r_j^{(n)}(\theta+i\varepsilon)\|,
    \end{align}
    where $r_j^{(n)}(\theta+i\varepsilon)$ is the $j$-th column of $D_n(\theta+i\varepsilon)-W_n^*(\theta+i\varepsilon)V_n(\theta+i\varepsilon)$.
    Clearly 
    \[\|r_j^{(n)}(\theta+i\varepsilon)\|\leq \sigma_j(M_{n,E}(\theta+i\varepsilon))+1,\] 
    which implies
    \begin{align}\label{eq:detM-I_upper}
        |\det(M_{n,E}(\theta+i\varepsilon)-I_{2d})|
        \leq &\prod_{j=1}^{2d}(\sigma_j(M_{n,E}(\theta+i\varepsilon))+1) \notag\\
        =&\sum_{k=1}^{2d}\, \sum_{1\leq j_1<...<j_{k}\leq 2d}\,\, \prod_{\ell=1}^{k} \sigma_{j_{\ell}}(M_{n,E}(\theta+i\varepsilon))+1.
    \end{align}
Let $\varepsilon_1=L_d(\omega,M_E)/4$. Then for $n$ large enough, and uniformly in $|\varepsilon|\leq \eta$ and $1\leq k\leq 2d$ one has
\begin{align}\label{eq:Lkn-Ln_eps1}
    L^k_{(n),\varepsilon}(\omega,M_E)\leq L^k_{\varepsilon}(\omega,M_E)+\varepsilon_1.
\end{align}
Combining Lemma \ref{lem:upperbd} with \eqref{eq:Lkn-Ln_eps1},
we have for $n$ large enough, for any $1\leq k\leq 2d$, $k\neq d$, uniformly in $\theta$ that
\begin{align}\label{eq:prod_sigma_1}
    \prod_{\ell=1}^k \sigma_{j_{\ell}}(M_{n,E}(\theta+i\varepsilon))\leq \prod_{\ell=1}^k \sigma_{\ell}(M_{n,E}(\theta+i\varepsilon))
    \leq e^{n(L^k_{\varepsilon}(\omega,M_E)+\varepsilon_1)}.
\end{align}
For $k=d$, and $(j_1,...,j_d)\neq (1,...,d)$, one has (see \cite[Lemma 5.11]{HS3}) that 
\begin{align}\label{eq:prod_sigma_2}
\prod_{\ell=1}^d \sigma_{j_{\ell}}(M_{n,E}(\theta+i\varepsilon))\leq e^{n(\max_{k\neq d} L^k_{\varepsilon}(\omega,M_E)+\varepsilon_1)}.
\end{align}
In fact if $\sigma_{j_d}(M_{n,E}(\theta+i\varepsilon))<1$, then
\begin{align}\label{eq:prod_sigma_21}
    \prod_{\ell=1}^d \sigma_{j_{\ell}}(M_{n,E}(\theta+i\varepsilon))\leq 
    \prod_{\ell=1}^{d-1} \sigma_{j_{\ell}}(M_{n,E}(\theta+i\varepsilon))\leq e^{n(L^{d-1}_{\varepsilon}(\omega,M_E)+\varepsilon_1)},
\end{align}
where we applied \eqref{eq:prod_sigma_1} in the last inequality. 
If $\sigma_{j_d}(M_{n,E}(\theta+i\varepsilon))\geq 1$, then
\begin{align}\label{eq:prod_sigma_22}
    \prod_{\ell=1}^d \sigma_{j_{\ell}}(M_{n,E}(\theta+i\varepsilon))\leq 
    \prod_{\ell=1}^{j_d} \sigma_{{\ell}}(M_{n,E}(\theta+i\varepsilon))\leq e^{n(\max_{k\neq d} L^{k}_{\varepsilon}(\omega,M_E)+\varepsilon_1)},
\end{align}
where we applied \eqref{eq:prod_sigma_1} and noted that $j_d>d$.
Combining \eqref{eq:prod_sigma_21} with \eqref{eq:prod_sigma_22} yields \eqref{eq:prod_sigma_2}.

Note for $|\varepsilon|\leq \eta$, by the choice of $\varepsilon_1$ and \eqref{eq:L_d>0_varepsilon}, we have
\begin{align}\label{eq:prod_sigma_3}
    L^d_{\varepsilon}(\omega,M_E)-\max_{k\neq d}(L^k_{\varepsilon}(\omega,M_E)+\varepsilon_1)\geq\frac{1}{4}L_d(\omega,M_E).
\end{align}
Therefore, combining \eqref{eq:prod_sigma_1} and \eqref{eq:prod_sigma_2} with \eqref{eq:detM-I_upper}, yields
\begin{align}
    |\det(M_{n,E}(\theta+i\varepsilon)-I_{2d})|
    \leq &\|\textstyle{\bigwedge}^d M_{n,E}(\theta+i\varepsilon)\|+C_d e^{n(L^d_{\varepsilon}(\omega,M_E)-\frac{1}{4}L_d(\omega,M_E))}\\
    \leq &e^{n(L^d_{\varepsilon}(\omega,M_E)+n^{-\delta})}+C_d e^{n(L^d_{\varepsilon}(\omega,M_E)-\frac{1}{4}L_d(\omega,M_E))}\\
    \leq &2e^{n(L^d_{\varepsilon}(\omega,M_E)+n^{-\delta})},
\end{align}
in which we applied Lemma \ref{lem:upperbd} to bound $\|\textstyle{\bigwedge}^d M_{n,E}(\cdot+i\varepsilon)\|$.
The claimed result follows from combining the above with 
\[\prod_{j=0}^{n-1}|\det B_j|\leq e^{n(\langle\log |B(\cdot+i\varepsilon)|\rangle+n^{-\delta'})},\]
for some $\delta'>\delta>0$ by the Diophantine property $\omega\in \mathrm{DC}$.
\end{proof}

As a preparation for the proof of Lemma \ref{lem:deno}, we first prove the following.
Recall that $n$ is $\kappa_0$-admissible if $\|n\omega\|_{\tor^b}\leq \kappa_0$. 

\begin{lemma}\label{lem:Msq}
Assume $L_d(\omega,M_E)\geq \gamma>0$. 
There exists $0<\kappa_0\ll1$ and $\delta_3\in (0,\delta)$ so that for any $|\varepsilon|\leq \eta/2$ and $\kappa_0$-admissible sufficiently large $n$,  we have
    \begin{align}\label{eq:meas_tBn}
\mes(\tilde{\mathcal{B}}_{n,E,\varepsilon}):=\mes\Big(\Big\{\theta\in\tor^b:\,  \Big|\frac{1}{2n}\log\|\textstyle{\bigwedge}^d \left(M_{n,E}^2(\theta+i\varepsilon)\right)\|- {L}^d_{\varepsilon}(\omega, M_E)\Big|>n^{-\delta_3}\Big\}\Big)\leq e^{-n^{\delta_3}}
    \end{align}
\end{lemma}
\begin{proof}
This is essentially \cite{GS3}*{Lemma 3.2}, with minor modifications for higher dimensional monodromy matrices (see the more recent \cite{HS3}*{Lemma 5.3} for details).
The proof follows by induction from the large deviation estimates in Lemmas~\ref{lem:upperbd}, \ref{lem:LDTsig} and the Avalanche principle of Theorem~\ref{thm:APDK}, together with Cartan-type bounds in higher dimensions. 
We briefly sketch the induction below.
Let $N_0$ be large enough so that the uniform upper bound (Lemma \ref{lem:upperbd}) and large deviation estimate (Lemma \ref{lem:LDTsig}) hold for any $n\geq N_0$ and $|\varepsilon|\leq 3\eta/4$. Then there exists (see e.g. \cite[Lemma 5.4]{HS3}) $\kappa_0=\kappa_0(N_0)$ such that for any $|\varepsilon|\leq \eta/2$, $|\kappa|\leq \kappa_0$, and any $\theta$ such that
\begin{align}
    \|\textstyle{\bigwedge}^d \left(M_{2N_0,E}(\theta+i\varepsilon)\right)\|\geq e^{2N_0L^d_{\varepsilon}-N_0^{1-\delta}},
\end{align}
we have
\begin{align}
    \|\textstyle{\bigwedge}^d \left(M_{N_0,E}(\theta+N_0\omega+\kappa+i\varepsilon)M_{N_0,E}(\theta+i\varepsilon)\right)\|\geq e^{2N_0L^d_{\varepsilon}-2N_0^{1-\delta}}.
\end{align}
This, together with the large deviation Lemma \ref{lem:LDTsig} with $n=2N_0$, implies the following
\begin{lemma}\label{lem:induct_N0}
For $|\kappa|\leq \kappa_0$, one has 
\begin{align}\label{eq:induction_N0}
\mathrm{mes}\left(\left\{\theta\in \T^b: \left|\frac{1}{2N_0}\log \|\textstyle{\bigwedge}^d \left(M_{N_0,E}(\theta+N_0\omega+\kappa+i\varepsilon)M_{N_0,E}(\theta+i\varepsilon)\right)\|-L^d_{\varepsilon}\right|>N_0^{-\delta}\right\}\right)\leq e^{-N_0^{\delta}}.
\end{align}
\end{lemma}
Next, consider any $N_1$ such that $4N_0+1\leq N_1\leq e^{N_0^{\delta/2}}$.
For any such $N_1$, we decompose $N_1=\ell_1 N_0+r_0$, $N_0\leq r_0<2N_0$,
\begin{align}
    &\textstyle{\bigwedge}^d\left(M_{N_1,E}(\theta+N_1\omega+\kappa+i\varepsilon)M_{N_1,E}(\theta+i\varepsilon)\right)\\
    =&\textstyle{\bigwedge}^d\Big(M_{r_0,E}(\theta+(N_1+\ell_1 N_0)\omega+\kappa+i\varepsilon)(\prod_{k=\ell_1-1}^{1} M_{N_0,E}(\theta+(N_1+kN_0)\omega+\kappa+i\varepsilon))\cdot.\\ 
    &\qquad\cdot M_{N_0,E}(\theta+N_1\omega+\kappa+i\varepsilon)M_{N_0,E}(\theta+(N_1-N_0)\omega+i\varepsilon)\\
    &.\qquad\cdot (\prod_{j=\ell_1-2}^{0} M_{N_0,E}(\theta+(jN_0+r_0)\omega+i\varepsilon))M_{r_0,E}(\theta+i\varepsilon)\Big).
\end{align}
By Lemma \ref{lem:induct_N0} applied to the product of the middle two matrices, and Lemma \ref{lem:LDTsig} and Theorem~\ref{thm:APDK}, one has 
\begin{align}\label{eq:MN1=MN0}
    &\log \|\textstyle{\bigwedge}^d\left(M_{N_1,E}(\theta+N_1\omega+\kappa+i\varepsilon)M_{N_1,E}(\theta+i\varepsilon)\right)\|\\
    &\qquad\qquad-\log \|\textstyle{\bigwedge}^d M_{N_1,E}(\theta+N_1\omega+\kappa+i\varepsilon)\|-\log \|\textstyle{\bigwedge}^d M_{N_1,E}(\theta+i\varepsilon)\|\notag\\
    =&\log \|\textstyle{\bigwedge}^d \left(M_{N_0,E}(\theta+N_1\omega+\kappa+i\varepsilon)M_{N_0,E}(\theta+(N_1-N_0)\omega+i\varepsilon)\right)\|\\
    &\qquad\qquad -\log \|\textstyle{\bigwedge}^d \left(M_{N_0,E}(\theta+N_1\omega+\kappa+i\varepsilon)\right)\|\notag\\
    &\qquad\qquad -\log \|\textstyle{\bigwedge}^d \left(M_{N_0,E}(\theta+(N_1-N_0)\omega+i\varepsilon)\right)\|+O(e^{-\frac{N_0}{2}L_d})
\end{align}
for $\theta\in \mathcal{B}_1$ where $\mathrm{mes}(\mathcal{B}_1)\leq e^{-N_0^{\delta}/2}$. Clearly, \eqref{eq:MN1=MN0} implies that for $\theta\in \mathcal{B}_1$,
\begin{align}\label{eq:LDT_MN1_before}
    u_{N_1,\kappa}(\theta+i\varepsilon):=\log \|\textstyle{\bigwedge}^d \left(M_{N_1,E}(\theta+N_1\omega+\kappa+i\varepsilon)M_{N_1,E}(\theta+i\varepsilon)\right)\|\geq 2N_1L^d_{\varepsilon}-3N_1^{1-\delta}.
\end{align}
We recall the following Cartan estimate for several variables:
\begin{lemma}\label{lem:Cartan}\cite{GS2}*{Lemma 2.15}
Let $\varphi(z_1,...,z_b)$ be an analytic function defined on a polydisck $\mathcal{P}=\prod_{j=1}^b \mathcal{D}(z_{j,0},1)$, $z_{j,0}\in \C$. Let $K\geq \sup_{z\in \mathcal{P}}\log |\varphi(z)|$, $m\leq \log |\varphi(z_0)|$, $z_0=(z_{1,0},...,z_{b,0})$. Given $H\gg 1$ there exists a set $\mathcal{B}\subset \mathcal{P}$, $\mathrm{mes}_{\R^b}(\mathcal{B}\cap \R^b)\leq C_b e^{-H}$, such that 
\begin{align}
    \log |\varphi(z)|>K-C_b H(K-m)
\end{align}
for any $z\in \frac{1}{6}\mathcal{P}\setminus \mathcal{B}$. 
\end{lemma}
\begin{remark}
    In the lemma above, we omitted the definition of Cartan sets, and instead only state the measure estimate.
\end{remark}
Note that \eqref{eq:LDT_MN1_before} provides us with a lower bound for $u_{N_1,\kappa}$, up to a set of measure $e^{-N_0^{\delta}/2}$, needed in the Cartan estimate. We still need an uniform upper bound in a neighborhood of $\T^b$.
Applying Lemmas \ref{lem:upperbd} and \ref{lem:Ln-L} to the cocycle $(\omega, M_E(\cdot+i\varepsilon+i\varepsilon'))$, implies uniformly in $\theta$ that
\begin{align}\label{eq:MM_upper1}
    \sup_{|\varepsilon'|\leq N_1^{-1}}\log \|\textstyle{\bigwedge}^d M_{N_1,E}(\theta+i\varepsilon+i\varepsilon')\|
    \leq &N_1 \sup_{|\varepsilon'|\leq N_1^{-1}} L^d_{\varepsilon+\varepsilon'}+N_1^{1-\delta}\notag\\
    \leq &N_1 (L^d_{\varepsilon}+Cb N_1^{-1})+N_1^{1-\delta},
\end{align}
where we used by Lemma \ref{lem:Lip_eps} that for $|\varepsilon'|\leq N_1^{-1}$, 
\[|L^d_{\varepsilon+\varepsilon'}-L^d_{\varepsilon}|\leq CbN_1^{-1}.\]
Clearly, \eqref{eq:MM_upper1} implies
\begin{align}\label{eq:MM_upper2}
    \sup_{|\varepsilon'|\leq N_1^{-1}}\log \|\textstyle{\bigwedge}^d \left(M_{N_1,E}(\theta+N_1\omega+\kappa+i\varepsilon+i\varepsilon') M_{N_1,E}(\theta+i\varepsilon+i\varepsilon')\right)\|
    \leq 2N_1 L^d_{\varepsilon}+3N_1^{1-\delta}.
\end{align}
We then have an upper bound of size $2N_1L^d_{\varepsilon}+3N_1^{1-\delta}$ for $u_{N_1,\kappa}(\theta+i\varepsilon)$ needed for the Cartan estimate.
Applying Lemma \ref{lem:Cartan} with $K_1=2N_1L^d_{\varepsilon}+3N_1^{1-\delta}$, $m_1=2N_1L^d_{\varepsilon}-3N_1^{1-\delta}$ and $H_1=N_1^{\delta/2}$ (we actually cover $\T^b$ by polydisks with radius $N_1^{-1}$, and apply Cartan to each of the polydisk. Note that $N_1^{-1}>e^{-N_0^{\delta}/2}$ hence within each such polydisk there is a lower bound $m_1$), we have the following.
\begin{lemma}\label{eq:LDT_M2_N1}
For any $|\kappa|\leq \kappa_0$ and any $N_1$ such that $4N_0<N_1\leq e^{N_0^{\delta/2}}$, the following holds
\begin{align}
    \mathrm{mes}\left(\left\{\theta\in \T^b: u_{N_1,\kappa}(\theta+i\varepsilon)<2N_1L^d_{\varepsilon}-C_b N_1^{1-\frac{\delta}{2}}\right\}\right)<C_b e^{-N_1^{\delta/2}},
\end{align}
for some constant $C_b>0$ depending on $b$ only.
\end{lemma}
Next, we perform the induction scheme. 
For $j\geq 1$, we fix an $N_j$ such that $N_j\in [e^{N_{j-1}^{\delta/4}}/4, e^{N_{j-1}^{\delta/4}}]$, and consider an arbitrary $N_{j+1}\in [4N_j, e^{N_j^{\delta/4}}]$.
By decomposing monodromy matrices of size $2N_{j+1}$ into blocks of sizes $\simeq N_j$, and arguing as above, we have a lower bound of size $m_{j+1}=2N_{j+1}L^d_{\varepsilon}-3N_{j+1}^{1-\delta}$ for 
\[ u_{N_{j+1},\kappa}(\theta+i\varepsilon):=\log \|\textstyle{\bigwedge}^d \left(M_{N_{j+1},E}(\theta+N_{j+1}\omega+\kappa+i\varepsilon) M_{N_{j+1},E}(\theta+i\varepsilon)\right)\|,\]
up to a set of measure $e^{-N_j^{\delta/2}/2}$.
We also have the upper bound
\begin{align}
\sup_{|\varepsilon'|\leq N_{j+1}^{-1}} u_{N_{j+1},\kappa}(\theta+i\varepsilon+i\varepsilon')\leq 2N_{j+1}L^d_{\varepsilon}+3N_{j+1}^{1-\delta}=:K_{j+1}.
\end{align}
Cartan's estimate (Lemma \ref{lem:Cartan}) with $H_{j+1}=N_{j+1}^{\delta/2}$, therefore implies
\begin{lemma}
For any $|\kappa|\leq \kappa_0$ and any $N_{j+1}$ such that $4N_j\leq N_{j+1}\leq e^{N_j^{\delta/4}}$, the following holds
    \begin{align}
    \mathrm{mes}\left(\left\{\theta\in \T^b: u_{N_{j+1},\kappa}(\theta+i\varepsilon)<2N_{j+1}L^d_{\varepsilon}-C_b N_{j+1}^{1-\frac{\delta}{2}}\right\}\right)<C_b e^{-N_{j+1}^{\delta/2}},
\end{align}
for some constant $C_b>0$ depending on $b$ only.
\end{lemma}
Finally, taking $n$ to be sufficiently large and $\kappa_0$-admissible and $\kappa=-n\omega$ yields the claimed result for 
$M_{n,E}^2(\theta+i\varepsilon)=M_{n,E}(\theta+n\omega+\kappa+i\varepsilon)M_{n,E}(\theta+i\varepsilon)$.
\end{proof}

Next, we prove Lemma \ref{lem:deno}
\begin{proof}[Proof of Lemma \ref{lem:deno}]
    The proof is the essentially same as that of \cite[Lemma 5.13]{HS3}. We briefly sketch it below.
    Let 
    \begin{align}\label{eq:M=WDV}
    M_{n,E}(\theta+i\varepsilon)=W_n(\theta+i\varepsilon)D_n(\theta+i\varepsilon)V_n^*(\theta+i\varepsilon)
    \end{align}
    be the singular value decomposition of $M_{n,E}(\theta+i\varepsilon)$ as in \eqref{eq:Mn_singular_decom}.
    Let $\theta\in (\mathcal{B}_{n,E,\varepsilon}\cup \tilde{\mathcal{B}}_{n,E,\varepsilon})^c$.
    We first show
    \begin{lemma}\label{lem:det_v1d_w1d}
    Let $\kappa_0$ be as in Lemma \ref{lem:Msq}. For $n$ large enough, one has for $\theta\in (\tilde{\mathcal{B}}_{n,E,\varepsilon})^c$ that 
    \begin{align}
        &|\langle v_{1}^{(n)}(\theta+i\varepsilon)\textstyle{\bigwedge}\cdots \textstyle{\bigwedge} v_{d}^{(n)}(\theta+i\varepsilon), w_{1}^{(n)}(\theta+i\varepsilon)\textstyle{\bigwedge}\cdots \textstyle{\bigwedge} w_{d}^{(n)}(\theta+i\varepsilon)\rangle|\\
        =&|\det (\langle v_{j}^{(n)}(\theta+i\varepsilon), w_{k}^{(n)}(\theta+i\varepsilon)\rangle)_{1\leq j,k\leq d}|\geq e^{-6n^{1-\delta_3}}.
    \end{align}
    \end{lemma}
    \begin{proof}
Towards a contradiction, suppose that for some $\theta\in (\tilde{\mathcal{B}}_{n,E,\varepsilon})^c$ one has
\begin{align}\label{eq:assume_a1-d_small}
|\det (\langle v_{j}^{(n)}(\theta+i\varepsilon), w_{k}^{(n)}(\theta+i\varepsilon)\rangle)_{1\leq j,k\leq d}|\leq e^{-6n^{1-\delta_3}}.
\end{align}
For any $1\leq m_1<...<m_d\leq 2d$, we expand
\begin{align}\label{eq:wedge_w=wedge_v}
    w_{m_1}^{(n)}(\theta+i\varepsilon)\textstyle{\bigwedge}\cdots\textstyle{\bigwedge} w_{m_d}^{(n)}(\theta+i\varepsilon)=\sum_{1\leq j_1<...<j_d\leq 2d} a^{j_1,...,j_d}_{m_1,...,m_d}(\theta+i\varepsilon) \cdot v_{j_1}^{(n)}(\theta+i\varepsilon)\textstyle{\bigwedge}\cdots\textstyle{\bigwedge} v_{j_d}^{(n)}(\theta+i\varepsilon),
\end{align}
in which
\begin{align}
    a^{j_1,...,j_d}_{m_1,...,m_d}(\theta+i\varepsilon)=\langle w_{m_1}^{(n)}(\theta+i\varepsilon)\textstyle{\bigwedge}\cdots\textstyle{\bigwedge}  w_{m_d}^{(n)}(\theta+i\varepsilon), v_{j_1}^{(n)}(\theta+i\varepsilon)\textstyle{\bigwedge}\cdots\textstyle{\bigwedge}  v_{j_d}^{(n)}(\theta+i\varepsilon)\rangle.
\end{align}
For arbitrary $1\leq m_1<...<m_d\leq 2d$ and $1\leq j_1<...<j_d\leq 2d$, there is the trivial bound
\begin{align}\label{eq:aj1-jd_trivial}
    |a^{j_1,...,j_d}_{m_1,...,m_d}(\theta+i\varepsilon)|\leq 1.
\end{align}
By assumption \eqref{eq:assume_a1-d_small}, one has
\begin{align}\label{eq:a1-d}
    |a_{1,...,d}^{1,...,d}(\theta+i\varepsilon)|\leq e^{-6n^{1-\delta_3}}.
\end{align}

By Lemmas \ref{lem:upperbd}, \ref{lem:Ln-L}, and estimates \eqref{eq:prod_sigma_1}, \eqref{eq:prod_sigma_2}, \eqref{eq:prod_sigma_3}, \eqref{eq:aj1-jd_trivial}, \eqref{eq:a1-d}, one has  for $n$ large enough,
\begin{align}\label{eq:M2_v1-d}
    \qquad &\|(\textstyle{\bigwedge}^d M_{n,E}^2(\theta+i\varepsilon)) (v_{1}^{(n)}(\theta+i\varepsilon)\textstyle{\bigwedge} \cdots \textstyle{\bigwedge} v_{d}^{(n)}(\theta+i\varepsilon))\|\notag\\
=&\prod_{j=1}^d \sigma_j(M_{n,E}(\theta+i\varepsilon))\cdot  \|(\textstyle{\bigwedge}^d M_{n,E}(\theta+i\varepsilon) (w_{1}^{(n)}(\theta+i\varepsilon))\textstyle{\bigwedge}\cdots \textstyle{\bigwedge} w_{d}^{(n)}(\theta+i\varepsilon))\| \notag\\
\leq &e^{n(L^d_{\varepsilon}+2n^{-\delta})} \sum_{1\leq j_1<...<j_d\leq 2d} |a^{j_1,...,j_d}_{1,...,d}(\theta+i\varepsilon)|\cdot \prod_{\ell=1}^d \sigma_{j_{\ell}}(M_{n,E}(\theta+i\varepsilon))\notag\\
=&e^{n(L^d_{\varepsilon}+2n^{-\delta})} \cdot \left(|a_{1,...,d}^{1,...,d}(\theta+i\varepsilon)|\cdot \prod_{j=1}^d\sigma_j(M_{n,E}(\theta+i\varepsilon))\right.\\
&\qquad\qquad\qquad\left.+\sum_{(j_1,...,j_d)\neq (1,...,d)} |a^{j_1,...,j_d}_{1,...,d}(\theta+i\varepsilon)|\cdot \prod_{\ell=1}^d\sigma_{j_{\ell}}(M_{n,E}(\theta+i\varepsilon))\right)\notag\\
\leq &e^{n(L^d_{\varepsilon}+2n^{-\delta})} \cdot \left(e^{n(L^d_{\varepsilon}-4n^{-\delta_3})}+C_d \cdot e^{n(L^d_{\varepsilon}-\frac{1}{4}L_{d,\varepsilon})}\right)\notag\\
\leq &e^{n(L^d_{\varepsilon}-n^{-\delta_3})}.
\end{align}
For any $(m_1,...,m_d)\neq (1,...,d)$, we have similarly,
\begin{align}\label{eq:M2_vm1-md}
        &\|(\textstyle{\bigwedge}^d M_{n,E}^2(\theta+i\varepsilon)) (v_{m_1}^{(n)}(\theta+i\varepsilon)\textstyle{\bigwedge} \cdots \textstyle{\bigwedge} v_{m_d}^{(n)}(\theta+i\varepsilon))\|\notag\\
=&\prod_{\ell=1}^d \sigma_{m_{\ell}}(M_{n,E}(\theta+i\varepsilon)) \cdot \|(\textstyle{\bigwedge}^d M_{n,E}(\theta+i\varepsilon) (w_{m_1}^{(n)}(\theta+i\varepsilon))\textstyle{\bigwedge}\cdots \textstyle{\bigwedge} w_{m_d}^{(n)}(\theta+i\varepsilon))\| \notag\\
\leq &e^{n(L^d_{\varepsilon}-\frac{1}{4}L_{d,\varepsilon})} \sum_{1\leq j_1<...<j_d\leq 2d} |a^{j_1,...,j_d}_{m_1,...,m_d}(\theta+i\varepsilon)|\cdot \prod_{\ell=1}^d \sigma_{j_{\ell}}(M_{n,E}(\theta+i\varepsilon))\notag\\
=&e^{n(L^d_{\varepsilon}-\frac{1}{4}L_{d,\varepsilon})} \cdot \left(\prod_{j=1}^d\sigma_j(M_{n,E}(\theta+i\varepsilon))+\sum_{(j_1,...,j_d)\neq (1,...,d)} \prod_{\ell=1}^d\sigma_{j_{\ell}}(M_{n,E}(\theta+i\varepsilon))\right)\notag\\
\leq &e^{n(L^d_{\varepsilon}-\frac{1}{4}L_{d,\varepsilon})} \cdot \left(e^{n(L^d_{\varepsilon}+2n^{-\delta})}+C_d \cdot e^{n(L^d_{\varepsilon}-\frac{1}{4}L_{d,\varepsilon})}\right)\notag\\
\leq &e^{n(L^d_{\varepsilon}-\frac{1}{8}L_{d,\varepsilon})}.
\end{align}
Combining 
\eqref{eq:M2_v1-d} with \eqref{eq:M2_vm1-md}, we arrive at a contradiction with $\theta\in (\tilde{\mathcal{B}}_{n,E,\varepsilon})^c$.
\end{proof}
Lemma \ref{lem:det_v1d_w1d} implies 
\begin{corollary}\label{cor:det_vd+1_2d_w}
    Let $\kappa_0$ be as in Lemma \ref{lem:Msq}. For $n$ large enough, we have for $\theta\in (\mathcal{B}_{n,E,\varepsilon}\cup \tilde{\mathcal{B}}_{n,E,\varepsilon})^c$ that 
    \begin{align}
    &|\langle v_{d+1}^{(n)}(\theta+i\varepsilon)\textstyle{\bigwedge}\cdots \textstyle{\bigwedge} v_{2d}^{(n)}(\theta+i\varepsilon), w_{d+1}^{(n)}(\theta+i\varepsilon)\textstyle{\bigwedge}\cdots \textstyle{\bigwedge} w_{2d}^{(n)}(\theta+i\varepsilon)\rangle|\\
    =&|\det (\langle v_{d+j}^{(n)}(\theta+i\varepsilon), w_{d+k}^{(n)}(\theta+i\varepsilon)\rangle)_{1\leq j,k\leq d}|\geq e^{-6n^{1-\delta_3}}.
\end{align}
\end{corollary}
Since the proof is exactly the same as that of \cite[Corollary 5.12]{HS3}, we shall not repeat it here.

Now we are in position to complete the proof of Lemma \ref{lem:deno}.

The singular value decomposition of $M_{n,E}(\theta+i\varepsilon)$ yields
\begin{align}\label{eq:M=DWV}
    |\det(M_{n,E}(\theta+i\varepsilon)-I_{2d})|=|\det (D_n(\theta+i\varepsilon)-W_{n}^*(\theta+i\varepsilon)V_{n}(\theta+i\varepsilon))|.
\end{align}
Note that
\begin{align}\label{eq:WV=Q}
    W_{n}^*(\theta+i\varepsilon)V_{n}(\theta+i\varepsilon)
    =&\left(\begin{matrix} (w_{1}^{(n)})^*(\theta+i\varepsilon)\\ \vdots\\ (w_{2d}^{(n)})^*(\theta+i\varepsilon)\end{matrix}\right)\cdot (v_{1}^{(n)}(\theta+i\varepsilon),...,v_{2d}^{(n)}(\theta+i\varepsilon))\notag\\
    =&
    \left(\langle w_{j}^{(n)}(\theta+i\varepsilon), v_{k}^{(n)}(\theta+i\varepsilon)\rangle\right)_{1\leq j,k\leq 2d}
    =:Q_n(\theta+i\varepsilon).
\end{align}
Hence, 
\begin{align}\label{eq:D-UU}
\qquad &|\det(D_n(\theta+i\varepsilon)-W_{n}^*(\theta+i\varepsilon)V_{n}(\theta+i\varepsilon))|\\
\leq &|\det(-Q_n(\theta+i\varepsilon))|\\
&+\sum_{k=1}^{2d}\,\, \sum_{1\leq j_1<...<j_k\leq 2d}\,\, \prod_{\ell=1}^k \sigma_{j_{\ell}}(M_{n,E}(\theta+i\varepsilon))\cdot \left|\det (Q_n(\theta+i\varepsilon))_{\{1,...,2d\}\setminus \{j_1,...,j_k\}, \{1,...,2d\}\setminus \{j_1,...,j_k\}}\right|,
\end{align}
in which, for a matrix $M$, $M_{B_1,B_2}$ refers to the submatrix with row numbers in the set $B_1$ and column numbers taken from the set $B_2$.
By Hadamard's inequality,
\begin{align}\label{eq:detQ_0}
|\det Q_n(\theta+i\varepsilon)|\leq \prod_{\ell=1}^{2d} \|(Q_n(\theta+i\varepsilon))_{\ell}\|\leq 1,
\end{align}
and
\begin{align}\label{eq:detQ_1}
    \left|\det (Q_n(\theta+i\varepsilon))_{\{1,...,2d\}\setminus \{j_1,...,j_k\}, \{1,...,2d\}\setminus \{j_1,...,j_k\}}\right|\leq \prod_{\ell\notin \{j_1,..,j_k\}} \| (Q_n(\theta+i\varepsilon))_{\ell}\|\leq 1
\end{align}
where $(Q_n(\theta+i\varepsilon))_{\ell}$ refers to the $\ell$-th column of $Q_n(\theta+i\varepsilon)$. Moreover,   we used the following bound
\begin{align}\label{eq:Q<1}
    \| (Q_n(\theta+i\varepsilon))_{\ell}\|\leq \|v_{\ell}^{(n)}(\theta+i\varepsilon)\|=1.
\end{align}
Corollary \ref{cor:det_vd+1_2d_w} implies that 
\begin{align}\label{eq:detQ_2}
    &\left|\det (Q_n(\theta+i\varepsilon))_{\{d+1,...,2d\}, \{d+1,...,2d\}}\right|\\
=& \left|\langle w_{d+1}^{(n)}(\theta+i\varepsilon)\textstyle{\bigwedge}\cdots \textstyle{\bigwedge} w_{2d}^{(n)}(\theta+i\varepsilon), v_{d+1}^{(n)}(\theta+i\varepsilon)\textstyle{\bigwedge} \cdots \textstyle{\bigwedge} v_{2d}^{(n)}(\theta+i\varepsilon)\rangle\right|
\geq e^{-6n^{1-\delta_3}}.
\end{align}
Combining the estimates \eqref{eq:detQ_0}, \eqref{eq:detQ_1}, \eqref{eq:detQ_2}, \eqref{eq:D-UU} with \eqref{eq:prod_sigma_1}, \eqref{eq:prod_sigma_2} and \eqref{eq:prod_sigma_3}, we infer that  for $\theta\in (\tilde{\mathcal{B}}_{n,E,\varepsilon})^c\cap (\mathcal{B}_{n,E,\varepsilon})^c$ 
\begin{align}\label{eq:detM-I}
    |\det( M_{n,E}(\theta+i\varepsilon)-I_{2d})|
\geq &\prod_{j=1}^d \sigma_j(M_n(\theta+i\varepsilon)) \cdot e^{-6n^{1-\delta_3}}-C_d\, \sup_{k=0,\ldots,2d}\, \sup_{\substack{1\leq j_1<...<j_k\leq 2d\\ (j_1,...,j_k)\neq (1,...,d)}}\, \prod_{\ell=1}^k \sigma_{j_{\ell}}(M_{n,E}(\theta+i\varepsilon))
\\ 
\geq &e^{n(L^d_{\varepsilon}-7n^{-\delta_3})}-C_d\, e^{n(L^d_{\varepsilon}-\frac{1}{4}L_{d,\varepsilon})}\\
\geq &e^{n(L^d_{\varepsilon}-8n^{-\delta_3})}.
\end{align}
The claimed result follows from combining the above with the following large deviation estimate 
\[\mathrm{mes}\left\{\theta\in \T^b: \left|\frac{1}{n}\sum_{j=0}^{n-1}\log |\det (B(\theta+j\omega+i\varepsilon))|-\langle \log |\det B(\cdot+i\varepsilon)|\rangle\right|>n^{-\delta'}\right\}<e^{-n^{\delta'}}\]
which holds for some $\delta'>0$ due to $\omega\in \mathrm{DC}$.
Thus we have completed the proof of Lemma \ref{lem:deno}.
\end{proof}

Below we present an alternate, simpler proof of \eqref{eq:detM-I}.
\begin{proof}
Let $\theta\in(\calB_{n,E,\varepsilon}\cup\tilde\calB_{n,E,\varepsilon})^c$.
We estimate $|\det(M_{n,E}(\theta)-I_{2d})|$ as in \eqref{eq:M=DWV} using Schur's lemma. 
Let $D_n=D_n(\theta+i\varepsilon)$ be as in \eqref{eq:M=WDV} and $Q_n=Q_n(\theta+i\varepsilon)$ be as in \eqref{eq:WV=Q}, and write (omitting $\theta+i\varepsilon$ below for simplicity)
\begin{align}
    Q_n=
    \left(\begin{matrix} 
    Q_1 & Q_2\\
    Q_3 & Q_4
    \end{matrix}\right), \text{ and } 
    D_n=\left(\begin{matrix} D_1 & 0\\
    0 & D_4
    \end{matrix}\right).
\end{align}
where each block is of size $d\times d$.
By Schur's lemma,
\begin{align}\label{eq:Schur_0}
    |\det (D_n-Q_n)|
    =&
    |\det (D_1-Q_1)|\cdot |\det (D_4-Q_4-Q_3(D_1-Q_1)^{-1}Q_2)|\notag\\
    =&|\det (D_1-Q_1)|\cdot |\det (D_4-Q_4)|\cdot |\det (I_d-(D_4-Q_4)^{-1}Q_3(D_1-Q_1)^{-1}Q_2)|.
\end{align}
Since $\theta\in \calB_{n,E,\varepsilon}^c$, 
\begin{align}\label{eq:Schur_1}
    \prod_{j=1}^d \sigma_j(M_{n,E}(\theta+i\varepsilon))\geq e^{n(L^d_{\varepsilon}-n^{-\delta})}.
\end{align}
By Lemma \ref{lem:upperbd} and that $L^{d-1}_{(n),\varepsilon}\leq L^{d-1}_{\varepsilon}+\nu$ for arbitrary small $0<\nu\ll L_d$ and $n$ large, 
\begin{align}\label{eq:Schur_2}
    \prod_{j=1}^{d-1}\sigma_j(M_{n,E}(\theta+i\varepsilon))\leq e^{n(L^{d-1}_{\varepsilon}+\nu)}.
\end{align}
Combining \eqref{eq:Schur_1} with \eqref{eq:Schur_2}, we conclude that
\begin{align}\label{eq:Schur_22}
    \min_{j=1\ldots d} \sigma_j(M_{n,E}(\theta+i\varepsilon))\geq \sigma_d (M_{n,E}(\theta+i\varepsilon))\geq e^{n(L_{d,\varepsilon}-2\nu)}.
\end{align}
Similarly, by Lemma \ref{lem:upperbd} applied to $L^{d+1}$,  
\begin{align}
    \prod_{j=1}^{d+1}\sigma_j(M_{n,E}(\theta+i\varepsilon))\leq e^{n(L^{d+1}_{\varepsilon}+\nu)},
\end{align}
which implies, analogously to \eqref{eq:Schur_22}, that
\begin{align}\label{eq:Schur_23}
\max_{j=d+1\ldots 2d}\sigma_j(M_{n,E}(\theta+i\varepsilon))\leq \sigma_{d+1}(M_{n,E}(\theta+i\varepsilon))\leq e^{n(L_{d+1,\varepsilon}+2\nu)}.
\end{align}
Recall that by \eqref{eq:L_d>0_varepsilon}, 
\begin{align}
    L_{d+1,\varepsilon}<-\frac{1}{2}L_d<0<\frac{1}{2}L_d<L_{d,\varepsilon}.
\end{align}
Combining the above with \eqref{eq:Schur_1}, \eqref{eq:Schur_22} and  $\|Q_1\|\leq 1$ (see \eqref{eq:Q<1}), yields
\begin{align}\label{eq:Schur_3}
    |\det (D_1-Q_1)|\geq e^{n(L^d_{\varepsilon}-2n^{-\delta})}, \text{ and } \|(D_1-Q_1)^{-1}\|\leq e^{-n(L_{d,\varepsilon}-3\nu)}.
\end{align}
Recall that by Corollary \ref{cor:det_vd+1_2d_w}, $|\det Q_4|\geq e^{-6n^{1-\delta_3}}$, and hence $\|Q_4^{-1}\|\leq C_d e^{6n^{1-\delta_3}}$. Since by \eqref{eq:Schur_23} that $\|D_4\|\leq e^{n(L_{d+1,\varepsilon}+2\nu)}$, we have
\begin{align}\label{eq:Schur_4}
    |\det (D_4-Q_4)|\geq e^{-7n^{1-\delta_3}}, \text{ and } \|(D_4-Q_4)^{-1}\|\leq C_d e^{7n^{1-\delta_3}}.
\end{align}
Combining \eqref{eq:Schur_3}, \eqref{eq:Schur_4} with  $\|Q_2\|+\|Q_3\|\leq 1$, we infer by \eqref{eq:Schur_0},
\begin{align}
    |\det(D_n-Q_n)|\geq e^{n(L^d_{\varepsilon}-10 n^{-\delta_3})},
\end{align}
as claimed.
\end{proof}

Next, we show Lemma \ref{lem:fn_ave} as a corollary of 
Lemmas \ref{lem:deno}, \ref{lem:fn_upper} combined with the Cartan estimate.
\begin{proof}[Proof of Lemma \ref{lem:fn_ave}]
    By Lemma \ref{lem:Lip_eps}, 
    \[\sup_{|\varepsilon'|\leq n^{-1}}L^d_{\varepsilon+\varepsilon'}\leq L^d_{\varepsilon}+C_b n^{-1},\]
    and similarly 
    \[\sup_{|\varepsilon'|\leq n^{-1}}\langle \log|\det B(\cdot+i\varepsilon+i\varepsilon')|\rangle \leq \langle \log|\det B(\cdot+i\varepsilon)|\rangle+C_b n^{-1}.\]
    Combining the above with Lemma \ref{lem:fn_upper}, one has 
    \begin{align}
        \sup_{|\varepsilon'|\leq n^{-1}}\frac{1}{n}\log |f_{E,n}(\theta+i\varepsilon+i\varepsilon')|\leq L^d_{\varepsilon}+\langle \log|\det B(\cdot+i\varepsilon)|\rangle+C_b n^{-\delta}.
    \end{align}
Covering the torus $\T^b$ by polydisks of radii $\simeq n^{-1}$, and apply Cartan's estimate (Lemma \ref{lem:Cartan}) to each  polydisk with upper bound $K=L^d_{\varepsilon}+\langle \log|B(\cdot+i\varepsilon)|\rangle+C_b n^{-\delta}$, lower bound $m=L^d_{\varepsilon}+\langle \log |B(\cdot+i\varepsilon)|\rangle -n^{-\delta_1}$ (note the lower bound is provided by Lemma~\ref{lem:deno}), and $H_k=2^k n^{\delta_4}$ with $\delta_4=\min(\delta,\delta_1)/2$, yields
\begin{align}
    \mathrm{mes}\left(\left\{\theta\in \T^b: \frac{1}{n}\log |f_{E,n}(\theta+i\varepsilon)|<L^d_{\varepsilon}+\langle \log |\det B(\cdot+i\varepsilon)|\rangle-C_b 2^k n^{-\delta_4}\right\}\right)<C_b e^{-2^k n^{\delta_4}},
\end{align}
for any positive integer $k$. 
This clearly implies the claimed result by summing in $k$ and with a slightly smaller $\delta_4$.
\end{proof}

\section{Localization}

\subsection{Proof of Theorem~\ref{thm:main}: nonarithmetic localization}\label{sec:main}

The following result does not require $n$ to be  $\kappa_0$-admissible. 
\begin{corollary}
    \label{cor:goodGreen}
    Let $\omega\in DC_{a,A}$ and assume $L_d(E,M_E)\geq \gamma>0$. For $n\ge n_0(a,A,b, \gamma, B,V)$, 
    there exists a set $\calS_n(E,\omega)\subset\tor^b$ with the property that $\mes(\calS_n(E,\omega))<\exp(-n^\delta)$ with $\delta=\delta(a,A)>0$ and 
    \begin{equation}
        \label{eq:GregE}
        |G_{E,n}(\theta;x,y)|\le \exp(-|x-y|\cdot L_d+n^{1-\delta})
    \end{equation}
    for all $x,y\in [0,dn-1]$  and all $\theta\in\tor^b\setminus\calS_n(E,\omega)$. 
\end{corollary}
\begin{proof}
First consider the case  $0\le x \le d-1$ or $(n-1)d\leq x\leq nd-1$, and $n$ is $\kappa_0$-admissible. 
    Then the claim follows from Lemmas~\ref{lem:numerator} and~\ref{lem:deno}, together with the elementary fact
    \[
    \log \prod_{j=0}^{n-1}|\det B_j(\theta)| = n \langle \log |\det B|\rangle + O(n^{1-\delta})
    \]
    which can be proved via Fourier series. To remove the restriction that $x$ is located near the edges of $[0,dn-1]$, as well as the admissibility condition, we pave $[0,dn-1]$ with intervals of sizes about $n^{\frac12
    }$, say, which are $\kappa_0$-admissible. Iterating the resolvent identity with these smaller intervals as in~\cite{BGS}*{Lemma 2.2} then yields the desired statement, with possibly smaller $\delta$. 
\end{proof}

The localization theorem can now be proved via the well-known Bourgain, Goldstein strategy~\cite{BG} and~\cite{B}*{Chapter 10}. The main steps are as follows:

\begin{enumerate}
    \item[(a)] Eliminate the energy. This involves a quantitative Seidenberg-Tarski theorem in the theory of semi-algebraic sets. 
    \item[(b)] Eliminate $\omega\in\tor^b$ that may lead to double resonances, uniformly in the energy. 
    \item[(c)]  Starting from Shnol's theorem, and using absence of double resonances, show exponential localization. 
\end{enumerate}

The key step (b) is based on the idea that $\omega$ (the ``slow variable") and $N\omega$ (the ``fast variable") resemble independent random variables in~$\tor$, but only when measured on sets consisting of $m\ll N$ intervals. The latter condition is the  reason that semi-algebraic sets enter into the analysis in a crucial and quantitative fashion. 
We will rely on~\cite{B}*{Chapter 9} for the technical statements about semi-algebraic theory, in particular on Corollary~9.7 and Lemma~9.9 in loc.\ cit. However, we would like to point out that the  Yomdin parametrization of Theorem~9.4 with quantitative polynomial bounds, which is needed in these results, was erroneously attributed to Gromov's work during the mid 1980s. In fact, these essential polynomial bounds were finally established by Binyamini and Novikov~\cite{BN}, thus closing a gap in~\cite{B} that had remained open for almost 15 years. 

To begin the localization proof, we start from a solution $H_0\Phi=E\Phi$ with $|\Phi_0|+|\Phi_1|>0$ and $|\Phi_n|\le C_\Phi(1+|n|)$ for all~$n\in\Z$. The latter is the folklore Shnol's theorem, see e.g.~\cite{S2}*{Theorem~2.1} for a self-contained derivation.   Arguing as in~\cite{B}*{(10.7), (10.8)} we see that we may cover the set $\calS_n(E,\omega)$ by a semi-algebraic set $\wt\calS_n(E,\omega)$ of polynomial degree $O(n^{p})$ in all the variables $(\theta, E,\omega)\in \tor^b\times\R\times DC_{A,a}$, and  of similar measure. In this step, one uses that the Diophantine condition $\|k\cdot\omega\|\ge a|k|^{-A}$ is only needed in the range $0<|k|\le n^{C_A}$ for Lemma~\ref{lem:deno} to hold (see~\cite{BG}*{page 859}). 
We denote this finite volume Diophantine condition by~$DC_{a,A}(n)$.

It follows from~\cite{B}*{Corollary 9.7, (10.12), (10.13)} that for any large $N_0$ there exists $j_0\in [1, N_0^{C_1}]$ so that with $I_0=[-j_0,j_0]$
\begin{equation}
    \label{eq:spec stab}
    \dist(E,\sigma(H_{I_0}(0,\omega)) ) < e^{-\gamma' N_0}
\end{equation}
where $0<\gamma'<\gamma$ and $C_1$ is some large constant depending on~$\omega$ through the Diophantine parameter $(a,A)$. Condition~\eqref{eq:spec stab} is what we mean by a simple resonance at energy~$E$.  Step (a) above begins by defining 
\[
\calE_\omega := \bigcup_{j_0\in [1, N_0^{C_1}]} \sigma(H_{[-j_0,j_0]}(0,\omega))
\]
Accordingly, we set $\wt\calS_{N_0}(\omega):=\bigcup_{E'\in\calE_\omega} \wt\calS_{N_0}(E',\omega)$, which is semi-algebraic of some degree~$N_0^p$ over all variables, and  the measure of any fixed $(E,\omega)$ slice as a set of $\theta$ alone is at most $e^{-N_0^\delta}$ for some $\delta>0$. Next, one introduces 
\begin{equation}
    \label{eq:2res set}
    \Omega_{N_0} := \{ (\theta,\omega)\in \tor^b\times DC_{a,A}(N_0)\::\: E\in\calE_\omega, \; \theta \in\calS_{N_0}(E,\omega)\}. 
\end{equation}
It is a deep fact that $\Omega_{N_0}$ is again semi-algebraic of polynomial degree in~$N_0$. Indeed, this requires the quantitative Seidenberg-Tarski theorem~\cite{BPR}, \cite{B}*{Proposition~9.2} and concludes Step~(a) above. 

To carry out Steps (b) and~(c) means excluding those $\omega\in  DC_{a,A}$ which have the property that $(\omega, k\omega)\in \Omega_{N_0}$ for some $N_0^{C_2}\le |k|\le  N_0^{2C_2}$ where $C_2>C_1$. This is the method of ``steep lines" (or rather in this setting, steep planes) from~\cite{BG}. Finally, one takes a limsup of sets as $N_0\to\infty$ to arrive at a null set of bad~$\omega$. For good $\omega$, Step~(c), and therefore Anderson localization of $\Phi$ at energy~$E$, are established by paving of the set $[N_0^{C_2}, N_0^{2C_2}]\cup [-N_0^{2C_2}, -N_0^{C_2}] $ with good $N_0$-intervals and iterating the resolvent identity (one needs to use the resonant condition~\eqref{eq:spec stab} here). This is standard. 

In contrast,  the steep planes argument is not and by the preceding relies  on the polynomial complexity bounds of~\cite{BN} via~\cite{B}*{Theorem 9.4, Lemma 9.9}.  We refer the reader to pages~59, 60 of Bourgain's book for the details, which apply here verbatim.

\subsection{Arithmetic Anderson localization}\label{sec:AL_arithmetic}

This section restricts to the case of a one-dimensional torus $\T$ and constant $B(\theta)\equiv B$.
We follow the strategy of \cite{HS2} to show the large deviation set $\mathcal{B}_{f,E,n}:=\mathcal{B}_{f,E,n,0}$ as in \eqref{def:B_fEn} can be covered by roughly $2\kappa^d(\omega,M_E)\cdot n$ many intervals of small length, and then use a Pigeon-hole principle argument to eliminate double resonances.
The number $2\kappa^d(\omega,M_E)\cdot n$ arises from the number of zeros of $f_{E,n}(z)$, with the identification $z=e^{2\pi i\theta}$, in a thin annulus containing $\mathcal{C}_1$.

The zeros of $f_{E,n}(z)$ off the unit circle form pairs, as can be seen from the following fact. 

\begin{lemma}\label{lem:otf=tf}
    For any $n$ and $\theta\in \T$, $f_{E,n}(\theta)=\overline{f_{E,n}(\theta)}$. 
\end{lemma}
\begin{proof}
\begin{align}
&\overline{f_{E,n}(\theta)}\\
&=\det
\left(\begin{matrix}
V(\theta+(n-1)\omega)-E & B^* & &  &B\\
B &V(\theta+(n-2)\omega)-E &\ddots \\
& \ddots &\ddots &\ddots \\
& &\ddots &\ddots &B^*\\
B^* & & &B &V(\theta)-E
\end{matrix}\right)^*\\
&=\det
\left(\begin{matrix}
V^*(\theta+(n-1)\omega)-E & B^* & &  &B\\
B &V^*(\theta+(n-2)\omega)-E &\ddots \\
& \ddots &\ddots &\ddots \\
& &\ddots &\ddots &B^*\\
B^* & & &B &V^*(\theta)-E
\end{matrix}\right)\\
&=f_{E,n}(\theta),
\end{align}
where we used $V^*=V$.
\end{proof}
Lemma \ref{lem:otf=tf} implies the two holomorphic function $f_{E,n}(z)$ and $\overline{f_{E,n}(1/\overline{z})}$ coincide on the unit circle $z\in \mathcal{C}_1$.
Hence $f_{E,n}(z)=\overline{f_{E,n}(1/\overline{z})}$ on $\mathcal{A}_{e^{2\pi\eta}}$.
This implies 
\begin{fact}\label{fact:fn_real_zero}
    If $w\in \mathcal{A}_{e^{2\pi\eta}}\setminus \mathcal{C}_1$ is a zero of $f_{E,n}(z)$, then $1/\overline{w}$ is also a zero.
\end{fact}

Next, we recall the Green's function on an annulus, which is standard. 
\begin{lemma}\cite[Lemma 3.1]{HS2}\label{lem:Green_AR}
The Green's function on the annulus $\mathcal{A}_R$ satisfies
\begin{align} 
G_R(z,w) &= \frac{1}{2\pi}\log |z-w| + H_R(z,w), \quad z\in \mathcal{A}_R, w\in \overline{\mathcal{A}_R} \label{def:Green_annulus} \\
\Delta_z H_R (z,w) &= 0 \nonumber
\end{align}
The Green's function is symmetric and invariant under rotations: $G_R(z,w)=G_R(w,z)$ and $G_R(z,w)=G_R(e^{i\phi} z,e^{i\phi} w)$, for any $\phi\in \R$.
\end{lemma}
We also recall the integral of the Green's function along a circle.
\begin{lemma}\cite[Lemma 3.2]{HS2}\label{lem:int_Green}
For $1/R\leq r\leq R$ and $w\in \mathcal{A}_R$, we have
\begin{align}\label{eq:int_Green}
I(\log r, \log R, w) &:=2\pi \int_0^{1}  G_R (re^{ 2\pi i\theta}, w)\, \mathrm{d}\theta \\
& =(2\log R)^{-1}
\begin{cases}
 \log (rR) \log|w/R|, \text{ if } |w|\geq r\\
\\
\log( r/ R) \log |w R|, \text{ if } |w|<r.
\end{cases}
\end{align}
\end{lemma}

We now turn to the basic Riesz representation of subharmonic functions. 
\begin{lemma}\cite[Lemma 3.3]{HS2}\label{lem:Riesz}
Let $v$ be a subharmonic function in a neighborhood of $\overline{\mathcal{A}_R}$, and assume $v|_{\partial \mathcal{A}_R}$ is a continuous function.
Let $G_R$ be the Green's function for $\mathcal{A}_R$, as  in \eqref{def:Green_annulus}.
There exists a positive finite measure $\mu$ on $\mathcal{A}_R$, and a harmonic function $h_R$ on $\mathcal{A}_R$, such that
\begin{align*}
v(w)=\int_{\mathcal{A}_R} 2\pi G_R(z,w)\, \mu(\mathrm{d}z)+h_R(w),
\end{align*}
where 
\begin{align}\label{eq:Poisson}
h_R(w)=\int_{\partial \mathcal{A}_R} v(z)\,\nu(w,\mathcal{A}_R)( \mathrm{d}z),
\end{align}
where $\nu(w,\mathcal{A}_R)$ is the harmonic measure of $\mathcal{A}_R$ with pole at $w$.
In particular, 
\begin{align}\label{eq:h=v}
h_R(z)=v(z), \text{ for } z\in \partial \mathcal{A}_R.
\end{align}
\end{lemma}
\begin{remark}
By the maximum principle, 
\begin{align}\label{eq:max_h}
\sup_{w\in \mathcal{A}_R} h_R(w)\leq \max_{z\in \partial \mathcal{A}_R} v(z).
\end{align}
\end{remark}

Recall that $f_{E,n}(z)$ is a holomorphic function in $\mathcal{A}_{e^{2\pi \eta}}$. 
For $0\leq \varepsilon\leq \eta$, let $$N_n(E,\varepsilon):=\#\{z\in \overline{\mathcal{A}_{e^{2\pi \varepsilon}}}:\, f_{E,n}(z)=0\}.$$
We have the following analogue of \cite[Theorem 4.4]{HS2}. 
Recall that for $b=1$, we shrank $\eta$ such that \eqref{eq:L_linear} holds and Lemmas \ref{lem:deno} and \ref{lem:fn_ave} hold for $|\varepsilon|\eta$. We may further shrink $\eta$ such that $f_{E,n}(z)$ is zero-free on $\partial \mathcal{A}_{e^{2\pi\eta}}$.
\begin{theorem}\label{thm:Riesz_un}
Let $E\in \R$ be such that $L_d(E,M_E)\geq \gamma>0$. 
Let $R=e^{2\pi\eta}$ and $w_1,...,w_{N_n(E,\eta)}$ be the zeros of $f_{E,n}(z)$ in $\mathcal{A}_{R}$
and  define 
\begin{align*}
G_{R,n}(z,E)=\frac{1}{n}\sum_{k=1}^{N_n(E,\eta)} G_R(z, w_k),
\end{align*}
where $G_R$ is the Green's function in \eqref{def:Green_annulus}. Then 
\begin{align}\label{eq:f=G+h}
f_{E,n}(z)=2\pi G_{R,n}(z,E)+h_{R,n}(z,E),
\end{align}
where the harmonic part satisfies $h_{R,n}(z,E)=f_{E,n}(z)$ on $\partial \mathcal{A}_R$.
Furthermore, let $\delta_2>0$ be as in Lemma \ref{lem:fn_ave}. Then 
\begin{itemize}
\item for $n$ large enough, for any $z\in \mathcal{A}_{r}$, $1\leq r<R$, that
\begin{align}\label{eq:h_constant_un}
L^d_{\eta}(\omega,M_E)-\frac{C}{R-r} n^{-\delta_2} \leq  h_{R,n}(z,E)\leq L^d_{\eta}(\omega,M_E)+n^{-\delta_2}.
\end{align}
\item for $n$ large enough,
\begin{align}\label{eq:N_est_un}
|\frac{1}{2n} N_n(E, \eta/3)-\kappa^d(\omega,M_E)|\leq C\eta^{-2} n^{-\delta_2}.
\end{align}
In particular, for any $\nu\in (0,1/10)$, for $n$ large enough,
\begin{align}\label{eq:N_est_un_kappa=d}
N_n(E, \eta/3)\leq 2n(\kappa^d(\omega,M_E)+\nu).
\end{align}
\end{itemize}
\end{theorem}
\begin{proof}
In the proof we shall omit the dependence of various parameters on $\omega, E$ for simplicity. We shall also write $b_{\varepsilon}:=\langle \log |B(\cdot+i\varepsilon)|\rangle$.
First, we estimate the harmonic part.
Note that the harmonic part satisfies $h_{R,n}=f_n$ on $\partial A_{R}$.
By Lemma \ref{lem:fn_upper}, one has that for $r=R$ or $1/R$ and $n$ large, uniformly in $\theta$,
\begin{align*}
h_{R,n}(re^{2\pi i\theta})=f_n(re^{2\pi i\theta})\leq L^d_{\eta}+b_{\eta}+n^{-\delta}.
\end{align*}
Hence by the maximum principle \eqref{eq:max_h},  
\begin{align}\label{eq:upper_h}
h_{R,n}(z)\leq L^d_{\eta}+b_{\eta}+n^{-\delta}, \text{ for } z\in \overline{A_{R}}.
\end{align}
We also have by Lemma \ref{lem:fn_ave} that for $n$ large enough,
\begin{align}\label{eq:int_h_lower}
\int_0^1 h_{R,n}(R e^{2\pi i\theta})\, \mathrm{d}\theta=\int_0^{1} f_n(R e^{2\pi i\theta})\, \mathrm{d}\theta\geq L^d_{\eta}+b_{\eta}-n^{-\delta_2}.
\end{align}
Let 
\begin{align}\label{def:tilde_h}
\tilde{h}_{R,n}(z):=L^d_{\eta}+b_{\eta}+n^{-\delta}-h_{R,n}(z)\geq 0,
\end{align} 
where we invoked \eqref{eq:upper_h}.
In view of \eqref{eq:int_h_lower} and Lemma \ref{lem:Ln-L},  for $n$ large,
\begin{align}\label{eq:int_th_lower}
\int_0^{1} \tilde{h}_{R,n}(R e^{2\pi i\theta})\, \mathrm{d}\theta\leq C n^{-\delta_2}.
\end{align}
By \eqref{eq:Poisson} and \eqref{eq:int_th_lower}, and the well-known estimate on the harmonic measure
$$0\leq \frac{\mathrm{d}\nu(w, A_R)(z)}{\mathrm{d}\sigma(z)}\leq C(\mathrm{dist}(w,\partial A_R))^{-1},$$ 
with arclength measure $\sigma$, 
one has that for $z\in A_r$, with $1\leq r<R$,
\begin{align*}
0\leq \tilde{h}_{R,n}(z)\leq \frac{C}{R-r} \int_{0}^{1} \tilde{h}_{R,n}(Re^{2\pi i\theta})\, \mathrm{d}\theta\leq \frac{C}{R-r} n^{-\delta_2}.
\end{align*}
This combined with \eqref{eq:upper_h} yields \eqref{eq:h_constant_un}.

Next, we evaluate the integrals of $G_{R,n}(z)$ along circles.
For $1\leq r\leq R$, let
\begin{align*}
I_n(\log r, \log R):=&\int_0^{1} 2\pi G_{R,n}(r e^{2\pi i\theta})\, d\theta\\
=&\frac{1}{n}\sum_{k=1}^{N_n(\eta)} I(\log r, \log R, w_k),
\end{align*}
where $I(\log r,\log R, w)$ is defined as in~\eqref{eq:int_Green}.
By \cite[(4.24)]{HS2},  
\begin{align}\label{eq:int_final}
I_n(\log r, \log R)=-\frac{\pi}{n} \int_{\frac{\log r}{2\pi}}^{\frac{\log R}{2\pi}} N_n(\varepsilon)\, \mathrm{d}\varepsilon.
\end{align}
Integrating \eqref{eq:f=G+h} along $z\in \mathcal{C}_{r_j}$, $1\leq r_1=e^{2\pi\varepsilon_1}<r_2=e^{2\pi\varepsilon_2}\leq R$, and combining with \eqref{eq:int_final}, one obtains 
\begin{align*}
\int_0^{1} f_n(r_j e^{2\pi i\theta})\, \mathrm{d}\theta=-\frac{\pi}{n}\int_{\varepsilon_j}^{\eta} N_n(\varepsilon)\, \mathrm{d}\varepsilon+\int_0^{1} h_{R,n}(r_je^{2\pi i\theta})\, \mathrm{d}\theta.
\end{align*}
Taking the difference of the equations above between $r_1$ and $r_2$, we arrive at 
\begin{align}\label{eq:N_0'}
\int_0^{1} f_n(r_2 e^{2\pi i\theta})\, \mathrm{d}\theta-\int_0^{1} f_n(r_1e^{2\pi i\theta})\, \mathrm{d}\theta=&\frac{\pi}{n}\int_{\varepsilon_1}^{\varepsilon_2} N_n(\varepsilon)\, \mathrm{d}\varepsilon \notag\\
&+\int_0^1 h_{R,n}(r_2 e^{2\pi i \theta})\, \mathrm{d}\theta-\int_0^1 h_{R,n}(r_1 e^{2\pi i \theta})\, \mathrm{d}\theta.
\end{align}
By Lemma \ref{lem:fn_ave}, we have for $n$ large,
\begin{align}\label{eq:N_1'}
\int_0^{1} f_n(r_j e^{2\pi i\theta})\, d\theta\geq L^d_{\varepsilon_j}+b_{\varepsilon_j}-n^{-\delta_2}.
\end{align}
while it follows from  Lemma \ref{lem:fn_upper} that
\begin{align}\label{eq:N_2'}
f_n(r_j e^{2\pi i\theta})\leq L^d_{\varepsilon_j}+b_{\varepsilon_j}+n^{-\delta}.
\end{align}
By \eqref{eq:h_constant_un}, one has for $n$ large,
\begin{align}\label{eq:N_2''}
\left|\int_0^1 h_{R,n}(r_2 e^{2\pi i \theta})\, \mathrm{d}\theta-\int_0^1 h_{R,n}(r_1 e^{2\pi i \theta})\, \mathrm{d}\theta\right|\leq \frac{C}{R-r_2} n^{-\delta_2}.
\end{align}
Hence plugging the estimates  \eqref{eq:N_1'} and \eqref{eq:N_2'}, \eqref{eq:N_2''} into \eqref{eq:N_0'}, one concludes that  for $n$ large,
\begin{align}\label{eq:N_3'}
\frac{\pi}{n}(\varepsilon_2-\varepsilon_1) N_n(\varepsilon_1)\leq \frac{\pi}{n}\int_{\varepsilon_1}^{\varepsilon_2} N_n(\varepsilon)\, \mathrm{d}\varepsilon
\leq L^d_{\varepsilon_2}-L^d_{\varepsilon_1}+\frac{C}{R-r_2} n^{-\delta_2},
\end{align}
and
\begin{align}\label{eq:N_4'}
\frac{\pi}{n}(\varepsilon_2-\varepsilon_1) N_n(\varepsilon_2)\geq \frac{\pi}{n}\int_{\varepsilon_1}^{\varepsilon_2} N_n(\varepsilon)\, \mathrm{d}\varepsilon 
\geq L^d_{\varepsilon_2}-L^d_{\varepsilon_1}-\frac{C}{R-r_2}n^{-\delta_2}.
\end{align}
Taking $r_1=\eta/3$ and $r_2=2\eta/3$ in \eqref{eq:N_3'} yields
\begin{align}\label{eq:N_5'}
\frac{\pi \eta}{3n} N_n(\eta/3)\leq L^d_{2\eta/3}-L^d_{\eta/3}+C \eta^{-1} n^{-\delta_2}.
\end{align}
Setting $\varepsilon_1=0$ and $\varepsilon_2=\eta/3$ in \eqref{eq:N_4'} yields
\begin{align}\label{eq:N_6'}
\frac{\pi\eta}{3n} N_n(\eta/3)\geq L^d_{\eta/3}-L^d_{0}-C\eta^{-1} n^{-\delta_2}.
\end{align}
Combining \eqref{eq:N_5'}, \eqref{eq:N_6'} with \eqref{eq:L_linear}, we infer that for $n$ large enough,
\begin{align}\label{eq:N_7'}
\left| \frac{1}{2n} N_n(\eta/3)-\kappa^d\right| \leq C\eta^{-2} n^{-\delta_2}.
\end{align}
This proves the claimed result.
\end{proof}

\subsection{Proof of Theorem \ref{thm:acceleration=1}}
It suffices to show that each generalized eigenfunction $u$, satisfying 
\begin{align}\label{eq:Schnol}
\max(|u_0|, |u_1|)=1, \text{ and } 
    |u_k|\leq C|k|, \text{ for all } k\neq 0,
\end{align}
decays exponentially.
Note that under the assumption that $\kappa^d(d^{-1}\omega, A_E)=1$,   by \eqref{eq:M=A}  we conclude that $\kappa^d(\omega,M_E)=d\cdot \kappa^d(d^{-1}\omega,A_E)=d$. 
Hence by Theorem \ref{thm:Riesz_un}, $N_n(E,\eta/3)\leq 2n(d+\nu)$.
We first show in this setting, there is an additional symmetry that reduce the number of zeros of $f_{E,n}(z)$ from at most $2n(d+\nu)$ to no more than $n(d+\nu)$ pairs. The following lemma is essential.

\begin{lemma}\label{lem:f_even}
For the operator $\tilde{H}_{\theta}$ as in Theorem \ref{thm:acceleration=1}, we have
\[f_{E,n}(\theta-\frac{nd-1}{2d}\omega)=\overline{f_{E,n}(-\theta-\frac{nd-1}{2d}\omega)}.\]
\end{lemma}
\begin{proof}
We first establish the following property of the matrix potential~$V$ and coefficient matrix~$B$ from~\eqref{def:tH_BV}.

\begin{lemma}\label{lem:symmetry}
For matrix $J$ as in \eqref{def:J} below, one has $JV(\theta-\frac{d-1}{2d}\omega)J^{-1}=V^T(-\theta-\frac{d-1}{2d}\omega)$ and \linebreak $JBJ^{-1}=B^T$.
\end{lemma}
\begin{proof}
In the following, we write $g(\theta,m):=g(\theta+md^{-1}\omega)$ for simplicity.
Recall that $B,V$ are as in~\eqref{def:tH_BV}. 
Let
\begin{align}\label{def:J}
    J=\left(\begin{matrix}
        & & & &1\\
        & & &1 &\\
        & &\udots & &\\
        &1 & & &\\
        1 & & & &
    \end{matrix}\right)_{d\times d}
\end{align}
We have
\begin{align}\label{eq:JVJ=V}
        JV(\theta-\frac{d-1}{2d}\omega)J^{-1} =&
        J\left(\begin{matrix}
g(\theta,\frac{d-1}{2}) &v_1 &\cdots &v_{d-2} &v_{d-1}\\
           \overline{v_1} &g(\theta,\frac{d-3}{2})&\ddots &\ddots &v_{d-2}\\
\vdots &\ddots &\ddots &\ddots &\vdots          \\
\overline{v_{d-2}} &\ddots &\ddots &g(\theta,-\frac{d-3}{2}) &v_1\\
\overline{v_{d-1}}         &\overline{v_{d-2}} &\cdots &\overline{v_1} &g(\theta,-\frac{d-1}{2})
        \end{matrix}\right)J^{-1}\notag\\
=&\left(\begin{matrix}
g(\theta,-\frac{d-1}{2}) &\overline{v_1} &\cdots &\overline{v_{d-2}} &\overline{v_{d-1}}\\
           v_1 &g(\theta,-\frac{d-3}{2}) &\ddots &\ddots &\overline{v_{d-2}}\\
\vdots &\ddots &\ddots &\ddots &\vdots          \\
v_{d-2} &\ddots &\ddots &g(\theta,\frac{d-3}{2})&\overline{v_1}\\
v_{d-1} &v_{d-2} &\cdots &v_1 &g(\theta,\frac{d-1}{2})
        \end{matrix}\right) \notag\\
=&\left(\begin{matrix}
g(-\theta,\frac{d-1}{2}) &\overline{v_1} &... &\overline{v_{d-2}} &\overline{v_{d-1}}\\
           v_1 &g(-\theta,\frac{d-3}{2}) &\ddots &\ddots &\overline{v_{d-2}}\\
\vdots &\ddots &\ddots &\ddots &\vdots          \\
v_{d-2} &\ddots &\ddots &g(-\theta,-\frac{d-3}{2})&\overline{v_1}\\
v_{d-1} &v_{d-2} &\cdots &v_1 &g(-\theta,-\frac{d-1}{2})
        \end{matrix}\right) \\
=&V^T(-\theta-\frac{d-1}{2d}\omega) \notag
\end{align}
To pass to \eqref{eq:JVJ=V} we used that $g$ is even.
It is also easy to check that 
\begin{align}
JBJ^{-1}
=\left(\begin{matrix}
\overline{v_d} & & & &\\
\overline{v_{d-1}}             &\ddots & &\\
\vdots &\ddots &\ddots & &          \\
\overline{v_2}&\ddots &\ddots &\ddots &\\
\overline{v_1} &\overline{v_2} &\cdots &\overline{v_{d-1}}  &\overline{v_d}
        \end{matrix}\right)=B^T.
\end{align}
Hence we have proved the claimed identities.
\end{proof}

In the following we write $V(\theta,k):=V(\theta+k\omega)-E$ (suppressing $E$ in the notation for simplicity).
One has 
\begin{align}
&f_{E,n}(\theta-\frac{nd-1}{2d}\omega)=f_{E,n}(\theta-\frac{n-1}{2}\omega-\frac{d-1}{2d}\omega)\\
&=
\det\left(\begin{matrix}
V(\theta,\frac{n-1}{2}-\frac{d-1}{2d}) & B^* & &  &B\\
B &V(\theta,\frac{n-3}{2}-\frac{d-1}{2d}) &\ddots \\
& \ddots &\ddots &\ddots \\
& &\ddots &\ddots &B^*\\
B^* & & &B &V(\theta,-\frac{n-1}{2}-\frac{d-1}{2d})
\end{matrix}\right)\\
    &=\det\left(\begin{matrix}
    JV(\theta,\frac{n-1}{2}-\frac{d-1}{2d})J^{-1} & JB^*J^{-1} & &  &\!\!\!\!JBJ^{-1}\\
JBJ^{-1} &\!\!\!\!JV(\theta,\frac{n-3}{2}-\frac{d-1}{2d})J^{-1} &\ddots \\
& \ddots &\ddots &\ddots \\
& &\ddots &\ddots &\!\!\!\!JB^*J^{-1}\\
JB^*J^{-1} & & &JBJ^{-1} &\!\!\!\!JV(\theta,-\frac{n-1}{2}-\frac{d-1}{2d})J^{-1} 
    \end{matrix}\right)
    \end{align}
Reordering the variables one sees that the previous line equals
    \begin{align}
  &\det\left(\begin{matrix}
    JV(\theta,-\frac{n-1}{2}-\frac{d-1}{2d})J^{-1} & JBJ^{-1} & &  &\!\!\!\! JB^*J^{-1}\\
JB^*J^{-1} &JV(\theta,-\frac{n-3}{2}-\frac{d-1}{2d})J^{-1} &\ddots \\
& \ddots &\ddots &\ddots \\
& &\ddots &\ddots &\!\!\!\! JBJ^{-1}\\
JBJ^{-1} & & &JB^*J^{-1} &\!\!\!\! JV(\theta,\frac{n-1}{2}-\frac{d-1}{2d})J^{-1} 
    \end{matrix}\right)\\
    &=\det\left(\begin{matrix}
    \overline{V^*(-\theta,\frac{n-1}{2}-\frac{d-1}{2d})} & \overline{B^*} & &  &\overline{B}\\
\overline{B} &\overline{V^*(-\theta,\frac{n-3}{2}-\frac{d-1}{2d})} &\ddots \\
& \ddots &\ddots &\ddots \\
& &\ddots &\ddots &\overline{B^*}\\
\overline{B^*} & & &\overline{B} &\overline{V^*(-\theta,-\frac{n-1}{2}-\frac{d-1}{2d})}
    \end{matrix}\right) \\
&=\overline{f_{E,n}(-\theta-\frac{nd-1}{2d}\omega)}, 
\end{align}
as claimed.
\end{proof}
Lemma \ref{lem:f_even} implies that
\begin{align}
    f_{E,n}(z)=\overline{f_{E,n}(\overline{z} e^{-2\pi i(nd-1)d^{-1}\omega})}, \text{ for any } z\in \mathcal{C}_1.
\end{align}
Since both sides of the above are holomorphic functions in $z$, they must be identical to each other. 

\begin{lemma}\label{lem:even_zero}
If $z$ is a zero of $f_{E,n}(z)$, then $\overline{z}e^{-2\pi i(nd-1)d^{-1}\omega}$ is also a zero.
\end{lemma}

Following the same arguments as in the proof of~\cite[Lemma 6.2]{HS2}, we obtain a complexity bound on the large deviation set of the determinant using the upper bound $2n(d+\nu)$ on the zero count  as in Theorem~\ref{thm:Riesz_un}, together with Lemma~\ref{lem:even_zero}. 

\begin{lemma}\label{lem:complexity_1}
For any $\nu\in (0,1/10)$, and any large $\kappa_0$-admissible $n$ (see \eqref{def:admissible} in Lemma \ref{lem:deno}), there exists an integer $N\leq n(d+\nu)$ and a collection of intervals $\mathcal{F}_{E,n}=\bigcup_{j=1}^N U_j$ such that
the following large deviation set satisfies 
\begin{align}\label{def:wt_B_fEn}
\widetilde{\mathcal{B}}_{f,E,n}&:=\left\{\theta\in \T:\, \log |f_{E,n}(\theta)|<n(\log |\det B|+L^d(\omega,M_E))-n^{1-\frac{\delta_1}{2}}\right\}\\
&\subseteq \bigcup_{j=1}^N (U_j\cup (-U_j-(nd-1)d^{-1}\omega)),
\end{align}
in which each $\mathrm{mes}(U_j)\leq e^{-n^{\delta_1/2}}$.
\end{lemma}
\begin{remark}
    Note the $\widetilde{\mathcal{B}}_{f,E,n}$ differs from $\mathcal{B}_{f,E,n,\varepsilon=0}$ in Lemma \ref{lem:deno} by a factor $1/2$ in $\delta_1$. This is due to the application of the Cartan estimate in the proof, see \cite[Lemma 6.2]{HS2}.
\end{remark}

The rest of the proof follows the same strategy as in \cite{HS2} with some minor modifications. We sketch the argument below.

\begin{lemma}\label{lem:I1_I2}
Let $\omega\in \mathrm{DC}_{a,A}$ for some $a>0$ and $A>1$.
For any large $\kappa_0$-admissible $n$, and any $y\in \Z$ such that\footnote{The proof for negative $y$ is analogous by symmetry.} $nd<y<10nd$, 
let 
\begin{align}
I_1:=&[-[\frac{7}{8}nd], -[\frac{1}{8}nd]]\\
I_2:=&[y-[\frac{7}{8}nd], y-[\frac{1}{8}nd]],
\end{align}
where $[x]$ stands for the integer part of $x\in \R$.
There exists $\ell\in I_1\cup I_2$ such that 
\begin{align*}
\theta+\ell d^{-1} \omega\notin \bigcup_{j=1}^{N} (U_j \cup (-(nd-1)d^{-1}\omega-U_j)).
\end{align*}
\end{lemma}
\begin{proof}
First note that if $\omega\in \mathrm{DC}_{a,A}$, then for any $k\in \Z\setminus \{0\}$ 
\begin{align}\label{eq:domega_Dio}
    \|kd^{-1}\omega\|_{\T}\geq \frac{a}{d |k|^{A}}.
\end{align}
Suppose otherwise, we have for some $k_0\neq 0$ that
\begin{align}
|k_0d^{-1}\omega-p|<\frac{a}{d|k_0|^A}
\end{align}
for some $p\in \Z$, which implies $\|k_0\omega\|_{\T}\leq |k_0\omega-dp|<\frac{a}{|k_0|^A}$ contradicting  $\omega\in \mathrm{DC}_{a,A}$.

Next, note that the cardinality
\begin{align}
\# I_1+\# I_2\geq \frac{3}{2} nd-2> n(d+\nu)\geq N,
\end{align}
for $n$ large enough.
It then suffices to prove that each pair $U_j\cup (-(nd-1)d^{-1}\omega-U_j)$ consists of at most one point in  $\{\theta+\ell d^{-1}\omega\}_{\ell\in I_1\cup I_2}$. 
Arguing by contradiction, suppose there exist $\ell_1, \ell_2$ such that
\begin{align*}
\theta+\ell_1 d^{-1}\omega\in U_j, \text{ and } \theta+\ell_2 d^{-1}\omega\in U_j.
\end{align*}
Then by \eqref{eq:domega_Dio} and that $|\ell_1-\ell_2|<11nd$, 
\begin{align*}
|U_j|\geq \|\theta+\ell_1 d^{-1}\omega-(\theta+\ell_2 d^{-1}\omega)\|_{\T}=\|(\ell_1-\ell_2)d^{-1}\omega\|_{\T}\geq \frac{a}{d(11nd)^A}>e^{-n^{\delta_1/2}},
\end{align*}
contradicting Lemma \ref{lem:complexity_1}.
The case when 
\begin{align*}
\theta+\ell_1d^{-1}\omega\in (-(nd-1)d^{-1}\omega-U_j), \text{ and } \theta+\ell_2d^{-1}\omega\in (-(nd-1)d^{-1}\omega-U_j). 
\end{align*}
is similar.
In fact, suppose there exist $\ell_1, \ell_2$ such that
\begin{align*}
\theta+\ell_1d^{-1}\omega\in U_j, \text{ and } \theta+\ell_2d^{-1}\omega\in (-(nd-1)d^{-1}\omega-U_j). 
\end{align*}
Since $\theta\in (\Theta_d)^c$, there exists $a'>0$ and $t>1$ such that for $k\in d^{-1}\Z$ large enough, one has 
\begin{align*}
\|2\theta+k\omega\|_{\T}\geq \frac{a'}{|k|^t}.
\end{align*}
Using that $-3nd/4\leq \ell_1+\ell_2+nd\leq 11nd$, we infer that 
\begin{align*}
|U_j|\geq &\|\theta+\ell_1d^{-1}\omega-(-\theta-\ell_2d^{-1}\omega-(nd-1)d^{-1}\omega\|_{\T}\\
=&\|2\theta+(\ell_1+\ell_2+nd-1)d^{-1}\omega\|_{\T}\geq \frac{a'}{(11n)^{t}}\geq e^{-n^{\delta_1/2}}.
\end{align*}
This contradicts with Lemma \ref{lem:complexity_1} again. 
Thus the claimed results hold.
\end{proof}

Next, we show the following.
\begin{lemma}\label{lem:I1_small}
Under the same conditions as Lemma \ref{lem:I1_I2}.
For any $\ell\in I_1$, one has $\theta+\ell d^{-1}\omega\in \bigcup_{j=1}^{N}(U_j\cup (-(nd-1)d^{-1}\omega-U_j))$.
\end{lemma}
\begin{proof}
Argue by contradiction. 
Suppose there exists $\ell_1\in I_1$ such that 
$$\theta+\ell_1\alpha\notin \bigcup_{j=1}^{N}(U_j\cup (-(nd-1)d^{-1}\omega-U_j)).$$
By Lemma \ref{lem:complexity_1}, it is necessary that $\theta+\ell_1d^{-1}\omega\notin \widetilde{\mathcal{B}}_{f,n,E}$, which implies
\begin{align}\label{eq:u_I1_large}
\frac{1}{n}\log |f_{E,n}(\theta+\ell_1d^{-1}\omega)|\geq L^d+\log |\det B|-n^{-\delta_1/2}.
\end{align}
Let $\ell_2:=\ell_1+nd-1$. 
By Lemma \ref{lem:numerator}, for any $\varepsilon_1>0$, we see that for $n$ large enough, 
\begin{align}\label{eq:Dk_upper}
&\sup_{m\in \{0,...,d-1\}}\max(\log |\mu_{n,m,-\ell_1}(\theta+\ell_1d^{-1}\omega)|, \log |\mu_{n,nd-1-m,-\ell_1}(\theta+\ell_1 d^{-1}\omega)|)\\
&\qquad\qquad \leq n\cdot \log |\det B|+\max(-[\ell_1d^{-1}] L^{d-1}+[\ell_2d^{-1}] L^d, -[\ell_1d^{-1}] L^d+[\ell_2d^{-1}] L^{d-1})+n\varepsilon_1
\end{align}
This implies by \eqref{eq:mufn} that for any $m\in \{0,...,d-1\}$,
\begin{align}\label{eq:Green_I1}
    |G_{E,n}(\theta+\ell_1d^{-1}\omega; -\ell_1, m)|=&\frac{|\mu_{n,m,-\ell_1}(\theta+\ell_1d^{-1}\omega)|}{|f_{E,n}(\theta+\ell_1d^{-1}\omega)|}\leq \max(e^{[\ell_1d^{-1}]L_d}, e^{-[\ell_2d^{-1}] L_d})\cdot  e^{n\varepsilon_1}\\
    |G_{E,n}(\theta+\ell_1d^{-1}\omega; -\ell_1, nd-1-m)=&\frac{|\mu_{n,nd-1-m,-\ell_1}(\theta+\ell_1d^{-1}\omega)|}{|f_{E,n}(\theta+\ell_1d^{-1}\omega)|}\leq \max(e^{[\ell_1d^{-1}] L_d}, e^{-[\ell_2d^{-1}] L_d})\cdot  e^{n\varepsilon_1}.
\end{align}
Combining the above with \eqref{eq:Poisson_exp} and \eqref{eq:Schnol}, we have
\begin{align}\label{eq:phi0}
|u_0|
\leq C\sum_{k\in \{0,...,d-1\}\cup \{(n-1)d,...,nd-1\}}\max(e^{[\ell_1d^{-1}] L_d}, e^{-[\ell_2d^{-1}] L_d})\cdot  e^{n\varepsilon_1} \cdot (nd)
\leq e^{-\frac{1}{10} n L_d}.
\end{align}
invoking $\min(|\ell_1|, |\ell_2|)\geq [nd/8]$. 
Similarly, one shows that $|u_1|<1/2$.
Hence we arrive at a contradiction with the assumption that $\max(|u_0|, |u_1|)=1$.
\end{proof}

Combining Lemmas \ref{lem:I1_I2} with \ref{lem:I1_small} yields
\begin{corollary}\label{cor:I2_large}
Under the same conditions as Lemma \ref{lem:I1_I2}.
There exists $\ell_3\in I_2$ such that $\theta+\ell_3\alpha\notin \widetilde{\mathcal{B}}_{f,E,n}$.
\end{corollary}
The proof of Anderson localization then follows from a similar argument as in the proof of Lemma~\ref{lem:I1_small}.
Indeed, similarly to \eqref{eq:phi0}, one has
\begin{align}\label{eq:phiy}
|\phi_y|\leq e^{-\frac{1}{10}nL_d}\leq e^{-\frac{L_d}{100d}y}.
\end{align}
This proves the claimed result.

\subsection{Proof of Theorem \ref{thm:acceleration=d_2}}
Recall that we assumed that there exists an orthonormal matrix $J$, such that 
\begin{align}\label{eq:JVJ=Vt_JBJ=Bt}
JV(\theta)J^{-1}=V^T(-\theta), \text{ and } JBJ^{-1}=B^T.
\end{align}
First, we prove a lemma which is analogous to Lemma \ref{lem:f_even}.
\begin{lemma}\label{lem:f_even2}
We have
\[f_{E,n}(\theta-\frac{n-1}{2}\omega)=\overline{f_{E,n}(-\theta-\frac{n-1}{2}\omega)}.\]
\end{lemma}
\begin{proof}
We calculate (recall that $V(\theta,k):=V(\theta+k\omega)-E$)
\begin{align}
&f_{E,n}(\theta-\frac{n-1}{2}\omega)\\
=&
\det\left(\begin{matrix}
V(\theta,\frac{n-1}{2}) & B^* & &  &B\\
B &V(\theta,\frac{n-3}{2}) &\ddots \\
& \ddots &\ddots &\ddots \\
& &\ddots &\ddots &B^*\\
B^* & & &B &V(\theta,-\frac{n-1}{2})
\end{matrix}\right)\\
=&
\det\left(\begin{matrix}
V(\theta,-\frac{n-1}{2}) & B & &  &B^*\\
B^* &V(\theta,-\frac{n-3}{2}) &\ddots \\
& \ddots &\ddots &\ddots \\
& &\ddots &\ddots &B\\
B & & &B^* &V(\theta,\frac{n-1}{2})
\end{matrix}\right)\\
=&
\det\left(\begin{matrix}
JV(\theta,-\frac{n-1}{2})J^{-1} & JBJ^{-1} & &  &JB^*J^{-1}\\
JB^*J^{-1} &JV(\theta,-\frac{n-3}{2})J^{-1} &\ddots \\
& \ddots &\ddots &\ddots \\
& &\ddots &\ddots &JBJ^{-1}\\
JBJ^{-1} & & &JB^*J^{-1} &JV(\theta,\frac{n-1}{2})J^{-1}
\end{matrix}\right)
\end{align}
This can further be simplified in the form
\begin{align}\notag
&f_{E,n}(\theta-\frac{n-1}{2}\omega)\\
=&
\det\left(\begin{matrix}
V^T(-\theta,\frac{n-1}{2}) & B^T & &  &\overline{B}\\
\overline{B} &V^T(-\theta,\frac{n-3}{2}) &\ddots \\
& \ddots &\ddots &\ddots \\
& &\ddots &\ddots &B\\
B^T & & &\overline{B} &V^T(-\theta,-\frac{n-1}{2})
\end{matrix}\right)\\
=&
\det\left(\begin{matrix}
\overline{V(-\theta,\frac{n-1}{2})} & \overline{B^*} & &  &\overline{B}\\
\overline{B} &\overline{V(-\theta,\frac{n-3}{2})} &\ddots \\
& \ddots &\ddots &\ddots \\
& &\ddots &\ddots &\overline{B^*}\\
\overline{B^*} & & &\overline{B} &\overline{V(-\theta,-\frac{n-1}{2})}
\end{matrix}\right)\\
=&\overline{f_{E,n}(-\theta-\frac{n-1}{2}\omega)},
\end{align}
in which we used $V^*(\theta)=V(\theta)$.
\end{proof}
In analogy with Lemma \ref{lem:f_even}, Lemma \ref{lem:f_even2} implies the following.
\begin{lemma}\label{lem:even_zero_2}
If $z$ is a zero of $f_{E,n}(z)$, then $\overline{z}e^{-2\pi i (n-1)\omega}$ is also a zero.
\end{lemma}

Taking into account the additional assumption \eqref{assume:t2_2} that $f_{E,n}(\theta)=f_{E,n}(\theta+\frac{1}{d})$, we have
\begin{lemma}\label{lem:f_symmetry}
If $z$ is a zero of $f_{E,n}(z)$, then $e^{2\pi i\frac{1}{d}}\cdot z$ is also a zero.
\end{lemma}
Note that by Theorem \ref{thm:Riesz_un}, for any small $\epsilon_1>0$, for $n$ large, $N_n(E,\eta/3)\leq 4nd(1-8\varepsilon_1)$.
In analogy  to Lemma \ref{lem:complexity_1}, Lemmas \ref{lem:even_zero_2} and \ref{lem:f_symmetry} imply the following.
\begin{lemma}\label{lem:complexity_2}
For any small $\varepsilon_1>0$, and any large $\kappa_0$-admissible $n$.
There exists $N\leq 2n(1-8\varepsilon_1)$ and a collection of intervals $\mathcal{F}_{E,n}=\bigcup_{j=1}^N U_j$ such that
the large deviation set satisfies 
\begin{align}\label{def:BfEn^(2)}
\widetilde{\mathcal{B}}_{f,E,n}&=\left\{\theta\in \T:\, \log |f_{E,n}(\theta)|<n(\log |\det B|+L^d(\omega,M_E))-n^{1-\frac{\delta_1}{2}}\right\}\\
&\subseteq \bigcup_{j=1}^N \left(\left(\bigcup_{m_1=1}^d (U_j+m_1/d)\right)\cup \left(\bigcup_{m_2=1}^d (-U_j-(n-1)\omega+m_2/d)\right)\right)=:\widetilde{\mathcal{B}}_{f,E,n}^{(2)}.
\end{align}
Furthermore each $U_j$ satisfies $\mathrm{mes}(U_j)<e^{-n^{\delta_1/2}}$.
\end{lemma}
similarly to Lemma \ref{lem:I1_I2}, we have
\begin{lemma}\label{lem:I1_I2_2}
    Let $\omega\in \mathrm{DC}_{a,A}$ for some $a>0$ and $A>1$. For any large $\kappa_0$-admissible $n$, and for any $y\in \Z$ such that $nd<y<10nd$, let
    \begin{align}
        I_1:=&[-[(1-\varepsilon_1)n], -[\varepsilon_1 n]]\\
        I_2:=&[[d^{-1}y]-[(1-\varepsilon_1)n], [d^{-1}y]-[\varepsilon_1 n]].
    \end{align}
    There exists $\ell\in I_1\cup I_2$ such that
    \begin{align}
        \theta+\ell \omega\notin \widetilde{\mathcal{B}}_{f,E,n}^{(2)}.
    \end{align}
\end{lemma}
This implies, analogously to Lemma \ref{lem:I1_small}, the following.
\begin{lemma}\label{lem:I1_small_2}
Under the same conditions as Lemma \ref{lem:I1_I2_2}.
For any $\ell\in I_1$, one has $\theta+\ell \omega\in \widetilde{\mathcal{B}}_{f,E,n}^{(2)}$.
\end{lemma}
\begin{proof}
    Suppose otherwise. Then there exists $\ell_1\in I_1$ such that $ \theta+\ell_1\omega \notin \widetilde{\mathcal{B}}_{f,E,n}^{(2)}$, implying $\theta+\ell_1\omega\notin \widetilde{\mathcal{B}}_{f,E,n}$.
    Hence
    \begin{align}
        \log |f_{E,n}(\theta+\ell_1\omega)|\geq n(\log |\det B|+L^d-n^{-\frac{\delta_1}{2}}).
    \end{align}
    Taking Lemma \ref{lem:numerator} into account, we have similarly to \eqref{eq:Green_I1} that for any $\varepsilon_1>0$ and $n$ large enough, for any $m\in \{0,...,d-1\}$:
    \begin{align}\label{eq:Green_I1_2}
        |G_{E,n}(\theta+\ell_1\omega;-\ell_1,m)|&\leq \max(e^{\ell_1 L_d}, e^{-\ell_2L_d})\cdot e^{n\varepsilon_1}\\
        |G_{E,n}(\theta+\ell_1\omega;-\ell_1,nd-1-m)|&\leq \max(e^{\ell_1L_d}, e^{-\ell_2L_d})\cdot e^{n\varepsilon_1}.
    \end{align}
Implying, in analogy with \eqref{eq:phi0} that
\begin{align}
    1=\max(|u_0|, |u_1|)\leq e^{-\frac{1}{2}\varepsilon_1 nL_d}.
\end{align}
Thus a contradiction. 
Taking Lemmas \ref{lem:I1_I2_2} and \ref{lem:I1_small_2} into account, we arrive at
\begin{lemma}
Under the same conditions as Lemma \ref{lem:I1_I2_2}.
For some $\ell_3\in I_2$, one has $\theta+\ell_3\omega \notin \widetilde{\mathcal{B}}_{f,E,n}^{(2)}$.   
\end{lemma}
This implies 
\begin{align}
    |u_y|\leq e^{-\frac{1}{2}\varepsilon_1 n L_d}\leq e^{-\frac{L_d}{20d}\varepsilon_1 y},
\end{align}
which is the claimed result.
\end{proof}

\section{Applications to spin chains and the skew shift}\label{sec:XY_sk}

\subsection{Anisotropic XY spin chain model in quasi-periodic magnetic fields}\label{sec:spin}
The XY-chain model is an exactly solvable model, first understood in \cite{LSM} by Lieb-Schultz-Mattis, due to the fact that the Jordan-Wigner transformation maps the XY-chain Hamiltonian to a Hamiltonian of free Fermions. 
Since then the XY-chain has become a prototypical model in understanding phenomena in many-body quantum
theory.

In this paper, we study the anisotropic XY-spin chain in magnetic fields defined as follows.
Let $\rho\in [0,1]$ represent the strength of the anisotropy. 
For $n\in \N$, we denote the finite volume, anisotropic XY Hamiltonian with free boundary conditions by
\begin{align*}
H_{[1,n],\omega,\theta,\rho,v}^{XY}=\sum_{j=1}^{n-1} [(1+\rho)\sigma_j^x\sigma_{j+1}^x+(1-\rho)\sigma_j^y\sigma_{j+1}^{y}]+\sum_{j=1}^n v(\theta+n\omega) \sigma_j^z,
\end{align*}
where $\omega\in \T^b$ is the frequency, $\theta\in \T^b$ is the phase and the potential $v$ is assumed to be a non-constant analytic function on $\T^b$. The underlying Hilbert space is $\calH=\bigotimes_{j=1}^n\C^2$. 
The matrices $\sigma^x, \sigma^y, \sigma^z$ are the Pauli matrices given by 
\begin{align*}
\sigma^x=\left(\begin{matrix}
0 \ &1\\
1 &0
\end{matrix}\right),\,
\sigma^y=\left(\begin{matrix}
0 \ &-i\\
i &0
\end{matrix}\right),
\text{ and }
\sigma^z=\left(\begin{matrix}
1 \ &0\\
0 &-1
\end{matrix}\right),
\end{align*}
and $\sigma_j^x$ etc.\ means that the Pauli matrix acts on the $j^{th}$ component of the tensorial state. 
It is well-known that the XY-chain models can be reduced to the following Hamiltonian $\widetilde{H}^{XY}_{n,\omega,\theta,\rho,v}$ of the free Fermions via the Jordan-Wigner transformation (see e.g. \cite{HSS}):
\begin{equation}\label{def:tH_XY}
\widetilde{H}^{XY}_{n,\omega,\theta,\rho,v}=\left( \begin{array}{ccccc} V(\theta+n\omega) & B^* & & & \\ B & V(\theta+(n-1)\omega) & \ddots & & \\ & \ddots & \ddots & \ddots & \\ & & \ddots & \ddots & B^*\\ & & & B & V(\theta+\omega) \end{array} \right),
\end{equation}
where 
\begin{align}
    B=\left(\begin{matrix}1 &\rho\\ -\rho &-1\end{matrix}\right), \text{ and } V(\theta)=\left(\begin{matrix} v(\theta) &0 \\ 0 &-v(\theta)\end{matrix}\right).
\end{align}
Note that in the isotropic case ($\rho=0$), the resulting Hamiltonian of the free Fermions is scalar valued, thus is much better understood than the anisotropic case. 

It was shown by Hamza-Sims-Stolz \cite{HSS} that dynamical localization of the anisotropic XY-chain, characterized by the zero-velocity Lieb-Robinson bound, is equivalent to the exponential dynamical localization of $\widetilde{H}^{XY}_{n,\omega,\theta,\rho,v}$, viz. 
\begin{align*}
E\left(\sup_{t\in \R} \left| (e^{-i t \widetilde{H}^{XY}_{n,\omega,\theta,\rho,v} })_{j,k}\right|\right)\leq C e^{-\eta |j-k|}
\end{align*}
for all $n$ and $j,k\in [1,2n]$. Later it was proved in \cite{CS} by Chapman-Stolz that dynamical localization holds for the anisotropic XY-chain with random magnetic fields at arbitrarily small disorder.
Spin chains with deterministic magnetic fields are much less understood.
For isotropic XY-chains with quasi-periodic magnetic fields, an interesting anomalous Lieb-Robinson bound was proved by Damanik-Lemm-Lukic-Yessen \cite{DLLY}. Positive lower bounds of the Lieb-Robinson velocity were obtained by \cite{DLY,Ka,Fi} in the periodic setting. Moreover, they covered  quasi-periodic models under a reducibility assumption, as well as the limit-periodic case. 
As far as the anisotropic case is concerned, the only results so far are by Damanik-Lukic-Yessen \cite{DLY} for periodic magnetic fields.

In this paper, we study the infinite volume version of the anisotropic $\widetilde{H}_{n,\omega,\theta,\rho,v}^{XY}$ in the positve Lyapunov exponent regime, which complements the earlier studies. Thus, we define 
\begin{align}
(\widetilde{H}^{XY}_{\omega,\theta,\rho,v}\phi)_n=B \phi_{n+1}+V(\theta+n\omega) \phi_n+B^* \phi_{n-1},
\end{align}
where 
\begin{align}
    B=\left(\begin{matrix}1 &\rho\\ -\rho &-1\end{matrix}\right), \text{ and } V(\theta)=\left(\begin{matrix} v(\theta) &0 \\ 0 &-v(\theta)\end{matrix}\right).
\end{align}
Clearly this is a $2\times 2$ block-valued Jacobi matrix. As an immediate corollary of Theorem \ref{thm:main}, one has:
\begin{corollary}\label{cor:thm1_XY}
$\widetilde{H}^{XY}_{\omega,\theta=0,\rho,v}$ is Anderson localized in the positive Lyapunov exponent regime for a.e. $\omega\in \T^b$.
\end{corollary}

We now restrict to the case of one-dimensional torus to address arithmetic Anderson localization.
\subsubsection*{Proof of Theorem \ref{thm:XY}}
We obtain Theorem \ref{thm:XY} as a corollary of Theorem \ref{thm:acceleration=d_2}.
It suffices to check the conditions \eqref{assume:t2_1} and \eqref{assume:t2_2}.
Let $J$ be the following orthonormal matrix:
\begin{align}
J=\left(\begin{matrix} 1 & 0\\ 0 &-1\end{matrix}\right). 
\end{align}
Clearly $JV(\theta)J^{-1}=V^T(-\theta)$, where we used $v$ is even, and $JBJ^{-1}=B^T$, which verifies \eqref{assume:t2_1} of Theorem \ref{thm:acceleration=d_2}.
Next, let us verify $f_{E,n}(\theta)=f_{E,n}(\theta+\frac{1}{2})$ for even $n$.
It is easy to see that for 
\begin{align}
    L=\left(\begin{matrix}0 & 1\\ 1 & 0\end{matrix}\right)
\end{align}
we have
\begin{align}
    L(V(\theta+\frac{1}{2})-E)L^{-1}=\left(\begin{matrix} -v(\theta+\frac{1}{2})-E &0\\ 0&v(\theta+\frac{1}{2})-E\end{matrix}\right)=V(\theta)-E,
\end{align}
where we used the assumption that $v(\theta+\frac{1}{2})=-v(\theta)$. It is also straightforward to check $LBL^{-1}=-B$.
Hence, with $V_E(\theta,j):=V(\theta+j\omega)-E$, one has 
\begin{align}
   &f_{E,n}(\theta+\frac{1}{2})\\
=&
\det\left(\begin{matrix}
V_E(\theta+\frac{1}{2},n-1) & B^* & &  &B\\
B &V_E(\theta+\frac{1}{2},n-2) &\ddots \\
& \ddots &\ddots &\ddots \\
& &\ddots &\ddots &B^*\\
B^* & & &B &V_E(\theta+\frac{1}{2},0)
\end{matrix}\right)\\
=&
\det\left(\begin{matrix}
LV_E(\theta+\frac{1}{2},n-1)L^{-1} & LB^*L^{-1} & &  &LBL^{-1}\\
LBL^{-1} &LV_E(\theta+\frac{1}{2},n-2)L^{-1} &\ddots \\
& \ddots &\ddots &\ddots \\
& &\ddots &\ddots &LB^*L^{-1}\\
LB^*L^{-1} & & &LBL^{-1} &LV_E(\theta+\frac{1}{2},0)L^{-1}
\end{matrix}\right)\\
=&
\det\left(\begin{matrix}
V_E(\theta,n-1) & -B^* & &  &-B\\
-B &V(\theta,n-2) &\ddots \\
& \ddots &\ddots &\ddots \\
& &\ddots &\ddots &-B^*\\
-B^* & & &-B &V_E(\theta,0)
\end{matrix}\right).
\end{align}
Finally we remove the minus signs in front of $B$ and $B^*$ via conjugating the matrix by (only possible for even $n$)
\[\mathrm{diag}(I_2, -I_2, I_2, -I_2,..., I_2, -I_2)_{2n\times 2n}.\]
Hence
\begin{align}
f_{E,n}(\theta+\frac{1}{2})=\det\left(\begin{matrix}
V_E(\theta,n-1) & B^* & &  &B\\
B &V_E(\theta,n-2) &\ddots \\
& \ddots &\ddots &\ddots \\
& &\ddots &\ddots &B^*\\
B^* & & &B &V_E(\theta,0)
\end{matrix}\right)
=f_{E,n}(\theta).
\end{align}
This verifies the assumption \eqref{assume:t2_2} of Theorem \ref{thm:acceleration=d_2} (for even $n$ only; hence one needs to further restrict the admissible sequence to even numbers. However, this is still a sequence of positive density, thus does not affect the proof of localization).
Hence Theorem \ref{thm:acceleration=d_2} implies Theorem \ref{thm:XY} as a corollary. \qed

\subsection{Skew-shift with rational frequencies}
Let $\theta,y\in \T$.
Let us consider the following operator on $\ell^2(\Z, \C^q)$:
\begin{align}\label{def:hatH_sk}
    (\widehat{H}_{\lambda,\theta,y,p/q}\hat{U})_k=B\, \hat{U}_{k+1}+B^* \hat{U}_{k-1}+V(\theta+ky)\hat{U}_{k},
\end{align}
where
\begin{align}\label{eq:V_sk}
    V(\theta)=
    \left(\begin{matrix}
        0 & e^{2\pi i \theta} & & &e^{-2\pi i\theta}\\
        e^{-2\pi i\theta} &0 &\ddots & &\\
        &\ddots &\ddots &\ddots &\\
        & &\ddots &\ddots &e^{2\pi i\theta}\\
        e^{2\pi i\theta}& & &e^{-2\pi i\theta} &0
    \end{matrix}\right)_{q\times q},
\end{align} 
and $B\in \mathrm{Mat}(\C, q)$ be an arbitrary diagonal matrix. A particular choice of $B$ that is of interest to us is the following:
\begin{align}\label{def:Bsk}
    B_{sk}=\lambda\, \mathrm{diag}(e^{2\pi i j(j-1)p/q})_{j=q-1}^0,
\end{align}
with $p/q\in (0,1)$ being a reduced rational with $q\geq 3$. 
The specific choice of $B_{sk}$ arises from taking the dual model of the following Schr\"odinger operator $H^{sk}_{\lambda,x,y,p/q}$ on $\ell^2(\Z)$ with skew-shift dynamics with frequency $p/q$, see \eqref{def:skew_shift}:
\begin{align}
    (H^{sk}_{\lambda,x,y,p/q}u)_n=u_{n+1}+u_{n-1}+2\lambda\cos(2\pi(x+ny+n(n-1)p/q)))u_n.
\end{align}
For any irrational $y$, the spectrum $\sigma(H^{sk}_{\lambda,x,y,p/q})$ is constant in $x$, and $\sigma(\widehat{H}_{\lambda,\theta,y,p/q})$ is constant in $\theta$.
In the rest of this section, we will fix a $y\in \mathrm{DC}$.

The first goal of this section is to prove the almost localization of $\widehat{H}_{\lambda,\theta,y,p/q}$.
\begin{definition}[Resonances of $\theta$]
Let $y,\theta\in \T$ and $\varepsilon>0$. We say $k\in \Z$ is an $\varepsilon$-resonance of $\theta$ if $\|q\cdot (2\theta-ky)\|_{\T}\leq e^{-|k|^{\varepsilon}}$ and $\|q\cdot (2\theta-ky)\|_{\T}=\min_{|j|\leq k} \|q\cdot (2\theta-jy)\|_{\T}$.
\end{definition}
In the rest of the section, we fix some $\varepsilon$ such that $0<\varepsilon<\delta_1/4$, where $\delta_1>0$ is as in Lemma \ref{lem:complexity_2}. 

\begin{definition}
Let $0=|n_0|\leq |n_1|\leq |n_2|\leq ...$ be the $\varepsilon$-resonances of $\theta$. If this sequence is infinite, we
say $\theta$ is $\varepsilon$-resonant, otherwise we say it is $\varepsilon$-non-resonant. Furthermore, if $\theta$ is $\varepsilon$-non-resonant with a finite sequence of resonances $0=|n_0|\leq |n_1|\leq ...\leq |n_j|$, we let $n_{j+1}=\infty$.
\end{definition}

\begin{remark}
If $y\in \mathrm{DC}_{a,A}$ for some $a>0$ and $A>1$, then for $j$ large enough, one has 
\[|n_{j+1}|>q^{-1} C_{a,A}\cdot e^{A^{-1}|n_j|^{\varepsilon}}.\] 
Indeed by triangle inequality, we have
\begin{align}
    e^{-|n_j|^{\varepsilon}}\geq \|q (2\theta-n_jy)\|_{\T}
    \geq &\|q(n_j-n_{j+1})y\|_{\T}-\|q (2\theta-n_{j+1}y)\|_{\T}\\
    \geq &\frac{a}{|q(n_j-n_{j+1})|^A}-e^{-|n_{j+1}|^{\varepsilon}},
\end{align}
implying the claimed inequality. 
\end{remark}

\begin{definition}[Almost localization]\label{def:almost_AL}
We say the family $\{\hat{H}_{y,\theta}\}_{\theta\in \T}$ is $(C_1, C_2,\tilde{\varepsilon})$-{\em almost localized} for some constant $C_1, C_2, \tilde{\varepsilon}>0$ if for every solution $\hat{u}$ of $\hat{H}_{y,\theta}\hat{u}=E\hat{u}$ for some energy $E\in \R$, satisfying $\hat{u}_0=1$ and $|\hat{u}_k|\leq 1+|k|$, and for every $C_1(1+|n_j|)<|k|<C_1^{-1}|n_{j+1}|$, the bound $|\hat{u}_k|\leq C_2 e^{-\tilde{\varepsilon} |k|}$ holds, where the $n_j$'s are the $\varepsilon$-resonances of $\theta$.
\end{definition}

\begin{theorem}\label{thm:sk_dual_almost_AL}
There exists a constant $C_2>0$ such that for any $y\in \mathrm{DC}$,
$\{\widehat{H}_{\lambda,\theta,y,p/q}\}_{\theta\in \T}$ is $(5, C_2, L_q(\omega,M_E)/(10q))$-almost localized in $\{E: L_q(\omega,M_E)>0\}$.
\end{theorem} 

\begin{proof}
We first check the conditions \eqref{assume:t2_1} and \eqref{assume:t2_2} of Theorem \ref{thm:acceleration=d_2} are satisfied and that for $E\in \sigma(\widehat{H}_{\lambda,\theta,y,p/q})\cap \{E: L_q(\omega,M_E)>0\}$,
\begin{align}\label{eq:kappa_q=q}
\kappa^q(\omega,M_E)\leq q.
\end{align}
Note this provides a tighter upper bound than the required $2q-1$ as in Theorem \ref{thm:acceleration=d_2}.
It is easy to compute that as $\varepsilon\to\infty$, $L^q_\varepsilon(\omega,M_E)=2\pi q\varepsilon+o(1)$, hence $\kappa^q_\varepsilon(\omega,M_E)=q$ for $\varepsilon$ large enough. 
Convexity of $L^q_\varepsilon(\omega,M_E)$ in $\varepsilon$ implies $\kappa^q(\omega,M_E)\leq q$,
which verifies \eqref{eq:kappa_q=q}.

Next, we turn to \eqref{assume:t2_1}, which follow directly from $V(-\theta)=V^T(\theta)$ and $B=B^T$.

To verify \eqref{assume:t2_2}, we let 
\begin{align}
    L=\mathrm{diag}(1, e^{2\pi i\frac{1}{q}}, ..., e^{2\pi i \frac{q-1}{q}}). 
\end{align}
We have
\begin{align}
    LV(\theta+\frac{1}{q})L^{-1}=
    L\left(\begin{matrix}
    0 & e^{2\pi i (\theta+\frac{1}{q})} & & &e^{-2\pi i(\theta+\frac{1}{q})}\\
        e^{-2\pi i(\theta+\frac{1}{q})} &0 & & &\\
        \\
        & & & &e^{2\pi i(\theta+\frac{1}{q})}\\
        e^{2\pi i(\theta+\frac{1}{q})}& & &e^{-2\pi i(\theta+\frac{1}{q})} &0
    \end{matrix}\right)L^{-1}=V(\theta),
\end{align}
and clearly since $B$ is diagonal, $LBL^{-1}=B$. 
Thus, with $V_E(\theta,j):=V(\theta+j\omega)-E$,  
\begin{align}
       &f_{E,n}(\theta+\frac{1}{q})\\
&=
\det\left(\begin{matrix}
V_E(\theta+\frac{1}{q},n-1) & B^* & &  &B\\
B &V_E(\theta+\frac{1}{q},n-2) &\ddots \\
& \ddots &\ddots &\ddots \\
& &\ddots &\ddots &B^*\\
B^* & & &B &V_E(\theta+\frac{1}{q},0)
\end{matrix}\right)\\
&=
\det\left(\begin{matrix}
LV_E(\theta+\frac{1}{q},n-1)L^{-1} & LB^*L^{-1} & &  &LBL^{-1}\\
LBL^{-1} &LV_E(\theta+\frac{1}{q},n-2)L^{-1} &\ddots \\
& \ddots &\ddots &\ddots \\
& &\ddots &\ddots &LB^*L^{-1}\\
LB^*L^{-1} & & &LBL^{-1} &LV_E(\theta+\frac{1}{q},0)L^{-1}
\end{matrix}\right)\\
&=
\det\left(\begin{matrix}
V_E(\theta,n-1) & B^* & &  &B\\
B &V_E(\theta,n-2) &\ddots \\
& \ddots &\ddots &\ddots \\
& &\ddots &\ddots &B^*\\
B^* & & &B &V_E(\theta,0)
\end{matrix}\right)\\
&=f_{E,n}(\theta).
\end{align}
Thus  assumption \eqref{assume:t2_2} of Theorem \ref{thm:acceleration=d_2} holds with $d=q$.
The proof of almost localization proceeds as that of the localization as in Theorem \ref{thm:acceleration=d_2}, except that Lemma \ref{lem:I1_I2_2} holds only for non-resonant $z$'s (instead of arbitrary large $|z|$) and the definitions of $I_1, I_2$ intervals require minor modifications.
Indeed, we have
\begin{lemma}\label{lem:I1_I2_2'}
Let $y\in \mathrm{DC}_{a,A}$ for some $a>0$ and $A>1$.
Let $\kappa_0>0$ be as in Lemma \ref{lem:deno}, and $C_*>0$ be the constant as in Remark \ref{rem:admissible}.
Let $\widetilde{\mathcal{B}}_{f,E,n}^{(2)}$ be as in \eqref{def:BfEn^(2)}.
For $j$ large enough, and for $z$ such that $5(1+|n_j|)<z/q<|n_{j+1}|/5$ \footnote{The case of negative $z$ can be handled similarly.}. Let $n\in \Z$ be $\kappa_0$-admissible such that $n\leq z/q<n+C_*$, let
\begin{align}
    I_1:=&[-[\frac{1}{2}n], -[\frac{1}{8}n]],\, I_2:=[[\frac{z}{q}]-[\frac{7}{8}n],[\frac{z}{q}]-[\frac{1}{8}n]], \text{ if } n_j\geq 0,  \\
    I_1:=&[-[\frac{7}{8}n], -[\frac{1}{2}n]],\, I_2:=[[\frac{z}{q}]-[\frac{7}{8}n],[\frac{z}{q}]-[\frac{1}{8}n]], \text{ if } n_j<0.
\end{align}
There exists $\ell\in I_1\cup I_2$ such that
\begin{align}
    \theta+\ell y\notin \widetilde{\mathcal{B}}^{(2)}_{f,E,n}.
\end{align}
\end{lemma}
The principle for choosing $I_1, I_2$ is to guarantee \eqref{eq:I_1+I_2_neq_nj} holds.

We will prove this lemma in details. The rest of the proof of almost localization is the same as that of Theorem \ref{thm:acceleration=d_2}.
\begin{proof}
Towards a contradiction, suppose $\theta+\ell\omega\in \widetilde{\mathcal{B}}_{f,E,n}^{(2)}$ for any $\ell\in I_1\cup I_2$.

Clearly $|I_1|+|I_2|>\frac{9}{8}n-C$, for some absolute constant $C>0$.
Also since  $\kappa^q\leq q$ (see \eqref{eq:kappa_q=q}), which is a tighter upper bound than $2q-1$, Lemma~\ref{lem:complexity_2} implies $N\leq n(1+\varepsilon_1)$ for $\varepsilon_1\in (0,1/10)$.
Since $|I_1|+|I_2|>N$, by the Pigeon hole principle, there exists $1\leq j\leq N$ and $\ell_1,\ell_2\in I_1\cup I_2$ such that 
\begin{align}
    \{\theta+\ell_1 y, \theta+\ell_2 y\}\subset \left(\bigcup_{m_1=1}^q (U_j+m_1/q)\right)\cup \left(\bigcup_{m_2=1}^q (-U_j-(n-1)y+m_2/q)\right)
\end{align}

It is clear that $y\in \mathrm{DC}_{a,A}$ and the measure estimate of $\mathrm{mes}(U_j)$ in Lemma \ref{lem:complexity_2} excludes the possibility of 
\begin{align}
    \{\theta+\ell_1 y, \theta+\ell_2 y\}\subset  \bigcup_{m_1=1}^q (U_j+m_1/q),
\end{align}
or 
\begin{align}
    \{\theta+\ell_1 y, \theta+\ell_2 y\}\subset  \bigcup_{m_2=1}^q (-U_j-(n-1)y+m_2/q)
\end{align}
It remains to consider the case when
\begin{align}\label{eq:assume_I1_I2_2'}
    \theta+\ell_1 y\in \bigcup_{m_1=1}^q (U_j+m_1/q), \text{ and } \theta+\ell_2 y\in \bigcup_{m_2=1}^q(-U_j-(n-1)y+m_2/q).
\end{align}
There exist $m_1,m_2\in \{1,...,d\}$ such that
\begin{align}
    \theta+\ell_1 y-m_1/q\in U_j, \text{ and } -(\theta+(\ell_2+n-1)y- m_2/q)\in U_j.
\end{align}
Taking the difference, we obtain
\begin{align}\label{eq:Uj_contra_1}
    \mathrm{mes}(U_j)
    \geq &\|2\theta+(\ell_1+\ell_2+n-1)y-(m_1+m_2)/q\|_{\T} \notag\\
    \geq &q^{-1}\|q\cdot (2\theta+(\ell_1+\ell_2+n-1)y)\|_{\T}.
\end{align}
The key to estimate the term on the right-hand side of the equation above is to show:
\begin{align}\label{eq:I_1+I_2_neq_nj}
\ell_1+\ell_2+n-1\neq -n_j.   
\end{align}
We divide into two different cases, depending on if $n_j\geq 0$.

\underline{Case 1}. If $n_j\geq 0$. Due to the choices of $I_1, I_2$ and $n$, one has
\begin{align}\label{eq:l1+l2_1}
    \ell_1+\ell_2+n-1\geq -n+n-1\geq -1>-n_j.
\end{align}

\underline{Case 2.1}. If $n_j<0$ and $\ell_1,\ell_2\in I_1$, then
\begin{align}
    \ell_1+\ell_2+n-1\leq -2[\frac{1}{2}n]+n-1\leq 1<-n_j.
\end{align}

\underline{Case 2.2}. If $n_j<0$ and $\{\ell_1,\ell_2\}\cap I_2\neq \emptyset$, we have
\begin{align}
    \ell_1+\ell_2+n-1\geq [\frac{z}{q}]-2[\frac{7}{8}n]+n-1\geq \frac{1}{4}n-2>\frac{z}{5q}>|n_j|.
\end{align}
Thus   \eqref{eq:I_1+I_2_neq_nj} holds.

Next, we show that 
\begin{align}\label{eq:l1+l2_neq_nj+1}
|\ell_1+\ell_2+n-1|<|n_{j+1}|.
\end{align}

\underline{Case 1}. If $n_j\geq 0$, we have by \eqref{eq:l1+l2_1} that
\begin{align}
    -\frac{2}{5}|n_{j+1}|<-1\leq \ell_1+\ell_2+n-1\leq 2[\frac{z}{q}]-2[\frac{1}{8}n]<\frac{2}{5}|n_{j+1}|
\end{align}

\underline{Case 2}. If $n_j<0$, then
\begin{align}
    -\frac{1}{5}|n_{j+1}|<-n<-2[\frac{7}{8}n]+n-1\leq \ell_1+\ell_2+n-1\leq 2[\frac{z}{q}]-2[\frac{1}{8}n]<\frac{2}{5}|n_{j+1}|
\end{align}
Hence we have verified \eqref{eq:l1+l2_neq_nj+1}.

Finally, suppose   $\ell_1+\ell_2+n-1=-n_k$ for some $k\leq j-1$. We again divide into two cases depending on the size of $\|q\cdot (2\theta-n_k y)\|_{\T}$.

\underline{Case 1}. 
$\|q\cdot (2\theta-n_k y)\|_{\T}\geq e^{-n^{\varepsilon}}$.
In this case we estimate \eqref{eq:Uj_contra_1} via the following:
\begin{align}\label{eq:Uj_contra_2}
    \mathrm{mes}(U_j)>q^{-1} e^{-n^{\varepsilon}}>e^{-n^{\delta_1/2}}, 
\end{align}
for $n$ large enough, due to $0<\varepsilon<\delta_1/4$.
This clearly contradicts with the measure estimate of $U_j$ from Lemma \ref{lem:complexity_2}.

\underline{Case 2}.
$\|q\cdot (2\theta-n_k y)\|_{\T}< e^{-n^{\varepsilon}}$. We are going to show this leads to a contradiction with the definition of $n_j$. 
Indeed, by $y\in \mathrm{DC}_{a,A}$,  
\begin{align}
    \|q (2\theta-n_j y)\|_{\T}
    \geq &\|q(n_k-n_j)y\|_{\T}-\|q(2\theta-n_k y)\|_{\T}\\
    \geq &\frac{a}{|q(n_k-n_j)|^A}-e^{-n^{\varepsilon}}\\
    \geq &C_{a,A} \frac{1}{|qn_j|^A}-e^{-n^{\varepsilon}}\\
    \geq &C_{a,A} \frac{1}{(qn)^A}-e^{-n^{\varepsilon}}\\
    \geq &e^{-n^{\varepsilon}}>\|q(2\theta-n_k y)\|_{\T},
\end{align}
this contradicts with $\|q(2\theta-n_jy)\|_{\T}=\min_{|m|\leq n_j}\|q(2\theta-my)\|_{\T}$.
\end{proof}
The rest of the proof of almost localization follows the same steps as that of Theorem \ref{thm:acceleration=d_2}.
\end{proof}

\subsubsection*{Proof of Theorem \ref{thm:skew_shift}}
The $2\times 2$ transfer matrix associated to $H^{sk}_{\lambda,x,y,p/q}$ along the skew-shift dynamics \begin{align}\label{def:skew_shift}
T^{sk}_{p/q}(x,y)=(x+y, y+p/q)
\end{align}
is
\begin{align}
    M^{\cos}_{E}(x):=\left(\begin{matrix}
        E-2\lambda\cos(2\pi x) &-1\\
        1 &0
    \end{matrix}\right).
\end{align}
Our goal is to show that if $|\lambda|>0$ is small enough, then for every $y\in \mathrm{DC}$, and every $E\in \sigma(H^{sk}_{\lambda,x,y,p/q})$, 
\begin{align}\label{eq:zero_LE_goal}
    L^{sk}_{p/q}(E,y):=\lim_{\ell\to\infty} \frac{1}{\ell}\int_{\T} \log \|\prod_{m=\ell-1}^0 M_E^{\cos}(x+my+m(m-1)p/q)\|\, \mathrm{d}x=0.
\end{align}

To this end, we first show $L^{sk}_{p/q}(E,y)=0$ for every eigenvalue of $\widehat{H}_{\lambda,0,y,p/q}$ with $B=B_{sk}$ as in \eqref{def:Bsk}.

Fix $\theta=0$ (one can also fix any non-resonant $\theta$). Clearly, for $y\in \mathrm{DC}$, $\theta=0$ is $\varepsilon$-non-resonant for any $\varepsilon>0$. Hence $n_{j+1}=\infty$ for some $j$. Theorem \ref{thm:sk_dual_almost_AL} implies $\widehat{H}_{\lambda,0,y,p/q}$ has pure point spectrum with exponentially decaying eigenfunctions if $|\lambda|>0$ is small enough. Indeed the potential matrix $V$ as in \eqref{eq:V_sk} does not have constant eigenvalue, hence by \cite[Theorem 2.3]{DK1}, $\widetilde{H}_{\lambda,\theta,y,p/q}$ has $q$ positive Lyapunov exponents, each of the order $\log (|\lambda|^{-1})+O(1)$.
Fixing an arbitrary eigenvalue $E$ of $\widehat{H}_{\lambda,0,y,p/q}$, and let $\hat{U}$ be the corresponding normalized eigenfunction. 
Let $\hat{U}=(...,\hat{U}_{1}, \hat{U}_0, \hat{U}_{-1},...)^T$, where for each $k\in \Z$,
\[\hat{U}_k=(\hat{u}_k^{(q-1)},...,\hat{u}_k^{(0)}).\]
For each $j\in \Z_q=\Z/(q\Z)$ and each $x\in \T$, let 
$u^{(j)}(x)=\sum_{k\in \Z} \hat{u}^{(j)}_k e^{-2\pi i kx}\in C^{\omega}(\T)$.
We have for each $j\in \Z_q$,
\begin{align}
    \lambda e^{2\pi i j(j-1)p/q} \hat{u}^{(j)}_{k+1}+\lambda e^{-2\pi i j(j-1)p/q} \hat{u}^{(j)}_{k-1}+e^{-2\pi i ky}\hat{u}^{(j+1)}_k+e^{2\pi i ky}\hat{u}^{(j-1)}_k=E\hat{u}^{(j)}_k,
\end{align}
This implies for any $x\in \T$ and $j\in \Z_q$ that 
\begin{align}
    2\lambda \cos(2\pi (x+j(j-1)p/q)) u^{(j)}(x)+u^{(j+1)}(x+y)+u^{(j-1)}(x-y)=Eu^{(j)}(x).
\end{align}
The above implies 
\begin{align}\label{eq:Mb=b}
    M_{\lambda,E}^{\cos}(x+j(j-1)p/q)\cdot b_{j,y}(x)=b_{j+1,y}(x+y),
\end{align}
in which for $j\in \Z_q$,
\begin{align}
    b_{j,y}(x):=\left(\begin{matrix} u^{(j)}(x)\\ u^{(j-1)}(x-y)\end{matrix}\right).
\end{align}
\eqref{eq:Mb=b} implies for any $\ell\in \Z$ and $x\in \T$ that
\begin{align}
    \prod_{m=\ell q-1}^{0}M_E^{\cos}(x+m y+m(m-1)p/q) \cdot b_{0, y}(x)=b_{0,y}(x+\ell q y).
\end{align}
We decompose
\begin{align}\label{eq:M_lq=tM_l}
    \prod_{m=\ell q-1}^{0}M_E^{\cos}(x+m y+m(m-1)p/q)
    =:\prod_{k=\ell-1}^{0} \widetilde{M}^{\cos}_{q,E,y,p/q}(x+kqy),
\end{align}
where 
\begin{align}
\widetilde{M}^{\cos}_{q,E,y,p/q}(x)=\prod_{m=q-1}^0 M_E^{\cos}(x+my+m(m-1)p/q).
\end{align}
Hence \eqref{eq:Mb=b} implies for every $x\in \T$,
\begin{align}
    \lim_{\ell\to\infty}\frac{1}{\ell}\log \|\prod_{k=\ell-1}^{0} \widetilde{M}^{\cos}_{q,E,y,p/q}(x+kqy)\cdot b_{0, y}(x)\|=0.
\end{align}
This implies via the Oseledets theorem, see~\cite{V}*{Theorem 4.2}, that
\begin{align}
    \lim_{\ell\to\infty}\frac{1}{\ell}\int_{\T} \log \|\prod_{k=\ell-1}^{0} \widetilde{M}^{\cos}_{q,E,y,p/q}(x+kqy)\|\, \mathrm{d}x=0,
\end{align}
which, due to \eqref{eq:M_lq=tM_l}, implies for every $y\in \mathrm{DC}$,
\begin{align}\label{def:LE_sk_y}
    L^{sk}_{p/q}(E,y):=\lim_{\ell\to\infty} \frac{1}{\ell}\int_{\T} \log \|\prod_{m=\ell-1}^0 M_E^{\cos}(x+my+m(m-1)p/q)\|\, \mathrm{d}x=0.
\end{align}
This implies $L^{sk}_{p/q}(E,y)=0$ on $\sigma(\widehat{H}_{\lambda,0,y,p/q})$, since eigenvalues of $\widehat{H}_{\lambda,0,y,p/q}$ form a dense set in $\sigma(\widehat{H}_{\lambda,0,y,p/q})$, and $L^{sk}_{p/q}(E,y)$ is a sub-harmonic function, hence upper semi-continuous in $E$.
Finally it suffices to show
\begin{align}\label{eq:claimed_dual}
\sigma(\widehat{H}_{\lambda,\theta,y,p/q})=\sigma(H^{sk}_{\lambda,x,y,p/q})
\end{align}
for any irrational $y$ and any $x,\theta\in \T$. 

Let $\mathcal{H}:=L^2(\T\times \Z)$ consisting of functions $u:\T\times \Z\to \C$ such that
\begin{align}
    \sum_{n\in \Z} \int_{\T} |u(x,n)|^2\, \mathrm{d}x<\infty,
\end{align}
and $\mathcal{H}_q:=L^2(\T\times (\Z\otimes \Z_q))$ consisting of functions $\hat{u}:\T\times (\Z\otimes \Z_q)\to \C$ such that
\begin{align}
    \sum_{m\in \Z} \sum_{j\in \Z_q} |\hat{u}(\theta,m,j)|^2\, \mathrm{d}\theta<\infty.
\end{align}
Define $H^{sk}_{\lambda,y,p/q}$ on $\mathcal{H}$ as follows, 
\begin{align}
    (H^{sk}_{\lambda,y,p/q}u)(x,n)=u(x,n+1)+u(x,n-1)+2\lambda\cos(2\pi(x+ny+n(n-1)p/q))u(x,n),
\end{align}
and $\widehat{H}_{\lambda,y,p/q}$ on $\mathcal{H}_q$ as:
\begin{align}
    (\widehat{H}_{\lambda,y,p/q}\hat{u})(\theta,m,j)
    =&\lambda e^{2\pi i\frac{j(j-1)p}{q}}\hat{u}(\theta,m+1,j)+e^{-2\pi i\frac{j(j-1)p}{q}}\hat{u}(\theta,m-1,j)\\
    &+e^{-2\pi i(\theta+my)}\hat{u}(\theta,m,j+1)+e^{2\pi i(\theta+my)}\hat{u}(\theta,m,j-1).
\end{align}
Following Chulaevsky-Delyon \cite{CD}, we define an unitary operator $U:\mathcal{H}\to \mathcal{H}_q$ as:
\begin{align}\label{eq:def_unitary_U}
(Uu)(\theta,m,j)=\sum_{n\equiv j (\mathrm{mod} q)} \int_{\T} e^{2\pi i(\theta+m y)n}e^{-2\pi i mx} u(x,n)\, \mathrm{d}x.
\end{align}
A straight-forward computation shows
\begin{align}
    UH^{sk}_{\lambda,y,p/q}=\widehat{H}_{\lambda,y,p/q}U.
\end{align}
In fact, 
\begin{align}
    &(UH^{sk}_{\lambda,y,p/q}u)(\theta,m,j)\\
    =&\sum_{n\equiv j (\mathrm{mod} q)}\int_{\T} e^{2\pi i(\theta+my)n} e^{-2\pi i mx} (H^{sk}_{\lambda,y,p/q}u)(x,n)\, \mathrm{d}x\\
    =&\sum_{n\equiv j (\mathrm{mod} q)}\int_{\T} e^{2\pi i(\theta+my)n} e^{-2\pi i mx}(u(x,n+1)+u(x,n-1)\\
    &\qquad +\lambda (e^{2\pi i (x+ny+n(n-1)\frac{p}{q})}+e^{-2\pi i (x+ny+n(n-1)\frac{p}{q})})u(x,n))\, \mathrm{d}x\\
    =&e^{-2\pi i(\theta+my)}(Uu)(\theta,m,j+1)+e^{2\pi i(\theta+my)}(Uu)(\theta,m,j)\\
    &+\lambda e^{2\pi i \frac{j(j-1)p}{q}} (Uu)(\theta,m+1,j)+\lambda e^{-2\pi i\frac{j(j-1)p}{q}}(Uu)(\theta,m-1,j)\\
    =&(\widehat{H}_{\lambda,y,p/q}Uu)(\theta,m,j).
\end{align}
Hence $\sigma(H^{sk}_{\lambda,y,p/q})=\sigma(\widehat{H}_{\lambda,y,p/q})$, which implies \eqref{eq:claimed_dual}, since for irrational $y$, and arbitrary $x,\theta\in \T$,
\begin{align}
\sigma(H^{sk}_{\lambda,x,y,p/q})=\cup_{x}\sigma(H^{sk}_{\lambda,x,y,p/q})=\sigma(H^{sk}_{\lambda,y,p/q})=\sigma(\widehat{H}_{\lambda,y,p/q})=\cup_{\theta}\sigma(\widehat{H}_{\lambda,\theta,y,p/q})=\sigma(\widehat{H}_{\lambda,\theta,y,p/q}).
\end{align}

\section{Applications of the non-arithmetic localization to graphene models}\label{sec:application}
Numerous models in condensed matter physics arise in block-valued form, rather than as scalar-valued operators.
We will now analyse some examples. 

\subsection{Dirac-Harper model for Moir\'e superlattice}
The following Dirac-Harper model for moir\'e bilayer superlattices was proposed in \cite{TM}:
\begin{align}
(H^{DH}_{\lambda,\theta}\phi)_n=B\phi_{n+1}+B^*\phi_{n-1}+V(\theta+n\omega) \phi_n,
\end{align}
where $\theta,\omega\in \T$, $\lambda\in \R$, and
\begin{align}
B=\left(\begin{matrix}
        0 &1 & 0 &0\\
        1 &0 &0 &0\\
        0 &0 &0 &1\\
        0 &0 &1 &0
    \end{matrix}\right), \quad V_0=
    \left(\begin{matrix}
        0 &1 &0 &0\\
        1 &0 &0 &0\\
        0 &0 &0 &1\\
        0 &0 &1 &0
    \end{matrix}\right)
\end{align}
and 
\begin{align}
    V(\theta)=&
    V_0
    +\lambda \left(\begin{matrix}
0 &0 &1+2\cos(2\pi\theta) &1-2\cos(2\pi\theta-\frac{\pi}{3})\\
0 &0 &1+2\cos(2\pi\theta+\frac{\pi}{3}) &1+2\cos(2\pi\theta)\\
1+2\cos(2\pi\theta) &1+2\cos(2\pi\theta+\frac{\pi}{3}) &0 &0\\
1-2\cos(2\pi\theta-\frac{\pi}{3})& 1+2\cos(2\pi\theta) &0 &0
    \end{matrix}\right)\\
    =:&V_0+\lambda V_1(\theta).
\end{align}
In the large coupling regime, one can apply the results of \cite{DK1,Kl} to this model and conclude uniformly positive Lyapunov exponents and Anderson localization. In fact, for $|\lambda|$ large enough, one has uniformly positive $L_4(\omega,M_E)$ of order $\log |\lambda|+O(1)$ on $\R$ by \cite[Theorem 2.3]{DK1}.
S.~Klein's result \cite{Kl} applied to the operator above for  large coupling  implies Anderson localization of $H^{DH}_{\lambda,\theta}$ for a.e.~$\omega$. To apply~\cite{DK1,Kl}, one needs to check the potential matrix $V_1$ has no constant eigenvalue. Denoting the right-hand upper $2\times 2$-block of $V_1$ by $A$, this reduces to verifying that $A^TA$ has no constant eigenvalues, which is indeed the case by an explicit computation. 
As an application of Theorem~\ref{thm:main}, we conclude that
\begin{theorem}\label{thm:H_DH}
For a.e. $\omega,\theta$, $H^{DH}_{\lambda,\theta}$ is Anderson localized in $\{E:\, L_4(\omega,M_E)>0\}$.
\end{theorem}
\begin{remark}\label{rem:DH}
    The same result holds for $H^{DH}_{\rho,\tau,\lambda_0,\lambda_1,\theta}$ below. 
\end{remark}
A modification of this model was introduced in \cite{BGW}, where
\begin{align}
    (H_{\rho,\tau,\lambda_0,\lambda_1,\theta}^{DH}\phi)_n=B_{\rho}\phi_{n+1}+B^*_{\rho}\phi_{n-1}+V(\theta+n\omega,\tau)\phi_n,
\end{align}
where $\rho,\tau,\theta,\lambda_0,\lambda_1\in \R$ and
\begin{align}
    B_{\rho}=\left(\begin{matrix}
        0 &e^{-2\pi i\rho} & 0 &0\\
        e^{2\pi i\rho} &0 &0 &0\\
        0 &0 &0 &e^{-2\pi i\rho}\\
        0 &0 &e^{2\pi i\rho} &0
    \end{matrix}\right),
\end{align}
and 
\begin{align}
    &V_{\lambda_0,\lambda_1}(\theta,\tau)=
    \left(\begin{matrix}
        0 &1 &0 &0\\
        1 &0 &0 &0\\
        0 &0 &0 &1\\
        0 &0 &1 &0
    \end{matrix}\right)\\
    &+\lambda_0 \left(\begin{matrix}
0 &0 &1+2\cos(2\pi(\theta-\tau\omega)) &0\\
0 &0 &0 &1+2\cos(2\pi(\theta+\tau\omega))\\
1+2\cos(2\pi(\theta-\tau\omega)) &0 &0 &0\\
0& 1+2\cos(2\pi(\theta+\tau\omega)) &0 &0
    \end{matrix}\right)\\
    &+\lambda_1 \left(\begin{matrix}
0 &0 &0 &1-2\cos(2\pi\theta-\frac{\pi}{3})\\
0 &0 &1+2\cos(2\pi\theta+\frac{\pi}{3}) &0\\
0 &1+2\cos(2\pi\theta+\frac{\pi}{3}) &0 &0\\
1-2\cos(2\pi\theta-\frac{\pi}{3})& 0 &0 &0
    \end{matrix}\right)
\end{align}
Becker-Ge-Wittsten proved in \cite[Theorem 1]{BGW} that for $|\lambda|$ large enough, for either the chiral coupling $(\lambda_0,\lambda_1)=(0,\lambda)$ or the anti-chiral coupling $(\lambda_0,\lambda_1)=(\lambda,0)$, $H^{DH}_{\rho,\tau,\lambda_0,\lambda_1,\theta}$ has Anderson localization. 
The proof of this theorem as stated in their paper is unfortunately wrong. The authors claimed falsely a lower bound of the denominator of the Green's function (associated to Dirichlet boundary condition) directly from the Thouless formula. 
This mistake was discussed in details in \cite[Remark 3.4]{HS3}.

Although the proof of \cite[Theorem 1]{BGW} is wrong, the theorem as stated is correct. Indeed, one can simply obtain it as a special case of \cite{Kl}, the same way as for $H^{DK}_{\lambda,\theta}$ as discussed above.
As application of our Theorem \ref{thm:main}, a stronger non-perturbative localization result can be obtained, see Remark \ref{rem:DH}.

Next, we discuss another mistake in \cite{BGW} regarding arithmetic Anderson localization. This mistake appears difficult to rectify. 
In \cite[Theorem 2]{BGW}, the authors claimed Anderson localization for the anti-chiral model $H^{DK}_{\rho,1/4,\lambda_0,0,\theta}$ (note $\lambda_1=0$ and $\tau=1/4$). 
Unfortunately the proof of this result is also wrong. 
The wrong proof is based on a chain of three critical mistakes, see \eqref{eq:wrong_1}, \eqref{eq:wrong_2}, \eqref{eq:wrong_3}.
Let us explain the mistakes briefly below.
Take $\tilde{f}^+_{E,n}(\theta)$, similarly to $f_{E,n}(\theta)$ in \eqref{def:fn}, to be the determinant with Dirichlet boundary condition (which is $p^{\mathcal{N}^+}(\theta)$ in \cite{BGW})):
\begin{align}
    \tilde{f}^+_{E,n}(\theta)=\det 
\left(\begin{matrix}
\widetilde{V}(\theta+(n-1)\omega)-E & \widetilde{B}^*_{\rho} & &  &\\
\widetilde{B}_{\rho} &\widetilde{V}(\theta+(n-2)\omega)-E &\ddots \\
& \ddots &\ddots &\ddots \\
& &\ddots &\ddots &\widetilde{B}^*_{\rho}\\
& & &\widetilde{B}_{\rho} &\widetilde{V}(\theta)-E
\end{matrix}\right)
\end{align}
Note each block is of size $4\times 4$, and $\widetilde{B}_{\rho}, \widetilde{V}$ differ from $B_{\rho},V$ by a conjugation. 
In fact
\begin{align}
    \widetilde{B}_{\rho}=\left(\begin{matrix}
        0 &0 & e^{2\pi i\rho} &0\\
        0 &0 &0 &e^{2\pi i\rho}\\
        e^{-2\pi i\rho} &0 &0 &0\\
        0 &e^{-2\pi i\rho} &0 &0
    \end{matrix}\right),
\end{align}
and
\begin{align}
    &\widetilde{V}(\theta)=
    \left(\begin{matrix}
        0 &0 &1 &0\\
        0 &0 &0 &1\\
        1 &0 &0 &0\\
        0 &1 &0 &0
    \end{matrix}\right)\\
    &+\lambda_0 \left(\begin{matrix}
0 &1+2\cos(2\pi(\theta+\frac{1}{4}\omega)) &0 &0\\
1+2\cos(2\pi(\theta+\frac{1}{4}\omega)) &0 &0 &0\\
0 &0 &0 &1+2\cos(2\pi(\theta-\frac{1}{4}\omega))\\
0& 0 &1+2\cos(2\pi(\theta-\frac{1}{4}\omega)) &0
    \end{matrix}\right)
\end{align}
Another $\tilde{f}_{E,n}^-(\theta)$ (which is $p^{\mathcal{N}^-}(\theta)$ in \cite{BGW}) was introduced in \cite[Equation (4.4)]{BGW}, where $\tilde{f}_{E,n}^-(\theta)$ is the determinant of the shifted (by $2$ units) matrix:
\begin{align}
    &\tilde{f}_{E,n}^-(\theta)=
    \det \\
&\left(\begin{matrix}
\mathcal{P}_2 \widetilde{V}(\theta+n\omega)\mathcal{P}_2^* &\mathcal{P}_2 \widetilde{B}_{\rho}^*\\
\widetilde{B}_{\rho}\mathcal{P}_2^* &\widetilde{V}(\theta+(n-1)\omega)-E & \widetilde{B}^*_{\rho} & &  &\\
&\widetilde{B}_{\rho} &\widetilde{V}(\theta+(n-2)\omega)-E &\ddots \\
&& \ddots &\ddots &\ddots \\
&& &\ddots &\widetilde{V}(\theta+\omega)-E &\widetilde{B}^*_{\rho}\mathcal{P}_1^*\\
&& & &\mathcal{P}_1 \widetilde{B}_{\rho} &\mathcal{P}_1(\widetilde{V}(\theta)-E)\mathcal{P}_1^*
\end{matrix}\right),
\end{align}
where $\mathcal{P}_1: \C^4\to \C^2$ is the projection onto the first two coordinates and $\mathcal{P}_2: \C^4\to \C^2$ is the projection onto the last two coordinates.

It was first falsely claimed on top of Page 20 of \cite{BGW} that (note their $1/L$ is our $\omega$)
\begin{align}\label{eq:wrong_1}
    \tilde{f}_{E,n}^+(\theta)=\tilde{f}_{E,n}^-(\theta-\frac{1}{2}\omega).
\end{align} 
To see this is wrong, without loss of generality, we let $\rho=0$.
Let 
\begin{align}
    C_j=\left(\begin{matrix}
        -E &\lambda_0(1+2\cos(2\pi(\theta-\frac{1}{4}\omega+\frac{j}{2}\omega)))\\ \lambda_0(1+2\cos(2\pi(\theta-\frac{1}{4}\omega+\frac{j}{2}\omega))) &-E
    \end{matrix}\right)
\end{align}
Then, for example when $n=4$, with $I=I_2$ the $2\times 2$ identity matrix,
\begin{align}
    \tilde{f}_{E,4}^+(\theta+\frac{1}{2}\omega)=\det
    \left(\begin{array}{c|c|c|c|c|c|c|c}
    C_8 &I & & I & & & &\\
    \hline
    I &C_7 &I & & & & &\\
    \hline
      &I &C_6 &I & &I & &\\
      \hline
    I & &I &C_5 &I & & &\\
    \hline
      & & &I &C_4 &I & &I\\
      \hline
      & &I & &I &C_3 &I &\\
      \hline
      & & & & &I &C_2 &I\\ 
      \hline
      & & & &I & &I  &C_1
    \end{array}\right)
\end{align}
and 
\begin{align}
    \tilde{f}_{E,4}^-(\theta)=
    \det
    \left(\begin{array}{c|c|c|c|c|c|c|c}
    C_8 &I & & & & & &\\
    \hline
    I &C_7 &I & & I & & & \\
    \hline
    &I &C_6 &I & & & & \\
    \hline
    &  &I &C_5 &I & &I & \\
      \hline
    &I & &I &C_4 &I & & \\
    \hline
    &  & & &I &C_3 &I & \\
      \hline
    &  & &I & &I &C_2 &I \\
      \hline
    &  & & & & &I &C_1 \\ 
    \end{array}\right)
\end{align}
The two determinants are not equal to each other. In fact, for  $\lambda_0=0$ one computes that 
\[
 \tilde{f}_{E,4}^+(\theta+\frac{1}{2}\omega) = (E^{8} - 10E^{6} + 23E^4 - 10E^2 + 1)^2
\]
while 
\[
  \tilde{f}_{E,4}^-(\theta) = (E^{8} - 9E^{6} + 18E^4- 8E^2 + 1)^2.
\]
This is the first critical mistake that invalidates their entire proof.

Even under the false \eqref{eq:wrong_1}, the authors in \cite{BGW} made another false claim that \begin{align}\label{eq:wrong_2}
\tilde{f}_{E,n}^{\pm} \text{ is an even functions of }\theta+\frac{n-1}{2}\omega.
\end{align} 

Indeed, one has
\begin{align}\label{eq:f+_even}
    \tilde{f}_{E,n}^+(\theta-\frac{n-1}{2}\omega)=\tilde{f}_{E,n}^+(-\theta-\frac{n-1}{2}\omega),
\end{align} 
but instead of the falsely claimed
\begin{align}\label{eq:f-_even_wrong}
    \tilde{f}_{E,n}^-(\theta-\frac{n-1}{2}\omega)=\tilde{f}_{E,n}^-(-\theta-\frac{n-1}{2}\omega),
\end{align}
in their paper, one in fact has
\begin{align}\label{eq:f-_even}
    \tilde{f}_{E,n}^-(\theta-\frac{n}{2}\omega)=\tilde{f}_{E,n}^-(-\theta-\frac{n}{2}\omega).
\end{align}
This is their second critical mistake. 

Even under the false \eqref{eq:wrong_1}, \eqref{eq:wrong_2}, the authors made yet another false claim that
\begin{align}\label{eq:wrong_3}
\tilde{f}_{E,n}^-(\theta)=\tilde{f}_{E,n}^-(\theta+\frac{1}{2}).
\end{align}
However, even under the incorrect equality \eqref{eq:wrong_1} and \eqref{eq:wrong_2}, one should obtain $\tilde{f}_{E,n}^-(\theta)=\tilde{f}^-_{E,n}(\theta+\omega)$ instead of \eqref{eq:wrong_3}.
This is their third critical mistake.

The proof of \cite[Theorem 2]{BGW}, relying crucially on a chain of erroneous equations~\eqref{eq:wrong_1},\eqref{eq:wrong_2},\eqref{eq:wrong_3}, is therefore completely false.

\subsection{AA-stacked graphene in magnetic fields} AA-stacked graphene model in magnetic fields has attracted a lot attention in the physics literature, see e.g. \cite{RSRN} and the references therein.
The Hamiltonian is the following: 
\begin{align*}
(H^{AA} u)^1_{m,n,A}=&\lambda_1 u^1_{m,n,B}+\lambda_2 u^1_{m-1,n,B}+e^{2\pi i m\omega}\lambda_3 u^1_{m,n+1,B}+\rho u^2_{m,n,A},\\
(H^{AA} u)^1_{m,n,B}=&\lambda_1 u^1_{m,n,A}+\lambda_2 u^1_{m+1,n,A}+e^{-2\pi i m\omega} \lambda_3 u^1_{m,n-1,A}+\rho u^2_{m,n,B}\\
(H^{AA} u)^2_{m,n,A}=&\mu_1 u^2_{m,n,B}+\mu_2 u^2_{m-1,n,B}+e^{2\pi i m\omega}\mu_3 u^2_{m,n+1,B}+\rho u^1_{m,n,A},\\
(H^{AA} u)^2_{m,n,B}=&\mu_1 u^2_{m,n,A}+\mu_2 u^2_{m+1,n,A}+e^{-2\pi i m\omega} \mu_3 u^2_{m,n-1,A}+\rho u^1_{m,n,B},
\end{align*}
where $(\lambda_1,\lambda_2,\lambda_3)$ are the intra-layer coupling constants within the first layer, and $(\mu_1,\mu_2,\mu_3)$ are the intra-layer coupling constants within the second layer, and $\rho>0$ is the inter-layer coupling. 
In the AA-stacked setting, the two layers are identical to each other, we will also assume $(\lambda_1,\lambda_2,\lambda_3)=(\mu_1,\mu_2,\mu_3)$. 

Reducing to a one-dimensional operator, via taking a Fourier transform in the variable $n$, we have
\begin{align}
    (H^{AA}_{\theta}u)^1_{m,A}&=\lambda_1 u_{m,B}^1+\lambda_2 u_{m-1,B}^1+\lambda_3 e^{2\pi i(\theta+m\omega)} u_{m,B}^1+\rho u_{m,A}^2\\
    (H^{AA}_{\theta}u)^1_{m,B}&=\lambda_1 u_{m,A}^1+\lambda_2 u_{m+1,A}^1+\lambda_3 e^{-2\pi i(\theta+m\omega)}u_{m,A}^1+\rho u_{m,B}^2\\
    (H^{AA}_{\theta}u)^2_{m,A}&=\lambda_1 u_{m,B}^2+\lambda_2 u_{m-1,B}^2+\lambda_3 e^{2\pi i(\theta+n\omega)} u_{m,B}^2+\rho u_{m,A}^1\\
    (H^{AA}_{\theta}u)^2_{m,B}&=\lambda_1 u_{m,A}^2+\lambda_2 u_{m+1,A}^2+\lambda_3 e^{-2\pi i(\theta+n\omega)} u_{m,A}^2+\rho u_{m,B}^2
\end{align}
It is known that $\sigma(H^{AA})=\bigcup_{\theta}\sigma(H^{AA}_{\theta})$.
Rewriting $H^{AA}_{\theta}$ in terms of block Jacobi matrix, one has
\begin{align}
\left(\begin{matrix}
    (H_{\theta}^{AA}u)_{m,A}^1\\
    (H_{\theta}^{AA}u)_{m,B}^2
\end{matrix}\right)
=
\left(\begin{matrix}
0 &0\\
0 &\lambda_2\end{matrix}\right)
\left(\begin{matrix}
    u_{m+1,B}^1\\ u_{m+1,A}^2
\end{matrix}\right)+
\left(\begin{matrix}
c(\theta+m\omega)&\rho\\
    \rho &d(\theta+m\omega)
\end{matrix}\right)
\left(\begin{matrix}
    u_{m,B}^1\\ u_{m,A}^2
\end{matrix}\right)
+\left(\begin{matrix}
\lambda_2 &0\\
0 &0\end{matrix}\right)
\left(\begin{matrix}
    u_{m-1,B}^1\\ u_{m-1,A}^2
\end{matrix}\right),
\end{align}
and
\begin{align}
  \left(\begin{matrix}
    (H_{\theta}^{AA}u)_{m,B}^1\\
    (H_{\theta}^{AA}u)_{m,A}^2
\end{matrix}\right)
=
\left(\begin{matrix}
\lambda_2 &0\\
0 &0\end{matrix}\right)
\left(\begin{matrix}
    u_{m+1,A}^1\\ u_{m+1,B}^2
\end{matrix}\right)+
\left(\begin{matrix}
d(\theta+m\omega) &\rho\\
    \rho &c(\theta+m\omega)
\end{matrix}\right)
\left(\begin{matrix}
    u_{m,A}^1\\ u_{m,B}^2
\end{matrix}\right)
+\left(\begin{matrix}
0 &0\\
0 &\lambda_2\end{matrix}\right)
\left(\begin{matrix}
    u_{m-1,A}^1\\ u_{m-1,B}^2
\end{matrix}\right),
\end{align}
in which $c(\theta):=\lambda_1+\lambda_3e^{2\pi i\theta}$ and $d(\theta)=\lambda_1+\lambda_3e^{-2\pi i\theta}$.
Clearly $d(\theta)=\overline{c(\theta)}$ holds iff $\theta\in \T$. Later we will complexify $\theta$ to estimate the Lyapunov exponent.

If we define $\widehat{H}_{\theta}$ acting on $\ell^2(\Z, \C^2)$ as
\begin{align}
(\widehat{H}_{\theta}U)_m
    =\left(\begin{matrix}
0 &0\\
0 &\lambda_2\end{matrix}\right)
U_{m+1}+
\left(\begin{matrix}
c(\theta+m\omega)&\rho\\
    \rho &d(\theta+m\omega)
\end{matrix}\right)
U_m
+\lambda_2\left(\begin{matrix}
1 &0\\
0 &0\end{matrix}\right)
U_{m-1}
\end{align}
Then from the calculations above, due to the bipartite nature of the AA-stacked graphene lattice, it is clear that $H_{\theta}^{AA}$ can be written in the following form:
\begin{align}\label{eq:HAA=HH}
    H_{\theta}^{AA}
    \left(\begin{matrix}
U_1\\
U_2
\end{matrix}\right)=\left(\begin{matrix} 0 & \widehat{H}_{\theta}\\ \widehat{H}_{\theta}^* &0\end{matrix}\right)    \left(\begin{matrix}
U_1\\
U_2
\end{matrix}\right),
\end{align}
where $U_1=(...,u_{m+1,A}^1, u_{m+1,B}^2, u_{m,A}^1, u_{m,B}^2,...)^T$ and $U_2=(...,u_{m+1,B}^1, u_{m+1,A}^2, u_{m,B}^1, u_{m,A}^2,...)^T$.
Hence
\begin{align}\label{eq:HAA^2=HH}
    (H_{\theta}^{AA})^2= \left(\begin{matrix}
    \widehat{H}_{\theta}\widehat{H}_{\theta}^* &0 \\ 0 &\widehat{H}_{\theta}^*\widehat{H}_{\theta}
    \end{matrix}\right).
\end{align}
Clearly 
\begin{align}
\label{eq:sig_HAA_1}(\sigma(H_{\theta}^{AA}))^2=\sigma((H_{\theta}^{AA})^2)=\sigma(\widehat{H}_{\theta}^*\widehat{H}_{\theta})\cup\sigma(\widehat{H}_{\theta}\widehat{H}_{\theta}^*),
\end{align}
in which for a set $U\subset \R$, $U^2:=\{E^2:\, E\in U\}$.
Since $(\widehat{H}_{\theta}^*\widehat{H}_{\theta})|_{(\mathrm{ker}\widehat{H}_{\theta})^{\perp}}$ and $(\widehat{H}_{\theta}\widehat{H}_{\theta}^*)|_{(\mathrm{ker}\widehat{H}_{\theta}^*)^{\perp}}$ are unitarily equivalent, we conclude from 
\eqref{eq:sig_HAA_1} that
\begin{align}\label{eq:sig_AA_2}
    (\sigma(H_{\theta}^{AA}))^2\setminus \{0\}=\sigma(\widehat{H}_{\theta}\widehat{H}_{\theta}^*)\setminus \{0\}.
\end{align}
For the AA-stacked graphene, in general it is hard to tell if the zero energy is in the spectrum $\sigma(H_{\theta}^{AA})$. In fact, we have the following criterion, in terms of the single layer Hamiltonian:
\begin{lemma}\label{lem:0_AA}
    Let $H_{\theta}^{g}$ be the single-layer Hamiltonian: 
    \begin{align}
    (H_{\theta}^{g}u)_{m,A}=&\lambda_1 u_{m,B}+\lambda_2 u_{m-1,B}+\lambda_3 e^{2\pi i(\theta+m\omega)}u_{m,B}\\
    (H_{\theta}^{g}u)_{m,B}=&\lambda_1 u_{m,A}+\lambda_2 u_{m+1,A}+\lambda_3 e^{-2\pi i(\theta+m\omega)}u_{m,A}.
\end{align}
Then 
\[\sigma(H_{\theta}^{AA})=(\sigma(H_{\theta}^g)+\rho)\cup (\sigma(H_{\theta}^g)-\rho).\] 
\end{lemma}
\begin{proof}
In fact, $H_{\theta}^{AA}$ can be expressed in terms of $H_{\theta}^g$ as follows:
\begin{align}\label{eq:HAA=Hg}
    H_{\theta}^{AA}\left(\begin{matrix}U_1\\ U_2\end{matrix}\right)=\left(\begin{matrix}
        H_{\theta}^g & \rho I\\
        \rho I &H_{\theta}^g
    \end{matrix}\right)\left(\begin{matrix} U_1\\ U_2\end{matrix}\right),
\end{align}
where $U_1$ and $U_2$ are the wave functions on the first/second layer respectively.
It is easy to see that claimed result follows from \eqref{eq:HAA=Hg}.
\end{proof}
\begin{remark}
    By Lemma \ref{lem:0_AA}, $E=0\in \sigma(H_{\theta}^{AA})$ iff $E=\pm \rho\in \sigma(H_{\theta}^g)$. In general, one would conjecture that $\sigma(H_{\theta}^g)$ is a Cantor set, which has been proved in \cite{BHJ} for the isotropic case $\lambda_1=\lambda_2=\lambda_3$. Hence it is in general a hard problem to tell if a particular energy (aside from $E=0$, which is always in $\sigma(H_{\theta}^g)$ for irrational $\omega$) falls in $\sigma(H_{\theta}^g)$.
\end{remark}

Aside from the zero energy, by \eqref{eq:sig_AA_2}, we are reduced to study the following operator:
\begin{align}
    (\widehat{H}_{\theta} \widehat{H}_{\theta}^*U)_m=B(\theta+(m+1)\omega)U_{m+1}+V(\theta+m\omega)U_m+B^{(*)}(\theta+m\omega)U_{m-1},
\end{align}
where 
\begin{align}\label{def:B_AA}
B(\theta)=\lambda_2 \left(\begin{matrix}c(\theta-\omega) &0\\ \rho &c(\theta)\end{matrix}\right),\, B^{(*)}(\theta)=\lambda_2 \left(\begin{matrix} d(\theta-\omega) &\rho\\ 0 &d(\theta)\end{matrix}\right),
\end{align}
and
\begin{align}
    V(\theta)=\left(\begin{matrix} c(\theta)d(\theta)+\rho^2+\lambda_2^2 & 2\rho\, c(\theta)\\ 2\rho\, d(\theta) &c(\theta)d(\theta)+\rho^2+\lambda_2^2
    \end{matrix}\right).
\end{align}
Let $M_{E,AA}^{HH^*}$ be the transfer matrix corresponding to $\widehat{H}_{\theta}\widehat{H}_{\theta}^*U=EU$.
As a corollary of Theorem \ref{thm:main}, one has the following fact, regarding the operator $\widehat{H}_{\theta}\widehat{H}_{\theta}^*$.
\begin{theorem}\label{thm:HH_loc}
Let $(\lambda_1,\lambda_2,\lambda_3)$ be such that $\det B(\theta)\neq 0$ on $\T$, for $B$ as in \eqref{def:B_AA}.
    For a.e. $\omega,\theta\in \T$, $\widehat{H}_{\theta}\widehat{H}_{\theta}^*$ is Anderson localized in $\{E:\, L_2(\omega, M_{E,AA}^{HH^*})>0\}$.
\end{theorem}
As a corollary, we have the following result about $H_{\theta}^{AA}$.
\begin{theorem}\label{thm:AA_loc}
Let $(\lambda_1,\lambda_2,\lambda_3)$ be such that $\det B(\theta)\neq 0$ on $\T$, for $B$ as in \eqref{def:B_AA}.
    For a.e. $\omega,\theta\in \T$, $H_{\theta}^{AA}$ is Anderson localized in $\{E:\, E\neq 0, \text{ and }  L_2(\omega,M_{E^2,AA}^{HH^*})>0\}$.
\end{theorem}
\begin{proof}
    Let $E\neq 0$ and $U$ be a (non-trivial) generalized eigenfunction solving $H_{\theta}^{AA}U=EU$. Our goal is to show that it decays exponentially. Let $U_1=(...,u_{m+1,A}^1, u_{m+1,B}^2, u_{m,A}^1, u_{m,B}^2,...)^T$. By \eqref{eq:HAA=HH} and \eqref{eq:HAA^2=HH}, $U_1$ is a generalized solution to \[\widehat{H}_{\theta}\widehat{H}_{\theta}^*U_1=E^2U_1.\]
    Since we assume $L_2(\omega,M_{E^2,AA}^{HH^*})>0$, by Theorem \ref{thm:HH_loc}, $U_1$ decays exponentially.
    By \eqref{eq:HAA=HH}, $EU_2=\widehat{H}^*_{\theta}U_1$, hence $U_2$ decays exponentially as well. 
\end{proof}

Next, we exhibit regions of parameters for which the conditions of Theorems \ref{thm:HH_loc} and \ref{thm:AA_loc} are satisfied.
\begin{lemma}\label{lem:AA_positive}
    Let $\lambda_2,\rho\in \R\setminus\{0\}$ be fixed. Let $|\tilde{\lambda}_1|>|\tilde{\lambda}_3|\neq 0$ and $(\lambda_1,\lambda_3)=\lambda(\tilde{\lambda}_1,\tilde{\lambda}_3)$. Then for $\lambda>\lambda_0=\lambda_0(|\lambda_2|,|\rho|,|\tilde{\lambda}_1|-|\tilde{\lambda}_3|)$, we have for $B$ as in \eqref{def:B_AA},
    \[\det B(\theta)\neq 0, \text{ for } \theta\in \T,\]
    and
    \begin{align}\label{eq:L_2>0_AA}
        L_2(\omega,M_{E,AA}^{HH^*})>0, \text{ uniformly in } \{E: |E|\leq 10(\lambda_1^2+\lambda_2^2+\lambda_3^2+\rho^2)\}\supset\sigma(H^{AA}_{\theta}).
    \end{align}
\end{lemma}
\begin{proof}
Let $\varepsilon_1<0$ be such that $|\tilde{\lambda}_3|e^{-2\pi\varepsilon_1}=|\tilde{\lambda}_1|$.
Note that for $\varepsilon\neq \varepsilon_1$, $c(\theta+i\varepsilon)\neq 0$ for any $\theta\in \T$. Hence in particular 
\[|\det (B(\theta+i\varepsilon))|\neq 0 \text{ for any } \varepsilon>\varepsilon_1.\]
Next, we verify \eqref{eq:L_2>0_AA}.
It is easy to verify asymptotically in $\varepsilon\to\infty$ that
\begin{align}
B^{(*)}(\theta+i\varepsilon)=e^{2\pi \varepsilon}e^{-2\pi i\theta}\lambda_2\lambda_3\left(\begin{matrix} e^{2\pi i\omega}& 0\\
0 &1\end{matrix}\right)+O(1), \text{ as } \varepsilon\to\infty,
\end{align}
\begin{align}
    B(\theta+i\varepsilon)^{-1}=\frac{1}{\lambda_1^2\lambda_2}\left(\begin{matrix} 
    \lambda_1 &0\\ 
    -\rho &\lambda_1 \end{matrix}\right)+O(e^{-2\pi\varepsilon}), \text{ as } \varepsilon\to\infty,
\end{align}
and
\begin{align}
    V(\theta+i\varepsilon)=e^{2\pi\varepsilon}e^{-2\pi i\theta}\lambda_3\left(\begin{matrix}
        \lambda_1 &0\\
        2\rho &\lambda_1 
    \end{matrix}\right)+O(1), \text{ as } \varepsilon\to\infty.
\end{align}
Hence
\begin{align}
M_{E,AA}^{HH^*}(\theta+i\varepsilon)
=&e^{2\pi\varepsilon}e^{-2\pi i\theta}
\left(\begin{array}{c|c}
    \left(\begin{matrix}-\frac{\lambda_3}{\lambda_2} &0\\
    -\frac{\rho \lambda_3}{\lambda_1\lambda_2} &-\frac{\lambda_3}{\lambda_2}\end{matrix}\right)
    &\left(\begin{matrix}
        -\lambda_2\lambda_3 e^{2\pi i\omega} &0\\
        0 &-\lambda_2\lambda_3
    \end{matrix}\right)\\
    \hline
    0_{2\times 2} &0_{2\times 2}
\end{array}\right)+O(1)\\
=:&e^{2\pi\varepsilon} e^{-2\pi i\theta}
\left(\begin{array}{c|c}
    Q_1 & Q_2\\
    \hline
    0_{2\times 2} & 0_{2\times 2}
\end{array}\right)+O(1)=:e^{2\pi\varepsilon} e^{-2\pi i\theta}Q+O(1).
\end{align}
By the continuity of Lyapunov exponents in the cocycles \cite{AJS}, we have
\begin{align}\label{eq:lb_L^2_AA}
    L^2_{\varepsilon}(\omega,M_{E,AA}^{HH^*})=4\pi\varepsilon+L^2(\omega,e^{-2\pi i\theta}Q)+o(1), \text{ as } \varepsilon\to\infty.
\end{align}
It suffices to compute
\begin{align}\label{eq:L2Q_AA}
    L^2(\omega,e^{-2\pi i\theta}Q)=\lim_{n\to\infty} \frac{1}{n} \log \|\textstyle{\bigwedge}^2 Q^n\|
    =&\lim_{n\to \infty}\frac{1}{n}\log \left\|\textstyle{\bigwedge}^2\left(\begin{array}{c|c}Q_1^n & Q_1^{n-1}Q_2\\ \hline 0 &0\end{array}\right)\right\|\notag\\
    =&2\log |\lambda_3/\lambda_2|. 
\end{align}
Combining \eqref{eq:lb_L^2_AA} with \eqref{eq:L2Q_AA} yields
\begin{align}\label{eq:lb_L^2_AA_2}
    L^2_{\varepsilon}(\omega,M_{E,AA}^{HH^*})=4\pi\varepsilon+2\log |\lambda_3/\lambda_2|+o(1)
\end{align}
By the convexity of $L^2_{\varepsilon}(\omega,M_{E,AA}^{HH^*})$ in $\varepsilon\in (\varepsilon_1,\infty)$, we have,
\begin{align}\label{eq:lb_L^2_AA_3}
    L^2_{\varepsilon=0}(\omega,M_{E,AA}^{HH^*})\geq 2\log |\lambda_3/\lambda_2|\geq 2\log |\lambda|+O(1).
\end{align}
Simple estimates on the sup norm of $M_{E,AA}^{HH^*}$ show
\begin{align}\label{eq:ub_L^1_AA}
    \|M_{E,AA}^{HH^*}(\cdot)\|_{\T,\infty}\leq \log \left(\frac{(|\lambda_1|+|\lambda_2|+|\lambda_3|+|\rho|)^3}{(|\lambda_1|-|\lambda_3|)^2}\right)+O(1)\leq \log |\lambda|+O(1),
    \end{align}
    uniformly in $|E|\leq 10(\lambda_1^2+\lambda_3^2+\lambda_2^2+\rho^2)$.
Combining \eqref{eq:lb_L^2_AA_3} with \eqref{eq:ub_L^1_AA}, we conclude that for $|\lambda|$ large enough
\begin{align}
    L_2(\omega,M_{E,AA}^{HH^*})\geq \log |\lambda|+O(1),
\end{align}
uniformly in $E$ in the interval specified above.
\end{proof}

\subsection{AB-stacked graphene in magnetic fields}
The AB-stacked graphene model has received a lot of attention in the physics literature as well, see e.g. \cite{LHCL}.
The Hamiltonian for the AB-stacked graphene model in magnetic fields is: 
\begin{align}
    (H^{AB}u)_{m,n,A}^2=&\lambda_1 u_{m,n,B}^2+\lambda_2 u_{m+1,n,B}^2+\lambda_3 e^{2\pi i m\omega} u_{m,n+1,B}^2+\rho u_{m,n,B}^1\\
    (H^{AB}u)_{m,n,B}^2=&\lambda_1 u_{m,n,A}^2+\lambda_2 u_{m-1,n,A}^2+\lambda_3 e^{-2\pi im\alpha}u_{m,n-1,A}^2\\
    (H^{AB}u)_{m,n,A}^1=&\mu_1 e^{2\pi i\frac{1}{3}\alpha}u_{m,n,B}^1+\mu_2 u_{m+1,n,B}^1+\mu_3 e^{2\pi i(m\alpha-\frac{2}{3}\alpha)}u_{m,n+1,B}^1\\
    (H^{AB}u)_{m,n,B}^1=&\mu_1 e^{-2\pi i \frac{1}{3}\alpha}u_{m,n,A}^1+\mu_2 u_{m-1,n,A}^1+\mu_3 e^{-2\pi i(m\alpha-\frac{2}{3}\alpha)}u_{m,n-1,A}^1+\rho u_{m,n,A}^2.
\end{align}
Note in this model, we only introduce inter-layer hopping when a vertex is exactly on top of another in the other layer.
We assume the two single layers are identical, hence $(\mu_1,\mu_2,\mu_3)=(\lambda_1,\lambda_2,\lambda_3)$.

The operator $H^{AB}$ can be reduced to a one-dimensional operator as
\begin{align}
    (H^{AB}_{\theta}u)_{m,A}^2=&\lambda_2 u_{m+1,B}^2+(\lambda_3 e^{2\pi i(\theta+m\omega)}+\lambda_1)u_{m,B}^2+ \rho u_{m,B}^1\\
    (H^{AB}_{\theta}u)_{m,B}^2=&\lambda_2 u_{m-1,A}^2+(\lambda_3 e^{-2\pi i(\theta+m\omega)}+\lambda_1)u_{m,A}^2\\
    (H^{AB}_{\theta}u)_{m,A}^1=&\lambda_2 u_{m+1,B}^1+(\lambda_3 e^{2\pi i(\theta+(m-\frac{2}{3})\omega)}+\lambda_1 e^{2\pi i\frac{1}{3}\omega})u_{m,B}^1\\
    (H^{AB}_{\theta}u)_{m,B}^1=&\lambda_2 u_{m-1,A}^1+(\lambda_3 e^{-2\pi i(\theta+(m-\frac{2}{3})\omega)}+\lambda_1 e^{-2\pi i\frac{1}{3}\omega})u_{m,A}^1+\rho u_{m,A}^2.
\end{align}
The spectrum is preserved in the sense that $\sigma(H^{AB})=\bigcup_{\theta\in \T}\sigma(H^{AB}_{\theta})$.

In terms of block-valued operator, $H_{\theta}^{AB}$ reads as
\begin{align}
    \left(\begin{matrix}(H^{AB}_{\theta}u)_{m,B}^2\\ (H^{AB}_{\theta}u)_{m,B}^1\end{matrix}\right)_m=
    \left(\begin{matrix}
d(\theta+m\omega) &0\\
\rho &e^{-2\pi i\frac{1}{3}\omega}d(\theta+(m-1)\omega)
    \end{matrix}\right)\left(\begin{matrix}
        u_{m,A}^2\\ u_{m,A}^1
    \end{matrix}\right)
+\lambda_2\left(\begin{matrix}
    u_{m-1,A}^2\\ u_{m-1,A}^1
    \end{matrix}\right),
\end{align}
and
\begin{align}
\left(\begin{matrix}(H^{AB}_{\theta}u)_{m,A}^2\\ (H^{AB}_{\theta}u)_{m,A}^1\end{matrix}\right)_m=
    \left(\begin{matrix}
c(\theta+m\omega) &\rho\\
0 &e^{2\pi i\frac{1}{3}\omega}c(\theta+(m-1)\omega)
    \end{matrix}\right)\left(\begin{matrix}
        u_{m,A}^2\\ u_{m,A}^1
    \end{matrix}\right)
+\lambda_2\left(\begin{matrix}
    u_{m+1,B}^2\\ u_{m+1,B}^1
    \end{matrix}\right),
\end{align}
in which $c(\theta)=\lambda_1+\lambda_3 e^{2\pi i\theta}$ and $d(\theta)=\lambda_1+\lambda_3 e^{-2\pi i\theta}$.
Let $\widehat{H}_{\theta}$ on $\ell^2(\Z, \C^2)$ be as follows:
\begin{align}
    (\widehat{H}_{\theta}U)_m=\left(\begin{matrix}
d(\theta+m\omega) &0\\
\rho &e^{-2\pi i\frac{1}{3}\omega}d(\theta+(m-1)\omega)
    \end{matrix}\right)U_m+\lambda_2 U_{m-1},
\end{align}
then similarly to \eqref{eq:HAA=HH}, 
\begin{align}\label{eq:HAB=HH}
    H_{\theta}^{AB}\left(\begin{matrix} U_1\\ U_2\end{matrix}\right)
    =\left(\begin{matrix} 0 &\widehat{H}_{\theta}\\ \widehat{H}_{\theta}^* &0\end{matrix}\right)
    \left(\begin{matrix} U_1\\ U_2\end{matrix}\right),
\end{align}
where $U_1=(...,u_{m+1,B}^2, u_{m+1,B}^1, u_{m,B}^2, u_{m,B}^1,...)^T$ and $U_2=(...,u_{m+1,A}^2, u_{m+1,A}^1, u_{m,A}^2, u_{m,A}^1,...)^T$.
This leads to 
\begin{align}\label{eq:HAB=HH_2}   (H_{\theta}^{AB})^2=\left(\begin{matrix}\widehat{H}_{\theta}\widehat{H}_{\theta}^* &0 \\ 0 &\widehat{H}_{\theta}^* \widehat{H}_{\theta}\end{matrix}\right).
\end{align}
Therefore, in analogy with \eqref{eq:sig_AA_2},  
\begin{align}\label{eq:sig_AB}
    (\sigma(H_{\theta}^{AB}))^2\setminus \{0\}=\sigma((H_{\theta}^{AB})^2)\setminus \{0\}=\sigma(\widehat{H}_{\theta}\widehat{H}_{\theta}^*)\setminus \{0\}.
\end{align}
The operator $\widehat{H}_{\theta}\widehat{H}_{\theta}^*$ takes the following form:
\begin{align}
(\widehat{H}_{\theta}\widehat{H}_{\theta}^*U)_m=B(\theta+(m+1)\omega)U_{m+1}+V(\theta+m\omega)U_m+B^{(*)}(\theta+m\omega)U_{m-1},
\end{align}
where
\begin{align}\label{def:B_AB}
    B(\theta)=\lambda_2\left(\begin{matrix}
        d(\theta-\omega) & 0\\
        \rho &e^{-2\pi i\frac{1}{3}\omega}d(\theta-2\omega)
    \end{matrix}\right),\, B^{(*)}(\theta)=\lambda_2 
    \left(\begin{matrix} c(\theta-\omega) & \rho\\
    0 &e^{2\pi i \frac{1}{3}\omega} c(\theta-2\omega)
    \end{matrix}\right),
\end{align}
and
\begin{align}
    V(\theta)=
    \left(\begin{matrix}
        c(\theta)d(\theta)+\lambda_2^2 &d(\theta)\\
        c(\theta) &\rho^2+\lambda_2^2+c(\theta-\omega)d(\theta-\omega)
    \end{matrix}\right).
\end{align}
Let $M_{E,AB}^{HH^*}$ be the cocycle associated to $\widehat{H}_{\theta}\widehat{H}_{\theta}^*U=EU$. 
Then similarly to Theorems \ref{thm:HH_loc} and \ref{thm:AA_loc},  we obtain 
\begin{theorem}\label{thm:HH_loc_2}
    Let $(\lambda_1,\lambda_2,\lambda_3)$ be such that $\det B(\theta)\neq 0$ on $\T$, for $B$ as in \eqref{def:B_AB}.
    For a.e. $\omega,\theta\in \T$, $\widehat{H}_{\theta}\widehat{H}_{\theta}^*$ is Anderson localized in $\{E:\, L_2(\omega, M_{E,AB}^{HH^*})>0\}$.
\end{theorem}
\begin{theorem}\label{thm:AB_loc}
    Let $(\lambda_1,\lambda_2,\lambda_3)$ be such that $\det B(\theta)\neq 0$ on $\T$, for $B$ as in \eqref{def:B_AB}.
    For a.e. $\omega,\theta\in \T$, $H_{\theta}^{AB}$ is Anderson localized in $\{E:\, E\neq 0, \text{ and }  L_2(\omega,M_{E^2,AB}^{HH^*})>0\}$.
\end{theorem}
Since the proofs are analogous to those of Theorems \ref{thm:HH_loc} and \ref{thm:AA_loc}, we don't repeat them here.

In view of Lemma \ref{lem:AA_positive}, the purpose of the following lemma is to show that for some regions of the parameters the conditions of Theorems \ref{thm:HH_loc_2} and \ref{thm:AB_loc} are satisfied.
\begin{lemma}\label{lem:AB_positive}
Let $\lambda_2,\rho\in \R\setminus\{0\}$ be fixed. Let $|\tilde{\lambda}_3|>|\tilde{\lambda}_1|\neq 0$ and $(\lambda_1,\lambda_3)=\lambda(\tilde{\lambda}_1,\tilde{\lambda}_3)$. Then for $\lambda>\lambda_0=\lambda_0(|\lambda_2|,|\rho|,|\tilde{\lambda}_3|-|\tilde{\lambda}_1|)$, we have for $B$ as in \eqref{def:B_AB},
    \[\det B(\theta)\neq 0, \text{ for } \theta\in \T,\]
    and
    \begin{align}\label{eq:L_2>0_AB}
        L_2(\omega,M_{E,AB}^{HH^*})>0, \text{ uniformly in } \{E: |E|\leq 10(\lambda_1^2+\lambda_2^2+\lambda_3^2+\rho^2)\}\supset\sigma(H^{AB}_{\theta}).
    \end{align}
\end{lemma}
\begin{proof}
Let $\varepsilon_1<0$ be such that $|\tilde{\lambda}_1|=|\tilde{\lambda}_3|e^{2\pi \varepsilon_1}$. Clearly for any $\varepsilon\neq \varepsilon_1$, $d(\theta)\neq 0$ for any $\theta\in \T$. 
Hence $\det B(\theta+i\varepsilon)\neq 0$ for any $\theta\in \T$ and $\varepsilon>\varepsilon_1$. 
One computes asymptotically in $\varepsilon\to\infty$ that 
    \begin{align}
B(\theta+i\varepsilon)=e^{2\pi\varepsilon}\lambda_2\lambda_3\left(\left(\begin{matrix}
        e^{-2\pi i(\theta-\omega)} &0\\
        0 &e^{-2\pi i(\theta-2\omega)}
        \end{matrix}\right)+o(1)\right),\, \text{ as } \varepsilon\to \infty, 
    \end{align}
\begin{align}
    B^{(*)}(\theta+i\varepsilon)=\lambda_2\left(\begin{matrix}\lambda_1 & \rho\\ 0 &e^{2\pi i\frac{1}{3}\omega}\lambda_1
    \end{matrix}\right)+o(1),\, \text{ as } \varepsilon\to \infty, 
\end{align}
\begin{align}
    B(\theta+i\varepsilon)^{-1}=e^{-2\pi\varepsilon}\frac{e^{2\pi i\theta}}{\lambda_2\lambda_3}\left(\left(\begin{matrix}
        e^{-2\pi i\omega} &0\\
        0 &e^{2\pi i\frac{1}{3}\omega} e^{-4\pi i\omega}
    \end{matrix}\right)+o(1)\right)
\end{align}
    and 
    \begin{align}
V(\theta+i\varepsilon)=e^{2\pi\varepsilon}e^{-2\pi i\theta}\lambda_3\left(\left(\begin{matrix}
    \lambda_1 & \rho\\
    0 & \lambda_1 e^{2\pi i \omega}
\end{matrix}\right)+o(1)\right),\, \text{ as } \varepsilon\to\infty.
    \end{align}
Hence
\begin{align}
    M_{E,AB}^{HH^*}(\theta+i\varepsilon)=\left(\begin{matrix} 
    -\frac{\lambda_1}{\lambda_2}e^{-2\pi i\omega} &-\frac{\rho}{\lambda_2}e^{2\pi i \frac{1}{3}\omega-2\pi i\omega} &\lambda_1\lambda_2 &0\\
    0 &-\frac{\lambda_1}{\lambda_2}e^{2\pi i \frac{1}{3}\omega-2\pi i\omega} &0 &\lambda_1\lambda_2 e^{2\pi i\frac{1}{3}\omega}\\
    0 &0 &0 &0\\
    0 &0 &0 &0
    \end{matrix}\right)+o(1), \text{ as } \varepsilon\to\infty.
\end{align}
This implies, by the continuity of Lyapunov exponents in the cocycles \cite{AJS}, similarly to \eqref{eq:lb_L^2_AA_2} that 
\begin{align}
    L^2_{\varepsilon}(\omega,M_{E,AB}^{HH^*})=2\log |\lambda_1/\lambda_2|+o(1), \text{ as } \varepsilon\to\infty.
\end{align}
By the convexity of $L^2_{\varepsilon}(\omega,M_{E,AB}^{HH^*})$ in $\varepsilon$ in the interval $(\varepsilon_1,\infty)$, we have
\begin{align}\label{eq:lb_L^2_AB}
    L^2_{\varepsilon=0}(\omega,M_{E,AB}^{HH^*})\geq 2\log |\lambda_1/\lambda_2|\geq \log |\lambda|+O(1).
\end{align}
In analogy with\eqref{eq:ub_L^1_AA}, we now conclude
\begin{align}\label{eq:ub_L^1_AB}
    \|M_{E,AB}^{HH^*}(\cdot)\|_{\T,\infty}\leq \log \left(\frac{(|\lambda_1|+|\lambda_2|+|\lambda_3|+|\rho|)^3}{(|\lambda_3|-|\lambda_1|)^2}\right)+O(1)\leq \log |\lambda|+O(1), 
\end{align}
uniformly in $|E|\leq 10(\lambda_1^2+\lambda_2^2+\lambda_3^2+\rho^2)$.
Combining \eqref{eq:lb_L^2_AB} with \eqref{eq:ub_L^1_AB} yields the claimed result.
\end{proof}

\section{Coupled Harper operators}
\label{sec:coupledHarper}
\subsection{The model and generalities}
Consider the eigenvalue problem, with Diophantine $\omega$, 
\begin{equation}
     \label{eq:coupledHarper}
     \begin{split}
          \phi_{n+1}+\phi_{n-1}+\epsilon\psi_n + 2\lambda_1 \cos(2\pi(x+n\omega)) \phi_n &= E\phi_n \\
     \psi_{n+1}+\psi_{n-1}+\epsilon\phi_n + 2\lambda_2 \cos(2\pi(x+n\omega)) \psi_n &= E\psi_n
     \end{split}
\end{equation}
where $\epsilon\in\R$ and $\lambda_2\ge \lambda_1>0$. One can couple more scalar quasi-periodic operators in this fashion and also allow for more general potentials. For simplicity we restrict ourselves to the system~\eqref{eq:coupledHarper} since it already poses sufficiently many challenges. Setting $\Phi_n=\binom{\phi_n}{\psi_n}$, we can rewrite~\eqref{eq:coupledHarper} in the form
\[
(\calH_x \Phi)_n = \Phi_{n+1}+\Phi_{n-1} + V_n(x) \Phi_n, \quad V_n = \left( \begin{matrix}
    2\lambda_1 \cos(2\pi(x+n\omega))  & \epsilon \\ \epsilon & 2\lambda_2 \cos(2\pi(x+n\omega)) 
\end{matrix}\right)
\]
which is a special case of~\eqref{eq:BVsys}. We order the Lyapunov exponents as above, i.e.,  $L_1\ge L_2 \ge0\ge L_3=-L_2\ge L_4=-L_1$. By Herman's method, $L_1(E)\ge \max(\log\lambda_1,\log\lambda_2,0)$.  

\begin{theorem}
    \label{thm:puretype}
    For $\epsilon$ small the following hold:
    \begin{itemize}
        \item if $\lambda_1>1$, then~\eqref{eq:coupledHarper} exhibits Anderson localization for a.e.~$\omega, x\in\tor$.
        \item if $\lambda_2<1$, then $L_2=0$ on $\sigma(\calH_x)$ and spectrum $\sigma(\calH_x)$ equals the essential support of the absolutely continuous spectrum. 
    \end{itemize}
\end{theorem}
\begin{proof}
    By continuity of $L_j$ in the cocycle~\cite{AJS}*{Theorem 1.5}, we conclude that $L_1(E), L_2(E)>0$ for all $E$ provided $\lambda_1>1$ and $\epsilon\ge0$ is sufficiently small. Hence, Theorem~\ref{thm:main} guarantees  the first property. 

    For the second property, denote the accelerations of $L_j$ by $\kappa_j$. Then $\kappa_1$ and $\kappa_1+\kappa_2$ are upper semi-continuous in the cocycle. They both vanish for all energies if $\epsilon=0$. Since they are moreover quantized, see~\cite{AJS}*{Theorem 1.4}, it follows that they still vanish for small~$\epsilon$. If  $L_2(E)>0$, from \cite{AJS}*{Theorem~1.2} it would follow that the cocycle is $2$-dominated which means that is uniformly hyperbolic. By \cite{HP}*{Theorem 2.1} this is impossible if $E\in\sigma(\calH_x)$. The a.c.~statement follows from $L_2(E)=0$ via~\cite{KS}*{Theorem 7.2}. 
\end{proof}

For the a.c.\ statement we would like to know that $\mes(\sigma(\calH_x))>0$. This is indeed the case if $0<\lambda_1\ll \lambda_2<1$, see the proof of Proposition~\ref{prop:posmeas} below which uses the Cantor structure of the Harper spectrum~\cite{Pu}.

\subsection{Coexistence of pure point and absolutely continuous spectra}
We now exhibit co-existence of pure point and a.c.~spectrum provided 
$0<\lambda_1\ll 1\ll \lambda_2$. Throughout, the frequency~$\omega$ is assumed to be Diophantine and we will make $0<\epsilon$ as small as needed for various arguments to go through.   

\begin{lemma}
    \label{lem:calHspec}
  Let $\calH_x=\calH_{\epsilon,\lambda_1,\lambda_2}(x,\omega)$ be the bounded self-adjoint operator on $\ell^2(\Z;\C^2)$ defined by the left-hand side of~\eqref{eq:coupledHarper}.   The spectrum $\sigma(\calH_x)$ does not depend on $x\in\tor$ and satisfies 
  \[
  \dist(\calH_x,\sigma(H_{x,\lambda_1})\cup \sigma(H_{x,\lambda_2}) ) \le\epsilon
  \]
  where $H_{x,\lambda}$ is the Harper operator. Moreover, $ \sigma( \calH_x) \setminus [-3,3]\neq  \emptyset$. In fact, this part of the spectrum (the ``edges") have positive measure, as does the ``interior"  $\sigma( \calH_x) \cap [-3,3]\ne\emptyset$.
\end{lemma}
\begin{proof}
    The system~\eqref{eq:coupledHarper} consists of $A(x)=H_{x,\lambda_1}$ and $B(x)=H_{x,\lambda_2}$, each is a Harper's model (but with different coupling), coupled by $\epsilon\Gamma$, where $\Gamma$ is the hopping operator $\phi\to\psi$.   
By standard perturbation theory of bounded self-adjoint operators~\cite{Kato}, the operator $\calH_x$ defined by the left-hand side of~\eqref{eq:coupledHarper} satisfies
\begin{equation}\label{eq:calHx}
    \begin{split}
       \sigma(\calH_x) &\subset \big(\sigma(H_{x,\lambda_1})\cup \sigma(H_{x,\lambda_2}) \big) + [-\epsilon,\epsilon] \\ 
       \sigma(H_{x,\lambda_1})\cup \sigma(H_{x,\lambda_2}) &\subset \sigma(\calH_x) + [-\epsilon,\epsilon] 
    \end{split}
\end{equation}
as claimed. 
 By unique ergodicity of irrational rotations on~$\tor$,  none of these spectra depend on~$x$ (and so we can drop~$x$ from the notation of spectra). Moreover, it is well-known~\cite{AvMS,JK,AK} that $\mes(H_{x,\lambda})=4|1-|\lambda||$ for any irrational $\omega$. 
 Hence for $\lambda_2>3$,
 \begin{align}
     \mathrm{mes}(\sigma(H_{x,\lambda_1})\cup \sigma(H_{x,\lambda_2}))\geq \mathrm{mes}(\sigma(H_{x,\lambda_2}))=4\lambda_2-4>6+2\varepsilon=\mathrm{mes}([-3,3]+[-\varepsilon,\varepsilon]).
 \end{align}
 This implies $\sigma( \calH_x) \setminus [-3,3]\neq  \emptyset$.
 The positive measure claims rely on Green's function estimates and will be proved later in Proposition \ref{prop:posmeas}. To deal with the interior part of $\sigma(\calH_x)$ we will use duality, see Corollary \ref{cor:dualAL}.
\end{proof}

 Next, we address the Anderson localization of $\calH_x$ on a positive measure set located at the edge of the spectrum.
 To do this, we introduce the following notion of regular Green's function.
 
     We say that $(\calH^{(N)}_x-E)^{-1}$ is regular, if for some $0<\nu<1$ and $\gamma>0$
    \begin{itemize}
        \item $\| (\calH^{(N)}_x-E)^{-1}\|\le e^{N^\nu}$, 
        \item $|(\calH^{(N)}_x-E)^{-1}(k,\ell)|\le e^{-\gamma|k-\ell|}$, for all $k,\ell\in [0,N]$ where $|k-\ell|\ge N/10$
    \end{itemize}

    \begin{lemma}
        \label{lem:calHreg}
        Let $E\in \sigma(\calH_x)\setminus\sigma(A(x))$.  Then  for all sufficiently small $\epsilon>0$, there exist $\nu,\gamma$ and  $\tau>0$ such that 
        \[
        \mes(\{x\in\tor\::\: (\calH^{(N)}_x-E)^{-1} \text{\ \ is not regular} \})\le e^{-N^\tau}
        \]
        for all $N$ large. Moreover, the set on the left-hand side is contained in  at most $O(N)$ intervals each of length at most~$e^{-N^\tau}$. The constants here depend only on $\omega, \lambda_1,\lambda_2$, and $\dist(E,\sigma(A(x))$.
    \end{lemma}
\begin{proof}
    By induction in $N$. To deal with the initial scale, we write the   operator $\calH_x$ defined by~\eqref{eq:coupledHarper} with Dirichlet boundary conditions on $[0,N]$  in block  form
    \[
    \calH^{(N)}_x = \left [ \begin{matrix}
        A_N& \epsilon \Gamma \\ \epsilon \Gamma^* &  B_N
    \end{matrix}\right]
    \]
    where $A_N, B_N$ are the Harper operators with Dirichlet boundary conditions. By choice of~$E$, the Green's function $(A_N(x)-E)^{-1}$ will be bounded and exponentially decaying for all $x$ and $N$ large enough. 
      Using the Feshbach formula, see Lemma~4.8 in~\cite{BGS}, one
      reduces the full Green's function $(\calH^{(N)}_x-E)^{-1}$ to the Schur complement 
    \[
    D_N(x,E):= B_N(x)-E - \epsilon^2 \Gamma (A_N(x)-E)^{-1} \Gamma^* 
    \]
    At an initial scale $N_0\gg1$, we control $D_{N_0}(x,E)^{-1}$ perturbatively by taking $\epsilon$ small and $x$ outside a small bad set governed by the LDT for $B_{N}(x)=H^{(N)}_{x,\lambda_2}$:
    \[
    \mes(\{ x\in\tor\::\: (H^{(N)}_{x,\lambda_2}-E)^{-1} \text{\ \ is not regular\  } \})\le e^{-N^\tau}
    \]
    where $\tau=\tau(\omega,\lambda_2)>0$ and $N$ large. 
    This proves the lemma for $N\in [N_0, N_0^K]$ where $N_0$ is large and some constant $K\ge2\tau^{-1}$. At these initial scales, the complexity bound of $O(N)$ follows from the fact that the connected components of the bad set $\calB_{N}(E)$ contain the zeros of~$\det(B_{N}(x)-E)$. 
    
    For larger scales, we run a multi-scale argument via the resolvent identity and Cartan in the spirit of~\cite{BGS}.
Let $N_1\simeq N_0^{C_0}$ where $C_0>1$ is a large constant that will be determined later. For any $x_0\in\tor$ we call  $n\in [0,N_1]$ {\em good} if the Green's function of $\calH_{x_0}-E$ restricted to $[0,N_1]\cap [n-N_0,n+N_0]=:J_n$ is regular. The number of bad $n\in [0,N_1]$ does not exceed $O(N_0)$ by the Diophantine condition and the inductive assumption. We write $\Lambda=[0,N_1]=\Lambda_* \cup \Lambda_{**}$ where 
\[
\Lambda_* = \bigcup_{n\text{\ bad}} J_n, \quad \# \Lambda_* \lesssim N_0^2
\]
Denote by $G_{\Lambda_{**}}(x,E)$ the Green's function of $\calH_x-E$ restricted to $\Lambda_{**}$ with Dirichlet boundary conditions. By iterating the resolvent identity we conclude that 
\[
\| G_{\Lambda_{**}}(x_0,E)\|\les N_0  e^{N_0^\nu},
\]
see~\cite{BGS}*{Lemma 2.2}. By a standard perturbative Neumann series argument, we further conclude that this bound is locally stable, i.e., 
\[
\| G_{\Lambda_{**}}(z,E)\|\le  e^{N_0}, \qquad \forall\; |z-x_0|<e^{-N_0},\; z\in\C
\]
Locally on $I_0=(x_0-e^{-N_0}, x_0+e^{-N_0})\subset\tor$ we write 
\[
\calH_x^{\Lambda}-E = \left [ \begin{matrix}
     \calH_x^{\Lambda_{*}} -E   &  \Gamma_0 \\   \Gamma_0^* &  \calH_x^{\Lambda_{**}} -E
    \end{matrix}\right]
\]
    where the operators on the diagonal are the restrictions to the respective sets  with Dirichlet conditions, while $\Gamma_0$ are the hopping terms. We reduce the Green's function of the full operator to the Schur complement of this block matrix, viz. 
    \[
    \calS_x := \calH_x^{\Lambda_{*}} -E  - \Gamma_0 (\calH_x^{\Lambda_{**}} -E)^{-1} \Gamma_0^*
    \]
Next, we cover $\Lambda$ by intervals $\Lambda_j$ of size $M_0=N_0^K$ to conclude that  all $G_{\Lambda_j}(x,E)$ are regular provided $x\in I_0\setminus \calB$, where  $\mes(\calB) \les e^{-M_0^\tau} = e^{-N_0^{\tau K}}\ll |I_0|$. 
By~\cite{BGS}*{Lemma 2.2, Lemma 4.8} we deduce that $\|\calS_x^{-1}\|\le e^{M_0}$ for those~$x$ as well as by self-adjointness of~$\calS_x$, 
\[
\log \det(\calS_x)\ge -M_0^2 = N_0^{2K}
\]
    Finally,  on the complex disk $\disk(x_0, e^{-N_0})$, we have
    \[
    \log\det (\calS_z)\le |\Lambda_*| N_0 \les N_0^3
    \]
    Taking $C_0$ large enough, the Riesz mass of the sub-harmonic function 
    \[ 
    u(\zeta )=\log\det (\calS_{x_0+\zeta e^{-N_0}})
    \]
    on~$\disk(0,1)$ is at most $N_0^{2K}=N_1^{{2K}/{C_0}}$. By Cartan's theorem, \[ u(\zeta)> - N_1^{{2K}/{C_0}} N_1^{\tau}> -N_1^{\frac{\nu}{2}}\] off a set of measure at most $e^{-2N_1^\tau}$ in~$\disk(0,1)$. Finally, we rescale and sum over the $x$-localization which costs a factor of~$e^{N_0}$. By Cramer's rule, and \cite{BGS}*{Lemma 4.8}, we obtain the first condition of regularity. For the exponential off-diagonal decay, we iterate the resolvent expansion using Green's functions of the smaller $N_0$ scale, allowing for $N_0^2= N_1^{2/C_0}$ many resonant intervals of that scale within~$\Lambda$. See~\cite{BGS}*{Lemma 2.4} for such a procedure in the much more complicated two-dimensional setting. The statement about $O(N)$ connected components follows from the fact that each such component must contain a zero of $\det(\calH_x-E)$. 
\end{proof}

The first result we prove about~\eqref{eq:coupledHarper} is Anderson localization for energies outside of~$[-3,3]$.

\begin{theorem}
    \label{thm:couplHarper} 
     For a.e.~$\omega$, the operator $\calH_0$ exhibits Anderson localization on   $\sigma(\calH)\setminus[-3,3]\neq\emptyset$.
     Moreover, the Lyapunov exponents do not vanish on that part of the spectrum. 
\end{theorem}
\begin{proof}
    This follows from the  double-resonance exclusion argument in~\cite{BG} via Lemma~\ref{lem:calHreg}, and the usual semi-algebraic techniques. For the Lyapunov exponents, we do not go through localization but rather invoke the continuity of the Lyapunov exponents in the cocycle, see~\cite{AJS}. In fact, for $\epsilon=0$ this is clearly correct, due to the properties of the Harper operator. Since we can rewrite~\eqref{eq:coupledHarper} as a cocycle over an irrational base, we can invoke the continuity results of~\cite{AJS} to conclude nonvanishing of the Lyapunov exponents for small~$\epsilon$. Clearly, the non-vanishing of the Lyapunov exponents leads to another proof of localization, via Theorem~\ref{thm:main}. However, the technique from~\cite{BGS} which we have followed above is independent of Theorem~\ref{thm:main} and more robust. In fact, it also applies to coupled PDEs.  
 \end{proof}

The energies exhibiting localization in Theorem~\ref{thm:couplHarper} form a set of positive measure. 

\begin{proposition}\label{prop:posmeas}
    For all Diophantine $\omega$ one has  $\mes(\sigma(\calH)\setminus[-3,3])>0$. 
\end{proposition}
    \begin{proof}
 The proof in \cite{B}*{p.~88-83}, see also~\cite{B0}, applies directly. Indeed, by the remark on p.~81 of~\cite{B}, Bourgain's key Lemma~12.15 does not require positive Lyapunov exponents, but rather a large deviation estimate for Green's functions. This is given by Lemma~\ref{lem:calHreg} above. The remainder of the proof of the positive measure statement, viz.~\cite{B}*{Proposition 12.14}, does not rely on the scalar nature of the Schr\"odinger operator but rather its self-adjointness and elementary semi-algebaic set considerations, cf.~\cite{B}*{eq.\ (12.28)}. In our case, these remain valid due to the fact that the underlying potenial is a trigonometric polynomial (in fact, a cosine). Hence, Bourgain's Proposition~12.14 remains valid for the system~\eqref{eq:coupledHarper} provided we are at the edges of the spectrum where Lemma~\ref{lem:calHreg} applies.  
    \end{proof}

Next, we show the existence of absolutely continuous spectrum.
We achieve this via proving the existence of a positive measure set of energies in $\sigma(\calH_x)$ for which at least one of the Lyapunov exponents is zero.
\begin{proposition}
    \label{prop:ac2Harper'}
    There exists a positive measure set of energies in the spectrum of $\calH_x$ defined by~\eqref{eq:coupledHarper} with exactly one vanishing Lyapunov exponent. Thus, \eqref{eq:coupledHarper} exhibits some a.c.\ spectrum of multiplicity~$2$. Moreover, for a.e.~$x\in\tor$ the following holds: for a.e.~$E$ in the a.c.\ spectrum of $\calH_x$ the generalized eigenfunctions in the a.c.\ spectral subspace of~$\calH_x$ are truly extended, i.e., they cannot decay exponentially at either end. 
\end{proposition}
\begin{proof}
Note the a.c. claim then follows from this by \cite{KS} which is valid from strip models.
We start with the following.
\begin{lemma}\label{lem:U_L=0}
There exists $U\subset \sigma(\calH_x)$ such that $\mathrm{mes}(U)>0$ and $L_2(E)=0$ for $E\in U$.    
\end{lemma}
\begin{proof}
    We begin by showing there exists some energy $E\in \sigma(\calH_x)\setminus \sigma(H_{x,\lambda_2})$.
    For $\varepsilon=0$, this is an immediate consequence of the Cantor property of the spectrum of Harper~\cites{Pu, AJ} and that $\mathrm{mes}(\sigma(H_{x,\lambda}))=4|1-|\lambda||$. In fact, $\sigma(H_{x,\lambda_2})$ has dense gaps in $[-2,2]$, we pick one such gap $J_0=(E_0-\tau_0,E_0+\tau_0)$. Then for $|\lambda_1|>0$ small enough, 
   \[8\lambda_1=\mathrm{mes}([-2-2\lambda_1,2+2\lambda_1]\setminus \sigma(H_{x,\lambda_1})<\tau_0/2.\]
   This implies the existence of $E\in \sigma(H_{x,\lambda_1})\cap (E_0-\tau_0/4,E_0+\tau_0/4)$.
    Hence for any $\epsilon\in (0,\tau_0/4)$, we can obtain perturbatively the existence of some $E\in\sigma(\calH_x)\cap (E_0-\tau_0/2,E_0+\tau_0/2)$.
    
    Let $U=\sigma(\calH_x)\cap (E_0-\tau_0/2, E_0+\tau_0/2)$.
    Next we show for $\epsilon>0$ small enough, for any $E\in U$, we have $L_2(E)=0$. 
    This proof is similar to that of the second part of Theorem \ref{thm:puretype}.
    In fact, denoting the accelerations of $L_j$ by $\kappa_j$. For $\epsilon=0$ and $E\in \sigma(\calH_x)\cap (E_0-\tau_0,E_0+\tau_0)$, one has $\kappa_1+\kappa_2=0$. Hence $\kappa_1+\kappa_2=0$ and $L_2(E)=0$ for small $\epsilon>0$ on $U$. 
    \end{proof}

    It remains to prove $\mathrm{mes}(U)>0$. 
    To do this, we pass to the dual system of~\eqref{eq:coupledHarper}.  Thus, let $\hat{\phi}(\theta)=\sum_n \phi_n e(n\theta)$ and similarly for $\hat{\psi}$. If $\phi_n,\psi_n$ solve~\eqref{eq:coupledHarper}, then 
    \[
    u_n = e(x+n\omega) \hat{\phi}(\theta+n\omega),\qquad v_n = e(x+n\omega) \hat{\psi}(\theta+n\omega)
    \]
    solve
    \begin{equation}
     \label{eq:coupledHarper*}
     \begin{split}
          \lambda_1(H_{\theta,\lambda_1^{-1}}u)_n+\epsilon v_n &= E u_n \\
     \lambda_2(H_{\theta,\lambda_2^{-1}}v)_n+\epsilon u_n&= Ev_n
     \end{split}
\end{equation}

\begin{lemma}
    \label{lem:Aubry}
    Denote the operator defined by the left-hand side of~\eqref{eq:coupledHarper*} by~$\widehat{\calH}_\theta$. It is unitarily equivalent to~$\calH_x$, if both are viewed as bounded self-adjoint operators on $L^2(\Z\times\tor;\C^2)$. These operators have spectra equal to $\sigma(\widehat{\calH}_\theta)=\sigma(\calH_x)$, which are constant in~$x,\theta\in\tor$. 
\end{lemma}
\begin{proof}
    This is Aubry duality, see for example~\cite{HP}*{Section~1.5.1}. 
\end{proof}
Recall that the Aubry duality for the Harper's model implies $\sigma(H_{x,\lambda})=\lambda \sigma(H_{x,\lambda^{-1}})$. Hence combined with Lemma \ref{lem:Aubry}, we have
\[\emptyset\neq U\subset \sigma(\calH_x)\cap \sigma(H_{x,\lambda_2})^c=\sigma(\widehat{\calH}_{\theta})\cap \lambda_2\sigma(H_{x,\lambda_2^{-1}})^c.\]
Furthermore, for $\epsilon>0$ small, the Lyapunov exponents $\hat{L}_1(E),\hat{L}_2(E)>0$ on $U$ for the dual operator $\widehat{\calH}_{\theta}$ (true for $\epsilon=0$ and stay positive by continuity in $\epsilon$).
One can then conclude 
\begin{align}\label{eq:mes_U>0}
    \mathrm{mes}(U)>0,
\end{align} 
by following the same arguments as in  Proposition \ref{prop:posmeas}.
One also has the analogue of Theorem \ref{thm:couplHarper} as follows.
 \begin{corollary}
     \label{cor:dualAL}
For a.e.~$\omega$, the operator $\widehat{\calH}_0$ exhibits Anderson localization on $U$. 
 \end{corollary}

 For the a.c. spectrum, we will not go through the localization established in the previous corollary, since it is not immediately clear how to proceed. Rather we derive it directly from combining Lemma \ref{lem:U_L=0}, \eqref{eq:mes_U>0} with the Kotani-Simon theory \cite{KS} for the strip model.

 The final claim about the absence of ``hybrid" states, i.e., the possibilty of exhibiting exponential decay as either $n\to+\infty$ or $n\to-\infty$, follows from Fubini and the two-sided version of Oseledets' theorem, see~\cite{V}*{Theorem 4.2}.
 \end{proof}


\begin{thebibliography}{99}

\bibitem[Av]{Global} Avila, A., 2015. {\em Global theory of one-frequency Schr\"odinger operators.} Acta Math., 215 (1), pp.\ 1--54.

\bibitem[Av2]{Av2}Avila, A., 2023. {\em KAM, Lyapunov exponents, and the Spectral Dichotomy for typical one-frequency Schrodinger operators.} arXiv preprint arXiv:2307.11071.

\bibitem[AJ]{AJ} Avila, A. and Jitomirskaya, S., 2009. {\em Almost localization and almost reducibility.} J.\ Eur.\ Math.\ Soc., 12(1), pp.\ 93--131.

\bibitem[AJS]{AJS} Avila, A., Jitomirskaya, S. and Sadel, C., 2014. {\em Complex one-frequency cocycles.} J.\ Eur.\ Math.\ Soc., 16(9), pp.~1915--1935.

\bibitem[AK]{AK}Avila, A. and Krikorian, R., 2006. {\em Reducibility or nonuniform hyperbolicity for quasiperiodic Schrödinger cocycles.} Annals of Mathematics, pp.911-940.

\bibitem[AvMS]{AvMS} Avron, J., van Mouche, P.\ H.\ M.,  Simon, B. {\em On the measure of the spectrum for the almost Mathieu operator.} Comm.\ Math.\ Phys.\ 132 (1990), no~1, 103--118.

\bibitem[BPR]{BPR} 
Basu, S.,  Pollack, R.,  Roy, M.-F. {\em   Algorithms in real algebraic geometry.} 
Algorithms Comput.\ Math., 10
Springer-Verlag, Berlin, 2006

\bibitem[BGW]{BGW}
Becker, S., Ge, L. and Wittsten, J., 2022. {\em Hofstadter butterflies and metal/insulator transitions for moir\'e heterostructures.} arXiv preprint arXiv:2206.11891.

\bibitem[BHJ]{BHJ}Becker, S., Han, R. and Jitomirskaya, S., 2019. {\em Cantor spectrum of graphene in magnetic fields.} Inventiones mathematicae, 218, pp.979-1041.


\bibitem[BinNov]{BN} Binyamini, G.,  Novikov, D. {\em 
Complex cellular structures.}
Ann.\ of Math.~(2)190(2019), no.~1, 145--248.

\bibitem[Bou1]{B0} Bourgain, J.  {\em
On the spectrum of lattice Schr\"odinger operators with deterministic potential.}
J.\ Anal.\ Math.~87 (2002), 37--75.

\bibitem[Bou2]{B1} Bourgain, J., 2007, {\em Positive Lyapounov exponents for most energies. In Geometric Aspects of Functional Analysis: Israel Seminar 1996–2000} (pp. 37-66). Berlin, Heidelberg: Springer Berlin Heidelberg.

\bibitem[B]{B} Bourgain, J.  {\em Green's function estimates for lattice Schr\"odinger operators and applications}, 
Ann.\ of Math.\ Stud., 158
Princeton University Press, Princeton, NJ, 2005. 

\bibitem[BG]{BG} Bourgain, J. and Goldstein, M., 2000. {\em On nonperturbative localization with quasi-periodic potential.} Ann.\ of Math., 152 (3),
 pp.~835--879.

 \bibitem[BGS]{BGS} Bourgain, J.,  Goldstein, M.,  Schlag, W. {\em 
Anderson localization for Schr\"odinger operators on $\Z^2$ with quasi-periodic potential.} 
Acta Math.~188 (2002), no.~1, 41--86.

 \bibitem[BJ]{BJ} Bourgain, J. and Jitomirskaya, S., 2000. {\em Anderson localization for the band model.} Geometric aspects of functional analysis, 1745, pp.~67--79.

\bibitem[CS]{CS}Chapman, J. and Stolz, G., 2015, February. {\em Localization for random block operators related to the XY spin chain.} In Annales Henri Poincaré (Vol. 16, No. 2, pp. 405-435). 

\bibitem[CD]{CD}Chulaevsky, V. and Delyon, F., 1989. {\em Purely absolutely continuous spectrum for almost Mathieu operators.} Journal of statistical physics, 55, pp.1279-1284.


\bibitem[DGSV]{DGSV} Damanik, D.,  Goldstein, M.,  Schlag, W.,  Voda, M. {\em 
Homogeneity of the spectrum for quasi-periodic Schr\"odinger operators.} J.\ Eur.\ Math.\ Soc.\ (JEMS) 20 (2018), no.~12, 3073--3111.

\bibitem[DLLY]{DLLY}Damanik, D., Lemm, M., Lukic, M. and Yessen, W., 2014. {\em New Anomalous Lieb-Robinson Bounds in Quasiperiodic XY Chains}. Physical review letters, 113(12), p.127202.

\bibitem[DLY]{DLY}Damanik, D., Lukic, M. and Yessen, W., 2015. {\em Quantum dynamics of periodic and limit-periodic Jacobi and block Jacobi matrices with applications to some quantum many body problems.} Communications in Mathematical Physics, 337(3), pp.1535-1561.

\bibitem[DK1]{DK1} Duarte, P. and Klein, S., 2014. {\em Positive Lyapunov exponents for higher dimensional quasiperiodic cocycles.} Communications in Mathematical Physics, 332, pp.189-219.

\bibitem[DK2]{DK} Duarte, P. and Klein, S., 2016. {\em Lyapunov exponents of linear cocycles.} Atlantis Studies in Dyn.\ Systems, 3.

\bibitem[E]{E} Eliasson, L.H., 1992. {\em Floquet solutions for the 1-dimensional quasi-periodic Schr\"odinger equation.} Communications in mathematical physics, 146, pp.447-482.

\bibitem[Fi]{Fi}Fillman, J., 2017. {\em Ballistic transport for limit-periodic Jacobi matrices with applications to quantum many-body problems.} Communications in Mathematical Physics, 350, pp.1275-1297.

\bibitem[GS1]{GS1} Goldstein, M. and Schlag, W., 2001. {\em H\"older continuity of the integrated density of states for quasi-periodic Schr\"odinger equations and averages of shifts of subharmonic functions.} Ann.\ of Math., pp.~155--203.

\bibitem[GS2]{GS2} Goldstein, M. and Schlag, W., 2008. {\em Fine properties of the integrated density of states and a quantitative separation property of the Dirichlet eigenvalues.} Geom.\ Funct.\ Anal., 18 (3), pp.~755--869.

\bibitem[GS3]{GS3} Goldstein, M. and Schlag, W., 2011. {\em On resonances and the formation of gaps in the spectrum of quasi-periodic Schr\"odinger equations.} Ann.\ of Math., pp.~337--475.

\bibitem[GSV]{GSV} Goldstein, M., Schlag, W. and Voda, M., 2016. {\em On localization and the spectrum of multi-frequency quasi-periodic operators.} arXiv preprint arXiv:1610.00380.

\bibitem[HSS]{HSS} Hamza, E., Sims, R. and Stolz, G., 2012. {\em Dynamical localization in disordered quantum spin systems.} Communications in Mathematical Physics, 315(1), pp.215-239.

\bibitem[HS2]{HS2} Han, R. and Schlag, W., 2022. {\em Avila's acceleration via zeros of determinants, and applications to Schr\" odinger cocycles.} arXiv preprint arXiv:2212.05988.

\bibitem[HS3]{HS3} Han, R. and Schlag, W., 2023. {\em Non-perturbative localization on the strip and Avila's almost reducibility conjecture.} arXiv preprint arXiv:2306.15122.

\bibitem[HP]{HP} Haro, A. and Puig, J., 2013. {\em A Thouless formula and Aubry duality for long-range Schr\"odinger skew-products.} Nonlinearity, 26(5), p.~1163.

\bibitem[JK]{JK} Jitomirskaya, S.Y. and Krasovsky, I.V., 2002. {\em Continuity of the measure of the spectrum for discrete quasiperiodic operators.} Mathematical Research Letters. Jul;9(4):413-21.

\bibitem[Ka]{Ka}Kachkovskiy, I., 2016. {\em On transport properties of isotropic quasiperiodic XY spin chains.} Communications in Mathematical Physics, 345, pp.659-673.

\bibitem[Kato]{Kato} Kato, T. {\em Perturbation Theory for Linear Operators}, Springer–Verlag, Berlin, 1966. 

\bibitem[Kl]{Kl} Klein, S. {\em 
Anderson localization for one-frequency quasi-periodic block Jacobi operators.}  J.\ Funct.\ Anal.~273 (2017), no~3, 1140--1164.

\bibitem[KS]{KS} Kotani, S., Simon, B. {\em 
Stochastic Schr\"odinger operators and Jacobi matrices on the strip.} Comm.\ Math.\ Phys.\ 119 (1988), no.~3, 403--429.

\bibitem[LHCL]{LHCL}Lai, Y.H., Ho, J.H., Chang, C.P. and Lin, M.F., 2008. {\em Magnetoelectronic properties of bilayer Bernal graphene.} Physical Review B, 77(8), p.085426.

\bibitem[LSM]{LSM}Lieb, E., Schultz, T. and Mattis, D., 1961. {\em Two soluble models of an antiferromagnetic chain.} Annals of Physics, 16(3), pp.407-466.

\bibitem[Pu]{Pu} Puig, J. {\em Cantor Spectrum for the Almost Mathieu Operator}, Commun.\ Math.\ Phys.\ 244, 297--309 (2004). 

\bibitem[RSRN]{RSRN} Rozhkov, A.V., Sboychakov, A.O., Rakhmanov, A.L. and Nori, F., 2016. {\em Electronic properties of graphene-based bilayer systems.} Physics Reports, 648, pp.1-104.


\bibitem[Sch1]{S1} Schlag, W., 2013. {\em Regularity and convergence rates for the Lyapunov exponents of linear cocycles.} Journal of Modern Dynamics,  7(4): 619--637.

\bibitem[Sch2]{S2} Schlag, W. {\em 
An introduction to multiscale techniques in the theory of Anderson localization, Part I.}
Nonlinear Anal.~220(2022), Paper No.~112869, 55 pp.

\bibitem[TM]{TM} Timmel, A. and Mele, E.J., 2020. {\em Dirac-Harper theory for one-dimensional moiré superlattices.} Physical Review Letters, 125(16), p.166803.

\bibitem[Via]{V} Viana, M. {\em Lectures on Lyapunov exponents.} Cambridge studies in advanced mathematics 145, 2014. 
 
\end{thebibliography}
\end{document}